\documentclass[onecolumn,aps,pra,nofootinbib,superscriptaddress,tightenlines,notitlepage,12pt]{revtex4-1}

\makeatletter
\renewcommand{\p@subsection}{}
\renewcommand{\p@subsubsection}{}
\makeatother

\makeatletter
\def\l@subsubsection#1#2{}
\makeatother

\usepackage{amsmath,amsfonts,amsthm}
\usepackage{amssymb}
\usepackage{subcaption}
\usepackage{caption}
\usepackage{enumerate}
\usepackage{float}
\usepackage{mathtools}
\usepackage[left=0.75in,right=0.75in,top=0.75in,bottom=0.75in]{geometry}
\usepackage{color}
\usepackage{graphicx}
\usepackage[breaklinks,colorlinks = true,linkcolor = blue,urlcolor  = blue,citecolor = red,anchorcolor = blue]{hyperref}
\usepackage[matrix,frame,arrow]{xypic}
\usepackage[braket]{qcircuit}
\usepackage{slashbox}
\usepackage{comment}
\usepackage{tikz}

\usepackage{verbatim}
\usepackage[framemethod=tikz]{mdframed}

\usepackage{algorithm}
\usepackage{algpseudocode}

\mathtoolsset{showonlyrefs}

\newcommand*{\cC}{\mathcal{C}}
\newcommand{\labelgroup}[5]{\POS"#1,#5"."#2,#5"."#1,#5"."#2,#5", \POS"#1,#5"."#2,#5"."#1,#5"."#2,#5"*!C!<1em,#3>=<0em>{#4}}

\theoremstyle{plain}
\newtheorem{result}{Result}
\newtheorem{lemma}{Lemma}

\newtheorem{theorem}[lemma]{Theorem}

\theoremstyle{definition}
\newtheorem{cond}{Condition}
\newtheorem{definition}[lemma]{Definition}

\newcommand{\id}{1}

\newcommand{\poly}{\mathrm{poly}}

\newcommand{\calA}{{\cal A }}
\newcommand{\calB}{{\cal B}}

\newcommand{\calL}{{\cal L }}
\newcommand{\calN}{{\cal N }}

\newcommand{\calF}{{\cal F }}
\newcommand{\calG}{{\cal G }}
\newcommand{\calV}{{\cal V }}
\newcommand{\calE}{{\cal E }}
\newcommand{\calC}{{\cal C }}
\newcommand{\calQ}{{\cal Q }}
\newcommand{\calS}{{\cal S }}

\newcommand{\calP}{{\cal P }}

\newcommand{\calW}{{\cal W }}
\newcommand{\be}{\begin{equation}}
\newcommand{\ee}{\end{equation}}
\newcommand{\la}{\langle}
\newcommand{\ra}{\rangle}

\newcommand{\dec}{\mathsf{Dec}}
\newcommand{\rec}{\mathsf{Rec}}
\newcommand{\repair}{\mathsf{Rep}}
\newcommand{\parity}{\mathsf{Parity}}
\newcommand{\lift}{\mathsf{Lift}}
\newcommand{\cor}{\mathsf{Cor}}

\newcommand{\pauli}{\mathrm{Pauli}}
\newcommand{\neigh}{\mathrm{neigh}}
\newcommand{\syn}{\mathrm{syn}}

\newcommand*{\CNOT}{\textrm{CNOT}}
\newcommand*{\ConeD}{C^{\textrm{1DMSP}}}

\newcommand*{\malpha}{\alpha} 
\newcommand*{\mbeta}{\beta}   

\newcommand*{\flightcone}{L^{\rightarrow}}
\newcommand*{\blightcone}{L^{\leftarrow}}
\newcommand*{\gluedgraph}{T_{\textrm{gl}}}

\mdfdefinestyle{mystyle}{leftmargin=0cm, innerleftmargin=-0.3cm}

\begin{document}

\title{\Large Quantum advantage with noisy shallow circuits in 3D}

\author{Sergey Bravyi}
\affiliation{IBM T. J. Watson Research Center, Yorktown Heights, USA}
\author{David Gosset}
\affiliation{Department of Combinatorics and Optimization, and\\
Institute for Quantum Computing, University of Waterloo, Waterloo, Canada}
\author{Robert K\"onig}
\affiliation{Institute for Advanced Study \& Zentrum Mathematik, Technical University of Munich, Munich, Germany}
\author{Marco Tomamichel}
\affiliation{Centre for Quantum Software and Information \& School of Software, University of Technology Sydney, Sydney, Australia}


\begin{abstract}
Prior work has shown that there exists a relation problem which can be solved with certainty by a constant-depth quantum circuit composed of geometrically local gates in two dimensions, but cannot be solved with high probability by any classical constant depth circuit composed of bounded fan-in gates. Here we provide two extensions of this result. Firstly, we show that a separation in computational power persists even when the constant-depth quantum circuit is restricted to geometrically local gates in one dimension. The  corresponding quantum algorithm is the simplest we know of which achieves a quantum advantage of this type. It may also be more practical for future implementations. Our second, main result, is that a separation persists even if the shallow quantum circuit is corrupted by noise. We construct a relation problem which can be solved with near certainty using a \textit{noisy}  constant-depth quantum circuit composed of geometrically local gates in three dimensions, provided the noise rate is below a certain constant threshold value. On the other hand, the problem cannot be solved with high probability by a noise-free classical circuit of constant depth. A key component of the proof is a quantum error-correcting code which admits constant-depth logical Clifford gates and single-shot logical state preparation. We show that the surface code meets these criteria. To this end, we provide a protocol for single-shot logical state preparation in the surface code which may be of independent interest.
\end{abstract}
\maketitle

\tableofcontents

\vspace{2cm}

\section{Introduction}
The appeal of quantum computing lies in the hope that  quantum devices may surpass
their classical counterparts in certain information processing tasks.
Indeed, a universal quantum computer could   efficiently solve certain computational problems such as factoring, for which no efficient classical algorithms are known to date. Yet,  even an experimental realization of such
universal quantum machines 
-- while impressive and potentially useful in applications -- would not conclusively establish  a computational quantum advantage
in the complexity-theoretic sense.
Instead, an efficient quantum algorithm must be accompanied with a  proof of the classical hardness of the considered problem.
For almost any problem of interest, such a proof would itself constitute a major complexity-theoretic advance.

To solidify the theoretical underpinnings of quantum computation, recent work has focused on computational problems where quantum advantage can be established, either conditionally or information-theoretically. Results of the former category rely on certain complexity-theoretic conjectures such as the non-collapse of the polynomial hierarchy as well as specific hardness assumptions for a given problem. For example, so-called IQP circuits and related proposals~\cite{bremner2016average,bremner2016achieving,farhi2016quantum,bermejo2017architectures} provide evidence that classically sampling from the output distribution of certain shallow quantum circuits may be intractable -- a key feature first identified 
by Terhal and DiVincenzo~\cite{terhal2002adaptive} and later strengthened by Aaronson's characterization of postBQP \cite{aaronson2005quantum}. Some of these works also provide experimental proposals for using a near-term quantum computer to perform a computational task that cannot be performed by any existing classical computer~\cite{boixo2018characterizing}. A rich debate concerning the feasibility of such proposals has prompted improvements to the performance of classical simulation algorithms for quantum computers \cite{pednault2017breaking,boixo2017simulation,li2018quantum, chen2018classical,bravyi2018simulation}.

While these results seek to separate efficient (i.e., polynomial-time) quantum computation from efficient classical computation, complementary unconditional results have been obtained for a more narrow question.
 It has been shown~\cite{bragokoe18} that constant-depth quantum circuits provide a provable computational advantage over constant-depth classical circuits, where both types of circuits are assumed to have bounded fan-in gates. Ref.~\cite{bragokoe18} introduced  a computational problem such that
\begin{enumerate}[(i)]
\item\label{it:noisefreecircuits}
the problem can be solved with certainty by a constant-depth quantum circuit
composed of geometrically local gates on a 2D grid of qubits, while
\item
any classical probabilistic circuit which solves the problem 
with success probability at least~$7/8$   must have depth growing logarithmically with the input size.
\end{enumerate}
This separation also holds in the average-case setting when
the classical circuit only needs to solve a few instances of the problem
that are drawn randomly from a suitable distribution~\cite[Supplementary Material]{bragokoe18}. 
Similar proofs of quantum advantage with associated average-case hardness results for classical circuits  have been obtained more recently in~\cite{coudronetal18,legall18}, see also~\cite{Watts2019advantage}. In this work we extend these results in two distinct ways.  

First, since the quantum algorithm described in Ref.~\cite{bragokoe18} is geometrically local in two dimensions, it is natural to ask whether a provable quantum advantage can also be achieved in a one-dimensional geometry. 
We answer this question in the affirmative.

Following Ref.~\cite{bragokoe18}, below we consider relation problems.
Recall that a relation  $R$ is defined as a set of valid input-output pairs
$(z_{\mathrm{in}},z_{\mathrm{out}})$, where $z_{\mathrm{in}}$ and $z_{\mathrm{out}}$
are bit strings of appropriate length. We shall describe a relation
by a function $R(z_{\mathrm{in}},z_{\mathrm{out}})$ that takes values $0$ or $1$.
A classical or quantum circuit  is said to 
solve a relation problem $R$ for some input $z_{\mathrm{in}}$
if it outputs a string $z_{\mathrm{out}}$ such that $R(z_{\mathrm{in}},z_{\mathrm{out}})=1$. A relation problem is said to have $l$ input-output bits if $|z_{in}|+|z_{out}|=l$.

\begin{result}[\textbf{Quantum advantage with 1D shallow circuits\,---\,informal}]
For each $n$ there exists a relation problem $R$ with roughly $n$ input-output bits and a set of inputs $S$ of 
size $|S|=\poly(n)$ such that the following holds:
\begin{itemize}
  \item The problem $R$ can be solved with certainty for all inputs  by 
  a constant-depth quantum circuit composed of geometrically local gates on a 1D grid.
  \item Any classical probabilistic circuit composed of constant fan-in gates that solves $R$ with probability exceeding $0.9$ for a uniformly random input from $S$ must have depth at least $\Omega(\log n)$.
\end{itemize}
\label{res:1dqa}
\end{result}

The proof of this result is given in Section \ref{sec:1Dmagicsquare}, where the formal statements appear as Theorems~\ref{thm:1dq} and~\ref{thm:lowerbound}. 
As in previous work \cite{bragokoe18, coudronetal18}, the separation described in Result~\ref{res:1dqa} is achieved by a quantum algorithm with input/output statistics that are related to those of a certain nonlocal game. Recall that in a nonlocal game, cooperating players are each provided with an input and must each produce an output without communicating with the other players. Their aim is to satisfy a given winning condition, or input/output relation. It is known that quantum players who share entanglement can win certain nonlocal games with higher probability than classical players who share randomness. To prove the above result, we exhibit a constant-depth one-dimensional quantum circuit and a set $S$ of inputs such that the input/output statistics of the circuit given any input in $S$ are directly related to a variant of the well known magic-square game ~\cite{peres90,mermin90}.  We further establish that for any classical circuit with low enough depth there are a significant fraction of inputs in $S$ for which the circuit can be viewed as executing a classical strategy for winning this nonlocal game. The result then follows as a result of upper bounds on the winning probability of any classical strategy. The constant-depth quantum circuit which achieves this quantum advantage is shown in Fig.~\ref{fig:1d}. It is a classically controlled Clifford circuit with a particularly simple one-dimensional structure, and may be suitable for a near-term experimental demonstration. 

Secondly, we ask if the separation between the power of  constant-depth classical and quantum circuits persists even for noisy quantum circuits, i.e., quantum circuits where each qubit/gate can be erroneous with a constant probability. In this paper we compare the computational power of noisy shallow quantum circuits with that of noise-free shallow classical probabilistic circuits. 
The quantum circuits we consider will be subject to \textit{local stochastic noise}~\cite{fawzi2018constant}.
This noise model assumes that a random Pauli error occurs at each time step
in the ideal circuit. The error may affect multiple qubits, but the probability
of high-weight errors must be exponentially suppressed. 
This is quantified by a {\em noise rate} $p\in [0,1]$ such that 
the probability of observing $k$ single-qubit errors at any given 
subset of $k$ qubits must be at most $p^k$, see Section~\ref{sec:noise} for formal definitions.
The (probabilistic) classical circuits we consider will be composed of gates of bounded fan-in, as defined in Section~\ref{sec:hardnessoneDsquare}.

We note that standard fault-tolerance constructions which emulate a noise-free universal quantum computation using faulty gates and measurements do not directly apply in this setting: these constructions typically lead to non-constant depth circuits. As an example, a quantum error-correcting code with extensive code distance does not have a constant-depth encoding circuit~\cite{bravyi2006lieb,eldar2017local,aharonov2018quantum}. Thus, standard quantum error correction methods do not directly provide a generic way to turn a separation such as that established in~\cite{bragokoe18}, or the one described in Result~\ref{res:1dqa}, into a separation between noisy constant-depth quantum and (noiseless) constant-depth classical circuits.  Nevertheless, in this paper we do provide such a generic recipe. Applying the recipe to the separation described in Result~\ref{res:1dqa} we obtain the following.

\begin{result}[\textbf{Quantum advantage with noisy shallow circuits\,---\,informal}]
\label{res:noisy}
For each $n$ there exists a relation problem $R$ with roughly $n$ input-output bits and a set of inputs $S$ of size $|S|=poly(n)$ such that the following holds:
\begin{itemize}
  \item The problem $R$ can be solved with probability at least $0.99$ for all inputs by 
 a constant-depth quantum circuit composed of geometrically local gates on a 3D grid, subject to local stochastic noise.
 The noise rate must be below a constant threshold value independent of $n$.
  \item Any classical probabilistic circuit composed of constant fan-in gates that solves $R$ with probability exceeding $0.9$ for a uniformly random input from $S$ must have depth at least
  \[
    \Omega\left(\frac{\log(n)}{\log(\log(n))}\right).
  \]
\end{itemize}
\end{result}

Let us briefly describe the main idea which allows us to convert a quantum advantage with ideal quantum circuits, such as in Result~\ref{res:1dqa}, into one with noisy quantum circuits. The recipe is detailed in Section \ref{sec:noisyadvantageconstruction}. It uses the facts that (A) the quantum circuits which achieves the separation  are  controlled Clifford circuits with a classical control (i.e., for any fixed input a Clifford unitary is applied), and (B) Certain classical computations, such as the decoding needed for quantum error correction, can be incorporated into the definition of the relation problem rather than performed explicitly in the quantum algorithm. 

Consider a relation problem $R$ such that a constant-depth controlled-Clifford circuit 
produces a solution to a given instance with certainty. We are interested in the setting where $R$ cannot be satisfied by any constant-depth classical circuit. Such relations $R$ are provided in Ref.~\cite{bragokoe18} and Result~\ref{res:1dqa}. 
For a fixed input the controlled-Clifford circuit implements a constant-depth Clifford unitary~$C$ 
acting on $n$ qubits
followed by measurement of all qubits in the computational basis. Suppose that our goal is to perform a fault-tolerant version of this computation.  We imagine encoding each logical qubit using $m$ physical qubits of some
CSS-type~\cite{calderbank1996good,steane1996multiple} stabilizer code $\mathcal{Q}_m$.

As noted above, since good codes do not admit constant-depth encoding circuits, we are unable to initialize all logical qubits in the state $|\overline{0}\rangle$.  However, we can hope to prepare a version of this state which is corrupted by a known Pauli operator (which may act nontrivially on all physical qubits). To do this we can initialize all $m$ physical qubits, along with a suitable number $m_{\mathrm{anc}}$ of ancilla qubits, in the all-zeros state, and then perform a Clifford circuit $W$ which measures all stabilizers of the code to obtain a syndrome $s$. The resulting state is then
\begin{equation}
\left(I\otimes |s\rangle\langle s|\right) W|0^{m}\rangle|0^{m_{\mathrm{anc}}}\rangle \propto \rec(s)|\overline{0}\rangle|s\rangle.
\label{eq:stateprep}
\end{equation}
where the ``recovery" Pauli operator $\rec(s)$ is a function of the syndrome $s$. We shall be interested in the case when the code $\calQ_m$ is a low-density parity-check (LDPC) code, i.e., it has constant weight
stabilizer generators such that each qubit is acted upon nontrivially by at most a constant number of them. The syndrome of such codes can be measured by a constant depth Clifford circuit $W$. Using this procedure we can prepare the desired logical  state $|\overline{0}\rangle$ 
modulo a Pauli recovery operator $\rec(s)$. 
The same method can be used to prepare
$n$ copies of the state $|\overline{0}\rangle$, modulo a Pauli recovery
$\rec(s)$ acting on $nm$ qubits.
Let $\overline{C}$ be the logical version of the Clifford circuit $C$. 
Applying this circuit  to the prepared logical all-zero state we obtain
\begin{equation}
\overline{C}\rec(s)|\overline{0}\rangle^{\otimes n}=P(s)\overline{C}|\overline{0}\rangle^{\otimes n}
\label{eq:corrupted2}
\end{equation}
where $P(s)=\overline{C}\rec(s)\overline{C}^{\dagger}$ is another Pauli operator which is a simple function of $s$. Here we require that the logical Clifford $\overline{C}$ is implementable by a constant-depth physical circuit (for example,
this holds for any CSS code with transversal logical Hadamard and phase gates).  In other words, using such a code $\mathcal{Q}_m$ we are able to implement a logical encoded version of the constant-depth Clifford circuit $C$, masked by a Pauli operator $P(s)$ that depends on the initial syndrome measurement $s$ obtained in state preparation. The computational basis measurement statistics of the encoded state with the mask Eq.~\eqref{eq:corrupted2} are related to those of the unencoded state with no mask $C|0\rangle^{\otimes n}$ by flipping the bits corresponding to the $X$-type part of $P(s)$ and then decoding the resulting bit string. Thus we can simulate the desired encoded quantum computation using a constant-depth quantum circuit along with some simple classical postprocessing. If we chose to incorporate this classical postprocessing into the quantum algorithm, it could pose a problem as its depth may not be constant. Happily, it turns out, we can instead modify the definition of the relation problem $R$ to account for the difference. 

Now let us consider the noise-tolerance of this procedure.  Since the above quantum circuit has a constant depth and 
 uses logical encoded qubits and operations, it can be made to work in the presence of noisy physical gates and measurements, as long as they occur after the state preparation step.  Unfortunately, the state preparation step Eq.~\eqref{eq:stateprep} is not generally fault-tolerant and the whole algorithm can fail due to errors in the measured syndrome $s$. For example, a single faulty bit of $s$ can potentially 
damage the recovery operator $\rec(s)$ at multiple qubits
resulting in an uncorrectable error.
  This can be addressed by using a code $\mathcal{Q}_m$ that admits a so-called \textit{single-shot state preparation} procedure.  The latter is closely related to a single-shot error correction~\cite{bombin2015single}.
The code $\mathcal{Q}_m$ is said to admit a single-shot state preparation for a single-qubit logical state $\overline{\phi}$ if there exists a number of ancillas $m_{\mathrm{anc}}$ (upper bounded by a polynomial function of $m$) and a constant-depth Clifford circuit $W$
acting on $m+m_{\mathrm{anc}}$ qubits such that, for any local stochastic Pauli error $E$ with noise rate $p$, we have
\[
\left(I\otimes |s\rangle\langle s|\right)EW|0^{m}\rangle|0^{m_{\mathrm{anc}}}\rangle \propto F \rec(s)|\overline{\phi}\rangle|s\rangle.
\]
where $F$ is also a local stochastic Pauli error with a possibly  larger noise rate $p'\leq c_1p^{c_2}$ for positive constants $c_1,c_2$. For example, single-shot state basis state preparation allows us to use a constant-depth circuit composed of noisy gates and measurements to prepare a state $F \rec(s)|\overline{0}\rangle|s\rangle$, where $F$ is a random Pauli error that can be viewed as residual noise. We can also consider single-shot preparation of $k$-qubit encoded states, with $k>1$, in which case $m$ should be replaced by $mk$ above.

Putting together these ingredients we obtain a recipe which starts with a relation $R$ defined by the input-output statistics of a constant-depth controlled-Clifford circuit, and converts this ``bare relation" into a ``noise-tolerant" relation $\mathcal{R}$ that is based on the encoded circuit with single-shot state preparation, and which incorporates the classical postprocessing in its definition. We further show that the input/output statistics of a constant-depth quantum circuit satisfy  $\mathcal{R}$, and we show that the depth required for a classical circuit to satisfy $\mathcal{R}$ is comparable to that required to satisfy the bare relation $R$. 

A crucial requirement for the recipe outlined above is the existence of a CSS stabilizer code $\mathcal{Q}_m$ such that elementary logical Clifford gates are implemented by constant-depth Clifford circuits, and which admits a single-shot state preparation procedure. Here we show that the standard surface code satisfies these desiderata. The first requirement follows from previous work~\cite{moussa2016transversal} which describes how to implement logical single-qubit Hadamard and phase gates in the surface code using constant-depth Clifford circuits. Together with the transversal logical CNOT gate this provides a complete set of Clifford generators which can each be implemented in constant depth. A central technical contribution of our work is to provide a single-shot state preparation procedure for the surface code. Specifically, we show how  to prepare a logical Bell state encoded in two identical surface codes.
\begin{result}[\textbf{Single-shot Bell state preparation in the surface code\,---\,informal}]
For each $d\geq 4$, there is a single-shot state preparation procedure for the encoded Bell state 
${2}^{-1/2}\left(|\overline{00}\rangle+|\overline{11}\rangle\right)$ shared between two 
distance-$d$ surface codes, each encoding one logical qubit into $m=d^2+(d-1)^2$ physical qubits.
The procedure uses a depth-$6$ Clifford circuit $W$ composed of geometrically local gates on a 3D grid and computational basis measurements.
\label{res:intro3}
\end{result}

The proof of Result~\ref{res:intro3}, given in Sections~\ref{sec:codes} and~\ref{sec:3Dcluster},
relies crucially on ideas introduced in Ref.~\cite{raussendorfetal05}.
The authors of Ref.~\cite{raussendorfetal05} showed 
how to prepare a logical Bell state encoded into a pair
of surface codes starting from a 3D grid of qubits initially prepared in a (noisy) cluster state
and measuring a suitable subset of qubits.
Here we extend the analysis of Ref.~\cite{raussendorfetal05} and prove that the
same protocol yields a single-shot state preparation scheme with a constant
error threshold in the presence of local stochastic noise. 
We leave as an open question whether Result~\ref{res:intro3}
in conjunction with Knill's syndrome measurement method~\cite{knill312190scalable,GottesmanChuangNature}
provides a single-shot error correction scheme based on the surface code.

The 3D constant-depth quantum circuit described in Result~\ref{res:noisy} is obtained by combining
the 3D Bell state preparation circuit of Result~\ref{res:intro3}
with the 1D circuit of Result~\ref{res:1dqa} encoded 
by the surface code (we shall see
that the first few gates of this circuit simply prepare Bell states).
We show that  the encoded 1D circuit can be made geometrically local 
on a 3D grid  using the lattice folding trick introduced in Ref.~\cite{moussa2016transversal}.
The folded encoded 1D circuit uses only nearest-neighbor two-qubit gates on a 3D grid
with $O(1)$ qubits per site, as detailed in Section~\ref{sec:3D}.

\subsubsection*{Outline}
The remainder of the paper is organized as follows:
In Section~\ref{sec:1Dmagicsquare}, we introduce a new computational problem, the {\em 1D Magic Square Problem},  
separating constant-depth classical and quantum circuits. Contrary to the hidden linear function problem considered in~\cite{bragokoe18} which relied on a 2D~qubit architecture, the quantum circuit for the 1D~Magic Square Problem is geometrically local in one dimension.  An added benefit is a  simpler proof of the computational hardness for  constant-depth classical circuits.

In Section~\ref{sec:noisyadvantageconstruction}, we show how noise can be addressed for suitable (relation) problems: given an ideal constant-depth classically controlled Clifford circuit solving a certain relation problem, we show how to define a noise-tolerant version of the problem. The latter can be solved with a noisy quantum circuit constructed using appropriate error-correcting codes, and retains the hardness (in terms of circuit depth) of the original relation for classical circuits. Instantiating this construction with the 1D~Magic Square Problem provides the desired separation between noisy constant-depth quantum and (noise-free) constant-depth classical  circuits.

In Section~\ref{sec:codes}, we explain how to obtain the required code properties using the standard 2D~surface codes.  We give a high-level overview of the procedure for single-shot encoded Bell state preparation based on a 3D~grid of qubits. The full proof is  provided in  Section~\ref{sec:3Dcluster}.

 In Section~\ref{sec:3D}, we argue that the required quantum circuit for the noise-tolerant 1D~Magic Square Problem  can be realized using a constant-depth circuit with geometrically local gates on a 3D~architecture of qubits.

\section{\mbox{The 1D Magic Square Problem: 
Quantum advantage in a 1D geometry}\label{sec:1Dmagicsquare}}

In this section we define a relation problem called the 1D Magic Square Problem. We show that it can be solved with certainty by a constant-depth quantum circuit with nearest neighbor gates in a one-dimensional geometry. Conversely, we prove that it cannot be solved with high probability by any constant-depth classical (probabilistic) circuit composed of bounded fan-in gates. We begin by describing the \emph{magic square game}~\cite{peres90,mermin90}.

\begin{figure}
\subcaptionbox{Quantum strategy}[5cm]{
\begin{tikzpicture}[scale=0.75]

\draw (0,0)--(3,0);
\draw (0,1)--(3,1);
\draw (1.5, 1.45) node {$|\Phi\rangle$};
\draw (1.5, -0.55) node {$|\Phi\rangle$};
\draw (0,0) node[circle, fill=black, scale=0.5]{};
\draw (3,0) node[circle, fill=black, scale=0.5]{};

\draw (0,1) node[circle, fill=black, scale=0.5]{};
\draw (3,1) node[circle, fill=black, scale=0.5]{};

\draw [->, thick](-1.5, 2)--(-1.5, 1);
\draw [->, thick](-1.5, 0)--(-1.5, -1);
\draw (-1.5, 2.25) node {$\alpha$};
\draw (-1.5, -1.5 ) node {$\mathbf{x}$};

\draw [dotted, thick] (-2.5, -0.5)--(-2.5, 1.5)--(0.5,1.5)--(0.5,-0.5)--(-2.5, -0.5);
\begin{scope}[shift={(5,0)}]
\draw [dotted, thick] (-2.5, -0.5)--(-2.5, 1.5)--(0.5,1.5)--(0.5,-0.5)--(-2.5, -0.5);
\draw [->, thick](-1.25, 2)--(-1.25, 1);
\draw [->, thick](-1.25, 0)--(-1.25, -1);
\draw (-1.25, 2.25) node {$\beta$};
\draw (-1.25, -1.5 ) node {$\mathbf{y}$};
\end{scope}
\draw (-1.5, 0.5) node {Alice};
\draw (4, 0.5) node {Bob};
\end{tikzpicture}
}
\hspace{2cm}
  \subcaptionbox{Optimal measurements}[6cm]{
    \begin{tabular}{|c|c|c|c|} \hline
      \backslashbox{~${\beta}$}{${\alpha}$} & $01$ & $10$ & $11$\\
      \hline
      $01$ & $X_1 \id_2$ & $\id_1 X_2$ & $X_1 X_2$ \\
      \hline
      $10$ &  $\id_1 Z_2$ & $ Z_1 \id_2$ & $Z_1 Z_2$ \\
      \hline
       $11$ & $-X_1 Z_2$ & $-Z_1 X_2$ & $Y_1 Y_2$ \\
      \hline
      \end{tabular}
  }
\caption{The magic square game. (a) The two-bit inputs $\alpha,\beta\in \{01,10,11\}$ to the magic square game specify a column and row for Alice and Bob respectively, and the outputs are three bits $\mathbf{x} = (x^1, x^2, x^3) \in \{-1,1\}^3$ and $\mathbf{y} = (y^1, y^2, y^3) \in \{-1,1\}^3$ for each entry in the column or row. The game is won when $x^1 x^2  x^3 = -1$, $y^1 y^2 y^3 = 1$ and $x^{\iota(\beta)} = y^{\iota(\alpha)}$ (where $\iota$ converts between binary and non-binary representation). (b) The commuting observables for Alice (columns) and Bob (rows) for each input, yielding 3 output bits for each Alice and Bob that always satisfy the parity constraints.\label{tb:magicstab}} \label{fig:magic}
 \end{figure}
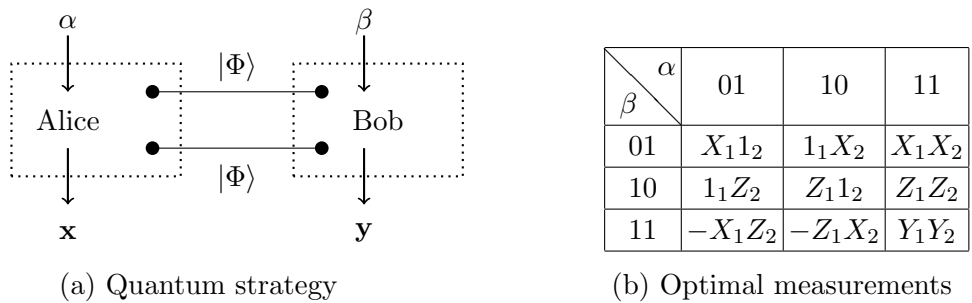
\subsection{The (generalized) magic square game}
The magic square game is a nonlocal game with two cooperating players Alice and Bob who cannot communicate. At the outset, Alice is given an input $\mathbf{\alpha}\in \{01,10,11\}$ which specifies  one of the three columns~$\iota(\mathbf{\alpha})\in \{1,2,3\}$ of a $3\times 3$ table and Bob is given an input $\mathbf{\beta}\in \{01,10,11\}$ which specifies one of the rows $\iota(\mathbf{\beta})\in \{1,2,3\}$. For later convenience, we use a binary encoding of integers with conversion map~$\iota:\{0,1\}^2\rightarrow \{0,1,2,3\}$.   Alice is asked to fill in the three entries in her column with either zeros or ones such that the overall parity is odd, while Bob is asked to fill out his row so that the parity is even.  They win if they satisfy this property and in addition they report the same value for the square where the column $\iota(\alpha)$ and row $\iota(\beta)$ overlap. There is no fixed assignment of the table satisfying the winning condition, which can be seen by noting that the total parity of all bits must be either even or odd, contradicting one of the restrictions. In fact, the maximal winning probability using a classical strategy is $8/9$. On the other hand, quantum players can win this game with certainty if Alice and Bob measure the observables in Fig.~\ref{fig:magic} on two maximally entangled states $|\Phi\rangle^{\otimes 2}$, where $|\Phi\rangle=2^{-1/2}(|00\rangle+|11\rangle)$.

This quantum strategy for the magic square game can alternatively be depicted using a quantum circuit as shown in Fig.~\ref{fig:magiccircuit}(a). After creating two Bell states $|\Phi\rangle^{\otimes 2}$, Alice (top) applies a Clifford unitary $U(\alpha)$ and then measures in the computational basis to obtain outcomes~$(x^1, x^2)$. 
Similarly, Bob (bottom) applies a Clifford unitary $V(\beta)$ and measures in the computational basis, getting outcomes~$(y^1, y^2)$.  Alice's output then is $(x^1,x^2,x^3)\in \{-1,+1\}^3$, where $x^1$ and $x^2$ are the measurement outcomes and the third component $x^3=-x^1x^2$ is fixed by the parity constraint in the magic square game. Similarly, Bob's output is $(y^1,y^2,y^3)\in \{-1,+1\}^3$, with $y^3=y^1y^2$. The Clifford unitaries $U(\alpha),V(\beta)$ implement the measurements described in Fig.~\ref{fig:magic}.  For example, $U(01) = H_1 \id_2$ implements the measurements $X_1 \id_2$ (on the first qubit) and $\id_1 Z_2$ (on the second qubit). For later convenience we set $U(00) = V(00) = I$. The full list of Clifford unitaries applied by Alice respectively Bob is given in Fig.~\ref{fig:magiccircuit}(b).

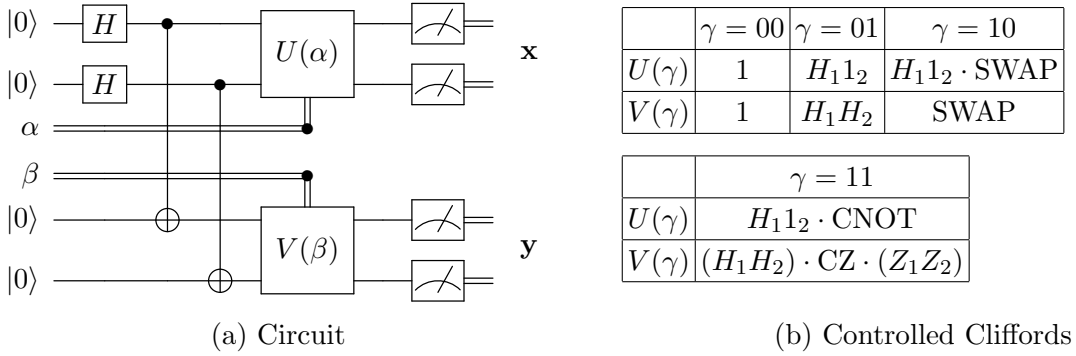
\begin{figure}
\subcaptionbox{Circuit}[0.5\textwidth]{
{
  \Qcircuit @C=1.0em @R=0.7em {
    & & & & \\
       \lstick{\ket{0}}     \lstick{\ket{0}}\labelgroup{2}{3}{1.0em}{\hspace{2.5em}{\mathbf{x}}}{9}					& \gate{H}  & \ctrl{5}    & \qw                & \multigate{1}{U({\alpha})} &\qw & \meter & \cw & \\ 
      \lstick{\ket{0}}             & \gate{H}  &  \qw        & \ctrl{5}           &      \ghost{U({\alpha})}      &\qw  & \meter & \cw & \\
       \lstick{{\alpha}}  & \cw          & \cw          & \cw                & \cctrl{-1} & \\
                              &                &              & \\
       \lstick{\beta}               & \cw          & \cw     & \cw        & \cctrl{1}  &\\
            \lstick{\ket{0}}\labelgroup{7}{8}{1.0em}{\hspace{2.5em}{\mathbf{y}}}{9} & \qw & \targ        & \qw             & \multigate{1}{V({\beta})} & \qw & \meter & \cw&\\
      \lstick{\ket{0}}  & \qw  &  \qw      & \targ            &        \ghost{V(\beta)}   &  \qw   & \meter & \cw &   }
  }
 }
\subcaptionbox{Controlled Cliffords}[0.45\textwidth]{
\begin{minipage}{0.45\textwidth}
\begin{flushleft}
\begin{tabular}{|c |c| c| c |}
\hline
&$\gamma=00$ & $\gamma=01$ & $\gamma=10$ \\ 
\hline
$U(\gamma)$ & $\id$ & $H_1\id_2$ & $H_1\id_2\cdot \mathrm{SWAP}$  \\
\hline
$V(\gamma)$ & $\id$ & $H_1H_2$ & $\mathrm{SWAP}$    \\
\hline
\end{tabular} 
\vspace{0.25cm}

\begin{tabular}{|c | c |}
\hline
&$\gamma=11$\\ 
\hline
$U(\gamma)$ & $H_1\id_2\cdot \mathrm{CNOT}$ \\
\hline
$V(\gamma)$ & $(H_1H_2)\cdot \mathrm{CZ}\cdot (Z_1 Z_2) $   \\
\hline
\end{tabular}
\vspace{0.25cm}

\end{flushleft}
\end{minipage}
}
\caption{
(a) The quantum strategy for the magic square game as a circuit.  The circuits takes as input $\alpha,\beta\in \{0,1\}^2$. The values $\alpha=00$ and $\beta=00$ are not used in the magic square game. The outputs are $\mathbf{x}=(x^1,x^2),\mathbf{y}=(y^1,y^2)\in \{+1,-1\}^2$. 
We set $x^3=-x^1x^2$ and $y^3=y^1y^2$ to satisfy the parity constraints.(b) When $\gamma=00$ we set $U(\gamma)=V(\gamma)=I$. For $\gamma\in \{01,10,11\}$ we choose $U(\gamma)$ and $V(\gamma)$ to be Cliffords which implement the basis changes needed to measure the observables described in Fig. \ref{fig:magic} (b). Here $\mathrm{CZ}=\mathrm{diag}(1,1,1,-1)$ is the controlled-$Z$ gate and $\mathrm{SWAP}$ is defined by $\mathrm{SWAP}|z_1z_2\rangle=|z_2z_1\rangle$ for all $z_1,z_2\in \{0,1\}$.
\label{fig:magiccircuit}}
\end{figure}

We briefly describe a \textit{generalized magic square game with parameters }$(s,t,s',t')\in \{-1, +1\}^4$. In the generalized game, Alice and Bob are still asked to fill in columns respectively rows of a $3\times 3$ table, with odd  or even parity constraints as above. However, the winning condition~$x^{\iota(\beta)}=y^{\iota(\alpha)}$ that  Alice's outputs~$(x^1,x^2,x^3)$ and Bob's outputs $(y^1,y^2,y^3)$ coincide in the entry where the column~$\iota(\alpha)$  and the row~$\iota(\beta)$ overlap is replaced by
\begin{equation}
x^{\iota(\beta)} y^{\iota(\alpha)} = f_{\alpha,\beta}(s,t,s',t')\ ,\label{eq:win2}
\end{equation}
where $f_{\alpha,\beta}(s,t,s',t')$ is given by the table in Fig. \ref{fig:modmag}(b).

It is straightforward to show that the maximal winning probability for the generalized magic square game using a classical strategy is again equal to~$8/9$. On the other hand, quantum players can win with probability one 
if they share the entangled state $\ket{\Phi_{s,t}}\otimes \ket{\Phi_{s',t'}}$, where
\begin{equation}
\ket{\Phi_{s,t}} = \left(Z^{\frac12(1+s)}X^{\frac12(1+t)}\otimes I\right) \ket{\Phi}  \qquad \qquad s,t\in \{-1, +1\}\ .
\label{eq:varphistzerozero}
\end{equation}
Here the tensor product separates Alice's qubit (on the left) from Bob's.  The corresponding winning strategy consists in measuring the same observables (i.e., from Fig.~\ref{fig:magic}) as before.

To see that this quantum strategy succeeds with probability~$1$, observe that because of~\eqref{eq:varphistzerozero}, this is equivalent to Alice and Bob sharing the initial state~$\ket{\Phi^{\otimes 2}}$, 
Bob measuring the same observables as before (described by the rows of Fig \ref{fig:magic}(a)), and Alice measuring the observables described in Fig.~\ref{fig:modmag}(a). Clearly, the outcomes 
$x^1, x^2, x^3$ for Alice and $y^1, y^2, y^3$ for Bob still satisfy the parity conditions
as Bob's measurement is the same as before, whereas Alice's 
outcomes again multiply to~$-1$ as can be seen by taking the product of
the operators in each column. The fact that~\eqref{eq:win2} is satisfied follows by comparing the observables in Fig.~\ref{fig:modmag}(a) with the definition of $f_{\alpha,\beta}$ (see Fig.~\ref{fig:magic}(b)).



In our arguments below (see Lemma~\ref{lem:modifiedgamemain}), we use a variant of this  generalized magic square game where $s,t,s',t'$ enter the winning conditions and may depend on the inputs $(\alpha,\beta)$, but only in a restricted way. For later reference, we note that the functions~$f_{\alpha,\beta}$ satisfy
\begin{align}
  \prod_{i=1}^3 f_{\iota^{-1}(i),\beta}(s, t, s', t') = \prod_{j=1}^3 f_{\alpha,\iota^{-1}(j)}(s, t, s', t') = 1\qquad\textrm{ for all }
(\alpha,\beta)\in \{01,10,11\}\ .
\label{eq:stprod}
\end{align}
Observe also that each of the functions $f_{\alpha,\beta}$ depends linearly on each of the arguments $(s,t,s',t')$. In particular,
if these variables are products of $\{+1,-1\}$-valued variables, then the value of the function factorizes as
\begin{align}
f_{\alpha,\beta}(s_As_B,t_At_B,s'_As'_B,t'_At'_B)=f_{\alpha,\beta}(s_A,t_A,s'_A,t'_A)\cdot f_{\alpha,\beta}(s_B,t_B,s'_B,t'_B)\label{eq:factorizationproperty}
\end{align}
for all $s_A,s_B,t_A,t_B,s'_A,s'_B,t'_A,t'_B\in \{+1,-1\}$.

\begin{figure}
\subcaptionbox{Alice's observables}[5cm]{

\begin{tabular}{|c| c| c | c |}
\hline
$\alpha=01$ & $\alpha=10$ & $\alpha=11$\\
\hline
$s X_1 \id_2$ &  $s'\id_1 X_2$ &  $ss'X_1 X_2$ \\
\hline
$t'\id_1 Z_2$ &  $tZ_1 \id_2$ &  $tt'Z_1 Z_2$ \\
\hline
$-st' X_1 Z_2$ & $-ts'Z_1 X_2$ &  $ss'tt'Y_1 Y_2$ \\
\hline
\end{tabular}
}
\hspace{2cm}
\subcaptionbox{Definition of $f_{\alpha,\beta}(s,t,s',t')$}[6cm]{
\begin{tabular}{ |c| c| c | c |}
\hline
 & $\alpha=01$ & $\alpha=10$ & $\alpha=11$\\ 
\hline
$\beta=01$ &   $s$ &  $s'$ &  $ss'$ \\
\hline
$\beta=10$ &   $t'$ & $t$ & $tt'$ \\
\hline
$\beta=11$ &    $st'$ & $s't$ & $ss'tt'$ \\
\hline
\end{tabular}
}
\caption{Generalized magic square game with initial state $|\Phi_{s,t}\rangle\otimes |\Phi_{s',t'}\rangle$.\label{fig:modmag}}
\end{figure}

\subsection{The 1D Magic Square Problem and its solution by a constant-depth quantum circuit}
We now describe the relation problem which we call the 1D Magic Square Problem. We simultaneously exhibit a constant-depth quantum circuit using classically controlled Clifford gates which solves this problem with certainty for any input.  In fact, we define the problem by giving this quantum circuit, but remark that, alternatively, a purely algebraic definition could be given without making reference to quantum circuits.

\begin{figure}
  \[
  \Qcircuit @C=1.0em @R=0.7em {
        & & & & & \\
     \lstick{p_1 \; \ket{0}}     \lstick{\ket{0}}\labelgroup{2}{3}{1.0em}{\hspace{2.5em}\mathbf{x}_1}{9}					& \gate{H}  & \ctrl{6}    & \qw                & \multigate{1}{U(\malpha_1)} &\qw & \meter & \cw & \\ 
       \lstick{q_1 \;\ket{0}}             & \gate{H}  &  \qw        & \ctrl{6}           &      \ghost{U(\malpha_1)}      &\qw  & \meter & \cw & \\
        \lstick{\malpha_1}  & \cw          & \cw          & \cw                & \cctrl{-1} & \\
                              &                &              & \\
                               &                &              & \\
        \lstick{\mbeta_1}               & \cw          & \cw     & \cw        & \cctrl{1}  & \cctrl{1}\\
        \lstick{p_2 \;\ket{0}}\labelgroup{8}{9}{1.0em}{\hspace{2.5em}\mathbf{y}_1}{9} & \qw  & \targ        & \qw             & \multigate{1}{V(\mbeta_1)} & \multigate{3}{W(\mbeta_1,\malpha_2)} & \meter & \cw&\\
     \lstick{q_2 \;\ket{0}}  & \qw  &  \qw      & \targ            &        \ghost{V(\mbeta_1)}   &  \ghost{W(\mbeta_1,\malpha_2)}    & \meter & \cw & \\
      \lstick{p_3 \;\ket{0}} \labelgroup{10}{11}{1.0em}{\hspace{2.5em}\mathbf{x}_2}{9} 			& \gate{H}  & \ctrl{6}    & \qw                & \multigate{1}{U(\malpha_2)} &\ghost{W(\mbeta_1,\malpha_2)}    & \meter & \cw& \\
      \lstick{q_3 \;\ket{0}}             & \gate{H}  &  \qw        & \ctrl{6}           &      \ghost{U(\malpha_2)}      &\ghost{W(\mbeta_1,\malpha_2)}     & \meter & \cw& \\
       \lstick{\malpha_2}  & \cw          & \cw          & \cw                & \cctrl{-1} & \cctrl{-1} \\
                              &                &              & \\
                               &                &              & \\
        \lstick{\mbeta_2}               & \cw          & \cw     & \cw        & \cctrl{1}  & \cctrl{1}\\
        \lstick{p_4 \; \ket{0}} \labelgroup{16}{17}{1.0em}{\hspace{2.5em}\mathbf{y}_2}{9} & \qw  & \targ        & \qw             & \multigate{1}{V(\mbeta_2)} & \multigate{3}{W(\mbeta_2,\malpha_3)} & \meter & \cw& \\
     \lstick{q_4 \; \ket{0}}  & \qw  &  \qw      & \targ            &        \ghost{V(\mbeta_2)}   &  \ghost{W(\mbeta_2,\malpha_3)}    & \meter & \cw & \\
      \lstick{p_5 \; \ket{0}}\labelgroup{18}{19}{1.0em}{\hspace{2.5em}\mathbf{x}_3}{9}  & \gate{H}  & \ctrl{3}    & \qw                & \multigate{1}{U(\malpha_3)} &\ghost{W(\mbeta_2,\malpha_3)}    & \meter & \cw &\\
     \lstick{q_5 \; \ket{0}}             & \gate{H}  &  \qw        & \ctrl{2}           &      \ghost{U(\malpha_3)}     &\ghost{W(\mbeta_2,\malpha_3)}     & \meter & \cw&\\
      \lstick{\malpha_3}  & \cw          & \cw          & \cw                & \cctrl{-1} & \cctrl{-1} \\
                              &                &              & \\
                               &                &              &\\
   & \vdots & & &  \vdots &  & \vdots \\                     
                                  &                &              & \\
                               &                &              & \\
                                 & & & \\
                                  & & \qwx[2] & \qwx[3] \\
                                \lstick{\mbeta_n}               & \cw          & \cw     & \cw        &  \cctrl{1} & &  &   \\
           \lstick{p_{2n} \; \ket{0}}  \labelgroup{29}{30}{1.0em}{\hspace{2.5em}\mathbf{y}_n}{9}        & \qw  & \targ        & \qw             & \multigate{1}{V(\mbeta_n)} & \qw & \meter & \cw  & \\
                   \lstick{q_{2n} \; \ket{0}}  & \qw  &  \qw        & \targ            &        \ghost{V(\mbeta_n)}    &  \qw   & \meter & \cw  &   
                    }
                    \]
  \caption[justification=centerlast]{Quantum circuit for the 1D Magic Square Problem.  Here the classically controlled Clifford~$W(\mbeta,\malpha)$ is the identity unless $\malpha=\mbeta=(0,0)$. In that case, we set $W(00,00)=M_{13}\otimes M_{24}$ where $M=(H\otimes I)\CNOT$ is the Bell basis change, leading to an entanglement-swapping measurement.\label{fig:1d}}
\end{figure}
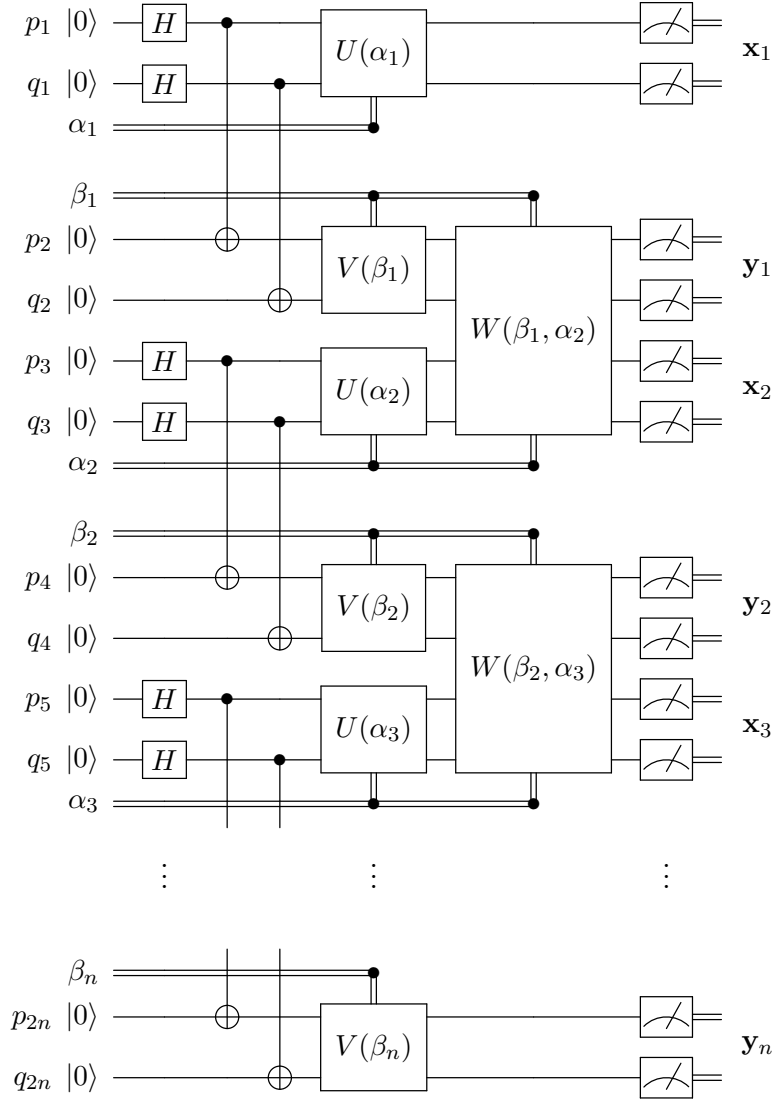
Consider the quantum circuit shown in Fig.~\ref{fig:1d}. To understand what is going on in this circuit, it may be useful to compare with Fig.~\ref{fig:magiccircuit}. The circuit in Fig.~\ref{fig:1d} takes inputs $\malpha_j, \mbeta_j \in \{0,1\}^2$  and outputs $\mathbf{x}_j, \mathbf{y}_j \in \{-1,1\}^2$,  for $j\in \{1,2,\ldots, n\}$. Thus the circuit has a total of $4n$ input bits and $4n$ output bits. It contains $4n$ data qubits which are labelled $p_1,q_1,p_2,q_2,\ldots, p_{2n},q_{2n}$ in the Figure. The gates $U(\malpha)$ and $V(\mbeta)$ are the same Clifford gates which appear in Fig.~\ref{fig:magiccircuit}. Recall that for $\malpha,\mbeta\in \{01,10,11\}$ they are chosen to implement Alice and Bob's measurements in the magic square game, and $U(00)=V(00)=I$. In addition, the circuit contains controlled $4$-qubit Clifford gates $W(\mbeta,\malpha)$ defined by 
\begin{align}
W(\mbeta,\malpha)=\begin{cases} M_{13}\otimes M_{24}, & \malpha=\mbeta=00\\
I, & \text{ otherwise}\ .
\end{cases}\label{eq:Wbetaalphadef}
\end{align}
where $M=(H\otimes I)\CNOT$ is the Bell basis change, mapping the Bell basis to the computational one. We note that the circuit realizes a Clifford unitary~$\ConeD_{z_{in}}$ on $4n$~qubits which is classically controlled
by the input~$z_{in}=(\malpha_1,\ldots,\malpha_n,\mbeta_1,\ldots,\mbeta_n)$, followed by a computational basis measurement yielding~$z_{out}=(\mathbf{x}_1,\ldots,\mathbf{x}_n,\mathbf{y}_1,\ldots,\mathbf{y}_n)$. In the 1D Magic Square Problem we are asked to reproduce the input-output relation satisfied by the circuit:
\begin{definition}[\textbf{1D Magic Square Problem}]
We are given an input $z_{in}\in\{0,1\}^{4n}$. The goal is to output any bit string $z_{out}\in \{0,1\}^{4n}$ which
 appears with nonzero probability when the circuit Fig.~\ref{fig:1d} is run with input $z_{in}$, i.e.,
 any $z_{out}$ satisfying
 \begin{align}
 p_{z_{in}}(z_{out})=|\bra{z_{out}}\ConeD_{z_{in}}\ket{0^{4n}}|^2>0\ .\label{eq:magicsquareonedrelation}
 \end{align}
 Any pair $(z_{in},z_{out})$ satisfying~\eqref{eq:magicsquareonedrelation} will be said to satisfy the 1D Magic Square Relation.
\end{definition}
In other words, the 1D Magic Square Problem is a relation problem: for a typical problem instance~$z_{in}$, the set of valid solutions~$z_{out}$ is non-unique. By construction, the problem can be solved with certainty by running the circuit in Fig.~\ref{fig:1d}. Thus, we immediately obtain the following.
\begin{theorem}
The 1D Magic Square Problem can be solved with certainty by a depth-$4$ quantum circuit in which each gate is a controlled Clifford gate with at most $4$ control bits and $4$ target qubits. The gates in this circuit are geometrically local in~1D.
\label{thm:1dq}
\end{theorem}
Below we restrict our attention to a particular set~$S$ of problem instances (see Eqs.~\eqref{eq:tuple} and~\eqref{eq:inputs}) where 
$\alpha=0$ iff $\beta=0$. In particular, according to the definition~\eqref{eq:Wbetaalphadef} of $W(\beta,\alpha)$, we can realize  such a gate by a pair of successively applied $2$-qubit  gates which are classically controlled by~$\alpha$. Thus we  may equivalently use a depth-$5$-circuit with controlled Clifford gates having at most~$2$~classical control  bits and $2$~target qubits for the 1D Magic Square Problem.

The fact that the quantum circuit for the 1D Magic Square Problem only requires geometrically local gates in 1D allows us to establish  a separation between noisy constant-depth quantum circuits with geometrically local gates in~3D, and general noise-free shallow classical circuits, as discussed in Section \ref{sec:3D}.

We will show that the 1D Magic Square Problem cannot be solved with high probability by a constant-depth classical circuit composed of bounded fan-in gates. Moreover, we show that any such circuit  must falter on a certain polynomial-sized subset of instances. In particular, we will be interested in the subset $S\subseteq \{0,1\}^{4n}$ of inputs  of the following form. Each member of the subset $S$ can be described by a tuple
\begin{equation}
(j,k,\alpha,\beta) \qquad \text{where} \quad 1\leq j<k\leq n \quad \text{ and } \alpha,\beta\in \{01,10, 11\}.
\label{eq:tuple}
\end{equation}
The associated input bits are given by
\begin{equation}
\malpha_i=\begin{cases} \alpha, & i=j\\
00, & i\neq j.
\end{cases}\qquad 
\mbeta_i=\begin{cases} \beta, & i=k\\
00, & i\neq k.
\end{cases}
\label{eq:inputs}
\end{equation}

Before establishing this hardness for classical constant-depth circuits, let us briefly comment on the way in which the described circuit $\ConeD_{z_{in}}$ from Fig.~\ref{fig:1d} achieves the  claimed quantum advantage. To this end, we show that the circuit~$\ConeD_{z_{in}}$ 
essentially executes the quantum winning strategy of the generalized magic square game. In other words, it uses quantum non-locality as a resource. 

In more detail, we argue that the output of the quantum circuit on an input from the set $S$ (i.e., of the form \eqref{eq:inputs}) obeys the winning condition of the generalized magic square game. This may be phrased as a necessary condition on pairs $(z_{in},z_{out})$ satisfying the 1D Magic Square Relation, as follows:
\begin{lemma}
Consider an instance $z_{in}=(\alpha_1,\ldots,\alpha_n,\beta_1,\ldots,\beta_n)$ of the 1D Magic Square Problem from the set~$S$, i.e., specified by~\eqref{eq:inputs} in terms of a tuple~$(j,k,\alpha,\beta)$  as in~\eqref{eq:tuple}. 
Let $z_{out}=(\mathbf{x}_1,\ldots,\mathbf{x}_n,\mathbf{y}_1,\ldots,\mathbf{y}_n)$ be any tuple such that $(z_{in},z_{out})$ satisfies the 1D Magic Square Relation. Then  the tuple
$(\malpha,\mbeta,\mathbf{x}_j,\mathbf{y}_k)$ satisfies  
\begin{align}
x^{\iota(\beta)}_jy^{\iota(\alpha)}_k&=f_{\alpha,\beta}(s,t,s',t')\qquad\textrm{ for all }\alpha,\beta\in \{01,10,11\} \label{eq:winconditionnecessary}
\end{align}
 with parameters
\begin{align}
s= \prod_{i=j}^{k-1} y_i^1 \qquad  t=\prod_{i=j}^{k-1} x_{i+1}^1 \qquad s'= \prod_{i=j}^{k-1} y_i^2  \qquad  t'=\prod_{i=j}^{k-1}x_{i+1}^2 \ . \label{eq:st}
\end{align}
Here we use $\mathbf{x}_i=(x_i^1,x_i^2)\in \{-1,+1\}^2$ (and similarly for $\mathbf{y}_i$) to denote the entries of~$\mathbf{x}_i$, and
set
\begin{align}
x^3_{j}=-x^1_jx^2_j\qquad\textrm{ and }\qquad y^3_k=y^1_ky^2_k\ .\label{eq:tripledefxyk}
\end{align}
\label{lem:stst}
\end{lemma}
We recognize  the winning condition  Eq.~\eqref{eq:win2}  for the generalized magic square game in the identity~\eqref{eq:winconditionnecessary}. 
In other words, when running the quantum circuit $\ConeD_{z_{in}}$ on an instance from~$S$, a winning output for the generalized magic square game is given by the measurement results from  qubits $p_{2j-1}q_{2j-1}$ and $p_{2k},q_{2k}$, whereby the particular generalized magic square game considered (i.e., the parameters $s,t,s',t'$)  is determined by the measurement outcomes of qubits between these sites, see Fig.~\ref{fig:2rows}.   We will later show that satisfying the necessary condition of Lemma~\ref{lem:stst} is infeasible for shallow classical circuits. The particular functional dependence captured by~\eqref{eq:st} will figure prominently in a reduction to the ``standard'' magic square game, whose classical value is~$8/9$ (see Lemma~\ref{lem:modifiedgamemain} below).

\begin{proof}
Consider the circuit from Fig.~\ref{fig:1d} applied to the input Eq.~\eqref{eq:inputs}.  We begin by describing how the input-output statistics are related to the generalized magic square game. The role of the gate $W(00,00)$ here is to perform \textit{entanglement swapping measurements}. 
Entanglement swapping refers to the fact that if a four-partite system is in a product~$\Phi_{AC_1}\otimes \Phi_{C_2B}$ of two Bell states, and the (Bell) observables~$X_1X_2$ and $Z_1Z_2$ are measured on systems~$C_1C_2$, then the outcomes $s,t\in \{-1,+1\}$ respectively are uniformly distributed and the associated postmeasurement state on $AB$ is (up to a sign) the Bell state $|\Phi_{s,t}\rangle$ defined in Eq.~\eqref{eq:varphistzerozero}.

Now imagine organizing the $4n$ data qubits in the circuit into two rows of $2n$ qubits each, as shown in Fig. \ref{fig:2rows}. Here the first row from left to right contains the qubits labeled $p_1,p_2,p_3,\ldots, p_{2n}$ in Fig.~\ref{fig:1d} and the bottom row contains the qubits $q_1,q_2,\ldots, q_{2n}$. It may be helpful to imagine that Alice holds the two data qubits $p_{2j-1},q_{2j-1}$ and receives input $\alpha\in \{01,10,11\}$ and Bob holds the two qubits $p_{2k},q_{2k}$ and receives input $\beta\in \{01,10,11\}$.     

All $4n$ qubits begin in the state $|0\rangle$. Then a layer of Hadamard gates and a layer of CNOT gates is applied. These gates prepare entangled states $|\Phi\rangle$ between adjacent pairs of qubits as shown in Fig.~\ref{fig:2rows}. The remainder of the circuit consists of: the $W(00,00)$ gates between pairs of adjacent qubits, Alice's Clifford gate $U(\malpha_j)$ on her qubits $p_{2j-1},q_{2j-1}$ and Bob's Clifford gate $V(\mbeta_k)$ acting on his qubits $p_{2k},q_{2k}$, and measurement of all data qubits. Since the remaining gates in the circuit commute, we can imagine that they are performed in two steps. In the first step, we apply the $W(00,00)$ gates and measure all data qubits \textit{except} $p_{2j-1},q_{2j-1}, p_{2k},q_{2k}$. This step performs entanglement swapping measurement between adjacent pairs of qubits in Fig.~\ref{fig:2rows} except those held by Alice or Bob. In particular, entanglement swapping measurements are performed on all qubit pairs
\begin{equation}
(p_{2i}, p_{2i+1}) \quad \text{and } \quad (q_{2i}, q_{2i+1}) \qquad i\in \{1,2,\ldots, n-1\}\setminus\{j-1,k\},
\end{equation}
resulting in uniformly random measurement outcomes 
\begin{equation}
s_i, t_i\in \{-1,+1\} \ \quad \text{and} \quad s'_i, t'_i\in \{-1,+1\} \quad  \qquad i\in \{1,2,\ldots, n-1\}\setminus\{j-1,k\}
\end{equation}
respectively. 
Here each of the outcomes 
\begin{align}
\begin{matrix}
s_i&=&y_i^1\\
t_i&=& x_{i+1}^1
\end{matrix}\qquad\textrm{ and }\qquad 
\begin{matrix}s'_i&=&y_i^2\\
t'_i&=&x_{i+1}^2 
\end{matrix}
\end{align}
is a $\pm 1$ valued variable determined by one of the output bits of the circuit in Fig.~\ref{fig:1d}. The crucial point is that after this step, Alice and Bob's 4 qubits~$(p_{2j-1},p_{2k},q_{2j-1},q_{2k})$ are in the state $|\Phi_{s,t}\rangle\otimes |\Phi_{s',t'}\rangle$ with  $s, s', t, t'$ given in 
\begin{align}
s= \prod_{i=j}^{k-1} s_i \qquad  t=\prod_{i=j}^{k-1} t_i \qquad s'= \prod_{i=j}^{k-1} s_i'  \qquad  t'=\prod_{i=j}^{k-1}t_i'
\end{align}
that is, by Eq.~\eqref{eq:st}.

In the second step, starting from the postmeasurement state $|\Phi_{s,t}\rangle\otimes |\Phi_{s',t'}\rangle$, Alice and Bob apply $U(\malpha_j)$ and $V(\mbeta_j)$ to their qubits and then measure to produce outputs $\mathbf{x}_j, \mathbf{y}_k$, respectively. In other words, they play the winning strategy for the generalized magic square game
 with parameters~$(s,t,s',t')$. This implies the claim. 
\end{proof}

\subsection{Hardness of the 1D Magic Square Problem for constant-depth classical circuits}
\label{sec:hardnessoneDsquare}
We use the following standard notions: A classical circuit~$\cC$
computes a function $F:\{0,1\}^N\rightarrow\{0,1\}^M$ (sometimes we use $\{-1,+1\}$ instead of $\{0,1\}$). It is given by a directed acyclic graph. There are $N$ vertices with in-degree~$0$ corresponding to the input bits and $M$ vertices with out-degree $0$ (corresponding to outputs). Every vertex with in-degree $k>0$  and out-degree $L$ is associated with a {\em gate}, that is,  a boolean function  $f:\{0,1\}^k\rightarrow\{0,1\}$. The output when applying a gate is copied to all~$L$ outgoing edges. Finally, the depth of the circuit, denoted {\em $\mathrm{depth}(\cC)$}, is the maximal number of gates along a path from an input to an output.  Here we consider classical circuits with the property that all gates have constant {\em fan-in}, i.e., the associated in-degree is at most some constant~$K=O(1)$  independent of the problem size. 
When solving a computational problem, we will sometimes only be interested in a subset of the $N$ input bits (encoding a problem instance) and a subset of the $M$~output bits (providing the computed solution).  The remaining in- and output-bits play no role in the following arguments. To model  probabilistic circuits, we proceed similarly: here some subset of the $N$~input bits may be randomly distributed according to an arbitrary distribution independent of the problem instance. 

Recall that $S\subseteq \{0,1\}^{4n}$ is the set of problem instances (inputs) of the 1D Magic Square Problem defined by~\eqref{eq:inputs} in terms of 
tuples $(j,k,\alpha,\beta)$. We prove the following circuit depth lower bound for classical circuits 
solving the 1D Magic Square Problem with constant probability. 
\begin{theorem} Suppose that $\mathcal{C}$ is a classical probabilistic circuit composed of gates of fan-in at most $K$ which, given a randomly chosen input from the set $S$, produces a solution to the corresponding instance of the 1D Magic Square Problem with probability at least $9/10$. Then 
\begin{equation}
\mathrm{depth}(\mathcal{C})\geq \frac{\log(0.00001n)}{2\log(K)}.
\label{eq:depthbound}
\end{equation}
\label{thm:lowerbound}
\end{theorem}
In the above, the probability which is at least $9/10$ is taken over all the randomness including both the choice of the input from $S$ as well as the randomness in the classical probabilistic circuit~$\mathcal{C}$.
\begin{figure}
\centering
\begin{tikzpicture}[xscale=0.7, yscale=0.75]
\begin{scope}[shift={(-7.5,0)}]
\draw (0,0)--(3,0);
\draw (0,1)--(3,1);
\draw (1.5, 1.45) node {$|\Phi\rangle$};
\draw (1.5, -0.55) node {$|\Phi\rangle$};
\draw (0,0) node[circle, fill=black, scale=0.5]{};
\draw (3,0) node[circle, fill=black, scale=0.5]{};

\draw (0,1) node[circle, fill=black, scale=0.5]{};
\draw (3,1) node[circle, fill=black, scale=0.5]{};

\draw (0,-0.4) node {$q_1$};
\draw (0,0.6) node {$p_1$};
\draw (3,-0.4) node {$q_2$};
\draw (3,0.6) node {$p_2$};

\draw (3.5,0) node {$\ldots$};

\draw (3.5,1) node {$\ldots$};
\end{scope}

\begin{scope}[shift={(14.5,0)}]
\draw (0,0)--(3,0);
\draw (0,1)--(3,1);
\draw (1.5, 1.45) node {$|\Phi\rangle$};
\draw (1.5, -0.55) node {$|\Phi\rangle$};
\draw (0,0) node[circle, fill=black, scale=0.5]{};
\draw (3,0) node[circle, fill=black, scale=0.5]{};

\draw (0,1) node[circle, fill=black, scale=0.5]{};
\draw (3,1) node[circle, fill=black, scale=0.5]{};

\draw (3,-0.4) node {$q_{2n}$};
\draw (3,0.6) node {$p_{2n}$};

\end{scope}

\draw (0,0)--(3,0);
\draw (0,1)--(3,1);
\draw (1.5, 1.45) node {$|\Phi\rangle$};
\draw (1.5, -0.55) node {$|\Phi\rangle$};
\draw (0,0) node[circle, fill=black, scale=0.5]{};
\draw (3,0) node[circle, fill=black, scale=0.5]{};

\draw (0.35,-0.4) node {$q_{2j-1}$};
\draw (0.35,0.6) node {$p_{2j-1}$};

\draw (0,1) node[circle, fill=black, scale=0.5]{};
\draw (3,1) node[circle, fill=black, scale=0.5]{};

\begin{scope}[shift={(3.5,0)}]
\draw (0,0)--(3,0);
\draw (0,1)--(3,1);
\draw (1.5, 1.45) node {$|\Phi\rangle$};
\draw (1.5, -0.55) node {$|\Phi\rangle$};
\draw (0,0) node[circle, fill=black, scale=0.5]{};
\draw (3,0) node[circle, fill=black, scale=0.5]{};

\draw (0,1) node[circle, fill=black, scale=0.5]{};
\draw (3,1) node[circle, fill=black, scale=0.5]{};
\end{scope}
\begin{scope}[shift={(10.5,0)}]
\draw (0,0)--(3,0);
\draw (0,1)--(3,1);
\draw (1.5, 1.45) node {$|\Phi\rangle$};
\draw (1.5, -0.55) node {$|\Phi\rangle$};
\draw (0,0) node[circle, fill=black, scale=0.5]{};
\draw (3,0) node[circle, fill=black, scale=0.5]{};

\draw (0,1) node[circle, fill=black, scale=0.5]{};
\draw (3,1) node[circle, fill=black, scale=0.5]{};

\draw (3.5,0) node {$\ldots$};

\draw (3.5,1) node {$\ldots$};
\end{scope}

\begin{scope}[shift={(-3.5,0)}]
\draw (0,0)--(3,0);
\draw (0,1)--(3,1);
\draw (1.5, 1.45) node {$|\Phi\rangle$};
\draw (1.5, -0.55) node {$|\Phi\rangle$};
\draw (0,0) node[circle, fill=black, scale=0.5]{};
\draw (3,0) node[circle, fill=black, scale=0.5]{};

\draw (0,1) node[circle, fill=black, scale=0.5]{};
\draw (3,1) node[circle, fill=black, scale=0.5]{};
\end{scope}

\draw (1.5, -1.3) node {};
\begin{scope}[shift={(7,0)}]
\draw (0.5,0)--(3,0);
\draw (0.5,1)--(3,1);
\draw (1.5, 1.45) node {$|\Phi\rangle$};
\draw (1.5, -0.55) node {$|\Phi\rangle$};
\draw (0.5,0) node[circle, fill=black, scale=0.5]{};
\draw (3,0) node[circle, fill=black, scale=0.5]{};

\draw (0,0) node {$\ldots$};

\draw (0,1) node {$\ldots$};

\draw (2.9,-0.4) node {$q_{2k}$};
\draw (2.9,0.6) node {$p_{2k}$};

\draw (0.5,1) node[circle, fill=black, scale=0.5]{};
\draw (3,1) node[circle, fill=black, scale=0.5]{};
\end{scope}
\draw (0.4, 1.8) node {\textbf{Alice}};
\draw[thick, dotted, fill=blue, opacity=0.1] (-0.35,-1) rectangle (1.1,2.2);
\draw [line width=0.5mm, ->] (0.4, 3)--(0.4,2.4);
\draw (0.4, 3.2) node {$\malpha_j=\alpha$};
\draw[line width=0.5mm, ->] (0.4, -1.2)--(0.4, -1.8);
\draw (0.4, -2.25) node {$\mathbf{x}_j$};
\begin{scope}[shift={(7, 0)}]
\draw (2.6, 1.8) node {\textbf{Bob}};
\draw[thick, dotted, fill=red, opacity=0.1] (1.9,-1) rectangle (3.35,2.2);
\draw [line width=0.5mm, ->] (2.5, 3)--(2.5,2.4);
\draw (2.5, 3.2) node {$\mbeta_k=\beta$};
\draw[line width=0.5mm, ->] (2.5, -1.2)--(2.5, -1.8);
\draw (2.5, -2.25) node {$\mathbf{y}_k$};
\end{scope}

\end{tikzpicture}
\caption{The measurement statistics of the circuit in Fig.~\ref{fig:1d}, when applied to an input of the form Eq.~\eqref{eq:inputs}, are related to those of the generalized magic square game played by Alice holding qubits $p_{2j-1},q_{2j-1}$ and Bob holding qubits $p_{2k},q_{2k}$. 
The measurement outcomes of the qubits located between Alice and Bob determine the particular generalized magic square game, i.e., the parameters $s,t,s',t'\in \{-1,+1\}$. 
\label{fig:2rows}}
\end{figure}
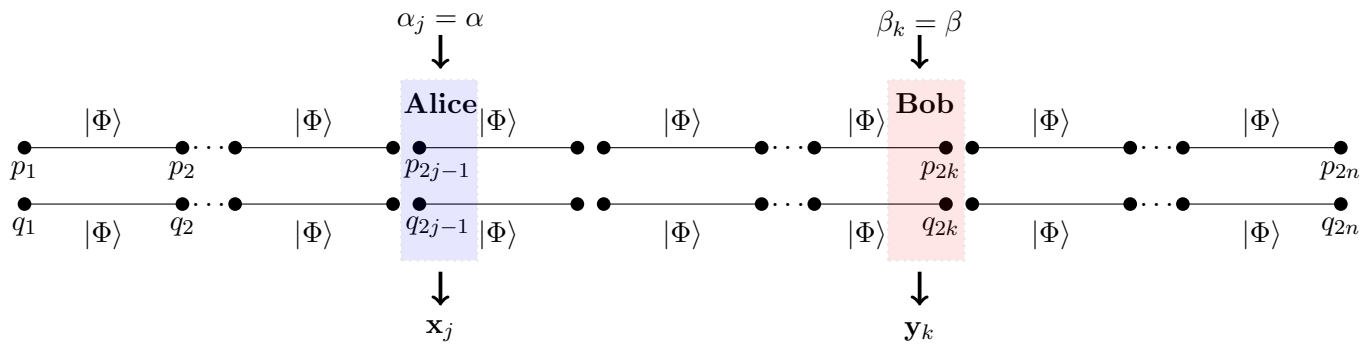
The proof of Theorem \ref{thm:lowerbound} is comprised of two main steps, which we will prove separately below. 
\begin{enumerate}
\item We show that with high probability the lightcones of the inputs and outputs of the magic square game do not intersect in a certain way (Lemma~\ref{lem:lightcones2}).

\item Finally, we show that if the lightcones do not intersect in a certain way, then a classical circuit can succeed with probability at most $8/9$ at solving the 1D Magic Square Problem (Lemma~\ref{lem:modifiedgamemain}).
\end{enumerate}


We shall consider circuits in which some of the input bits are random; they can be drawn from any probability distribution which is independent of the problem instance. Since we are interested in upper bounding the average success probability when solving a computational problem, by convexity it suffices to assume that the circuit is in fact deterministic and that the input bits are assigned fixed values independent of the instance which maximize the average success probability.  Thus our results apply to  probabilistic circuits, but we will not explicitly deal with randomness in the following arguments.

\subsubsection{Signaling properties of constant-depth circuits: lightcones}
The connectivity of the acyclic graph underlying a classical circuit~$\cC$ determines its capability to signal, ultimately restricting the set of functions $F:\{0,1\}^N\rightarrow\ \{0,1\}^M$ it can compute and thus the range of problems a given circuit can solve.  In the following, we will use variables  $x_j$ (with $1\leq j\leq N$) and $z_k$ (with $1\leq k\leq M$) 
to denote individual input- and output-bits of a circuit~$\cC$ computing $z=F(x)$. 
If there is a string $x\in \{0,1\}^N$ such that the value of the $k$-th bit of the output $F(x)$ changes when flipping the $j$-th bit of $x$, we say that $\{x_j,z_k\}$ are correlated.
 The following definitions are convenient:
 \begin{definition}
 Given every input bit $x_j$, the {\em forward lightcone} $\flightcone_{\cC}(x_j)$ is the set of output bits~$z_k$ such that $\{x_j,z_k\}$ are correlated.  Similarly, for every output bit $z_k$, the {\em backward lightcone} $\blightcone_{\cC}(z_k)$ is the set of input bits~$x_j$ such that $\{x_j,z_k\}$ are correlated. More generally, for any set $I$ of input bits and any set~$O$ of output bits, we set
\begin{align}
\flightcone_{\cC}(I)=\bigcup_{x\in I}\flightcone_{\cC}(x)\qquad\textrm{and }\qquad
\blightcone_{\cC}(O)=\bigcup_{z\in O}\blightcone_{\cC}(z)\ .
\end{align}
 \end{definition}

A circuit~$\cC$ of depth~$D$ and gates with fan-in upper bounded by $K=O(1)$ has the following crucial property:
the backward lightcone of each output bit~$z$ satisfies
\begin{equation}
|\blightcone_{\mathcal{C}}(z)| \leq K^D\ . \label{eq:lcone}
\end{equation}

Eq.~\eqref{eq:lcone} also imposes limitations on the size of ``most'' forward lightcones, as well as the restrictions on 
the intersections of forward lightcones associated with distinct input bits. We will use the following probabilistic statements:
\begin{lemma}
Let $\mathcal{C}$ be a classical circuit consisting of gates with fan-in upper bounded by~$K$, depth~$D$, and $M$ output bits.
Then the following holds:
\begin{enumerate}[(i)]
\item\label{it:claimonelightcone}
Let $O$ be a fixed subset of output bits and suppose $I$ is a randomly chosen subset of input bits such that 
\begin{align}
\Pr[v\in I]&\leq q
\end{align}
for every input bit $v$. Then
\begin{align}
\Pr[O\cap \flightcone_{\cC}(I)\neq \emptyset ]&\leq q|O| 2^{|O|} K^{D}\ .
\end{align}

\item\label{it:claimtwolightcone}
 Suppose $I$ and $J$ are randomly chosen disjoint subsets of input bits such that, for any two  input bits $v,w$ we have
\begin{equation}
\mathrm{Pr}[v\in I \text{ and } w\in J]\leq p.
\label{eq:pupperbound}
\end{equation}
Then 
\[
\mathrm{Pr}\left[\flightcone_{\cC}(I)\cap \flightcone_{\cC}(J)\neq \emptyset\right]\leq pM K^{2D}.
\]
\end{enumerate}
\label{lem:lightcones}
\end{lemma}
\begin{proof}
We have
\begin{align}
\Pr[O\cap \flightcone_{\cC}(I)\neq \emptyset ]&=\sum_{\substack{P\subseteq O\\
P\neq \emptyset}} \Pr\left[O\cap \flightcone_{\cC}(I)=P\right]\\
&\leq \sum_{\substack{P\subseteq O\\
P\neq \emptyset}} \Pr\left[I\cap \blightcone_{\cC}(P)\neq\emptyset\right]\\
&\leq \sum_{\substack{P\subseteq O\\
P\neq \emptyset}} \sum_{v\in \blightcone_{\cC}(P)}\Pr\left[v\in I\right]\\
&\leq 2^{|O|} |O|K^D q\ .
\end{align}
Here we have used~\eqref{eq:lcone}, which implies that  $|\blightcone_{\cC}(P)|\leq |P| K^D\leq |O| K^D$ 
for $P\subseteq O$. This shows the claim~\eqref{it:claimonelightcone}.

For the proof of claim~\eqref{it:claimtwolightcone}, let $V_{out}$ be the set of all output bits, so that $|V_{out}|=M$. A union bound gives
\begin{align}
\mathrm{Pr}\left[\flightcone_{\cC}(I)\cap \flightcone_{\cC}(J)\neq \emptyset\right] & \leq \sum_{z\in V_{out}} \mathrm{Pr}\left[I\cap  \blightcone_{\cC}(z)\neq \emptyset \textrm{ and } J \cap  \blightcone_{\cC}(z)\neq \emptyset\right] \\
& \leq  \sum_{z\in V_{out}} \; \sum_{v,w\in  \blightcone_{\cC}(z)}\mathrm{Pr}[v\in I \text{ and } w\in J]\\\
& \leq  \sum_{z\in V_{out}} \sum_{v,w\in  \blightcone_{\cC}(z)}  p\\
&\leq pMK^{2D}.
\end{align}
where in the last line we used Eq.~\eqref{eq:lcone}.
\end{proof}

\subsubsection{Proof of hardness using lightcones}
Let us now give and prove the formal statements corresponding to the steps outlined after the statement of Theorem~\ref{thm:lowerbound}. We consider a classical circuit~$\cC$ for the 1D Magic Square Problem. Such a circuit has  
inputs
 $(\alpha_1,\ldots,\alpha_n,\beta_1,\ldots,\beta_n)\in (\{0,1\}^2)^{2n}\equiv \{0,1\}^{4n}$ and outputs
 $(\mathbf{x}_1,\ldots,\mathbf{x}_n,\mathbf{y}_1,\ldots,\mathbf{y}_n)\in ((\{-1,+1\})^2)^{2n}\equiv \{-1,+1\}^{4n}$. 
 We first define an event which occurs with high probability and ensures that the lightcones of certain input/output bits do not intersect. 
\begin{lemma}
Consider a classical probabilistic circuit $\cC$ of depth $D$, with $4n$ output bits, and composed of gates of fan-in at most $K$.
Define the event $E_{\cC} \subset S$ in which the input parameters $1\leq j<k\leq n$ in Eq.~\eqref{eq:inputs} satisfy
\begin{align}
	\flightcone_{\cC}(\malpha_j)\cap \flightcone_{\cC}(\mbeta_k) =\emptyset \quad \text{ and }\quad  \mathbf{y}_k \cap \flightcone_{\cC}(\malpha_j)= \emptyset \quad \text{and} \quad \mathbf{x}_j \cap \flightcone_{\cC}(\mbeta_k)=\emptyset\ .\label{eq:eventdefinitionproperty}
\end{align}
Under a uniform choice of input from $S$, the event $E_{\cC}$ occurs with probability $\mathrm{Pr}[E_{\cC}] \geq 1-\frac{80K^{2D}}{n}$.
\label{lem:lightcones2}
\end{lemma}
\begin{proof}
Consider a random input from the set $S$. In other words, suppose $(j,k,\alpha,\beta)$ is a uniformly random tuple satisfying Eq.~\eqref{eq:tuple} and consider the associated input Eq.~\eqref{eq:inputs}. We claim that with high probability the lightcones of the input bits $\malpha_j \in \{0,1\}^2$ and $\mbeta_k\in \{0,1\}^2$ do not intersect. In particular, we may apply Lemma \ref{lem:lightcones} with subsets $I=\malpha_j$ and $J=\mbeta_k$ each containing two input bits, for $j<k$ chosen uniformly at random. Note that by definition, these sets are disjoint. 
Fix any two (distinct)  input bits~$v,w$. Since the sets $\{\alpha_j\}_j\cup \{\beta_k\}_k$ form a disjoint partition of the set of input bits, it is clear that~$\mathrm{Pr}\left[v\in \malpha_j \text{ and } w\in \mbeta_k\right]$ (for randomly chosen $j<k$) vanishes unless $v\in \malpha_{j_*}$ and $w\in \mbeta_{k_*}$ for some $j_*<k_*$. In the latter case, we have
\begin{align}
\mathrm{Pr}\left[v\in \malpha_j \text{ and } w\in \mbeta_k\right]&= 
\Pr\left[(j,k)=(j_*,k_*)\right]=\frac{2}{n(n-1)}
\end{align}
since there are ${n}\choose{2}$ pairs $1\leq j<k\leq n$. Applying part~\eqref{it:claimtwolightcone} of the Lemma with $p=2/(n(n-1))$ and $M=4n$ gives
\begin{equation}
\mathrm{Pr}\left[ \flightcone_{\cC}(\malpha_j)\cap \flightcone_{\cC}(\mbeta_k) \neq \emptyset\right]\leq \frac{8K^{2D}}{n-1}.
\label{eq:intersectbound}
\end{equation}

Consider a set $O=\mathbf{x}_{j}$ consisting of two output bits for a fixed $j\in \{1,\ldots,n\}$. Let $I=\beta_{k}$ for a uniformly chosen~$k\in \{1,\ldots,n\}$. Applying~\eqref{it:claimonelightcone} of Lemma~\ref{lem:lightcones}
to the random subset $I$ (consisting of two input bits) with  $q=1/n$ gives 
\begin{align}
\Pr_{k}\left[\mathbf{x}_{j}\cap \flightcone_{\cC}(\beta_{k})\neq \emptyset \right]\leq \frac{8 K^D}{n}\ .
\end{align}
Since this holds for any $1\leq j\leq n$, the probability that~$\mathbf{x}_{j}\cap \flightcone_{\cC}(\beta_{k})\neq \emptyset$ for uniformly chosen $(j,k)\in \{1,\ldots,n\}^2$ is upper bounded by~$8K^D/n$. This implies  that choosing $(j,k)$ uniformly at random subject to $j<k$, we have
that
\begin{align}
\Pr\left[\mathbf{x}_{j}\cap \flightcone_{\cC}(\beta_{k})\neq \emptyset \right]\leq \frac{32 K^D}{n}\ \label{eq:eventpart2}
\end{align}
because the number~$\binom{n}{2}$ of such pairs satisfies $n^2/\binom{n}{2}\leq 4$. By the same reasoning, we have
\begin{equation}
\mathrm{Pr}\left[ \mathbf{y}_k \cap \flightcone_{\cC}(\malpha_j)\neq \emptyset\right]\leq \frac{32K^D}{n}.
\label{eq:eventpart3}
\end{equation}
when $(j,k)$  is part of a uniformly chosen tuple~$(j,k,\alpha,\beta)$ from $S$. 

Applying the union bound and Eqs.~\eqref{eq:intersectbound}, \eqref{eq:eventpart2}, \eqref{eq:eventpart3} we get $\mathrm{Pr}[E_{\cC}] \geq 1-\frac{8K^{2D}}{n-1}-\frac{64K^{D}}{n}\geq 1-\frac{80K^{2D}}{n}.$
\end{proof}

We remark that for a uniform choice of input~$(j,k,\alpha,\beta)$ from~$S$, the marginal distribution of $(\alpha,\beta)$ conditioned on the event $E_{\cC}$ is uniform on the set $\{01,10,11\}^2$. This is because the event~$E_{\cC}$ only depends on $(j,k)$ by definition, and the uniform distribution over $S$ is of product form.  
 In the following, we consider such fixed values of $(j,k)$ and establish an upper bound on the probability that the output of the classical circuit~$\cC$ is a valid solution of the 1D~Magic Square Problem. To this end, we show that the necessary condition of Lemma~\ref{lem:stst} is only satisfied with probability at most~$8/9$ for uniformly chosen~$(\alpha,\beta)$:  
 \begin{lemma}\label{lem:modifiedgamemain}
Consider a classical circuit $\cC$ as in Lemma~\ref{lem:lightcones2}, and let $1\leq j<k\leq n$ 
be such that the event~$E_{\cC}$ occurs. Suppose $(\alpha,\beta)\in \{01,10,11\}^2$ are chosen uniformly at random,  and the input~$z_{in}=(\alpha_1,\ldots,\alpha_n,\beta_1,\ldots,\beta_n)$  of~$S$  specified by $(j,k,\alpha,\beta)$ is 
fed to~$\cC$. Then the average probability that $\cC$ outputs~$z_{out}=(\mathbf{x}_1,\ldots,\mathbf{x}_n,\mathbf{y}_1,\ldots,\mathbf{y}_n)$ such that $(z_{in},z_{out})$ satisfy the condition of Lemma~\ref{lem:stst} is at most~$8/9$.
\end{lemma}
\begin{proof}
For $(j,k)$ as described, the definition of~$E_{\cC}$ implies that  any output bit of the circuit 
is either independent of both $\alpha$ and $\beta$, or depends  on $\malpha_j=\alpha$ or $\mbeta_k=\beta$ only (but not both). To check the condition
of Lemma~\ref{lem:stst}, we should focus on the output bits $\mathbf{x}_\ell=(x_\ell^1,x_\ell^2)$ and  $\mathbf{y}_{\ell}=(y_\ell^1,y_\ell^2)$ for $\ell=j,\ldots,k$; these define
the variables  $s,t,s',t'$ as described in Eq.~\eqref{eq:st}, as well as the triples $(x^1,x^2,x^3)\equiv (x^1_j,x^2_j,x^3_j)$ and $(y^1,y^2,y^3)\equiv (y^1_k,y^2_k,y^3_k)$ by~\eqref{eq:tripledefxyk}.

Since no output bit can depend on both $\alpha$ and $\beta$, and the parameters $s,s',t,t'$ are computed by taking products of output bits of $\cC$, their dependence on $(\alpha,\beta)$ has  the functional form (suppressing their dependence on all other inputs):
\begin{align}
\begin{matrix}
s(\alpha,\beta)&=&s_A(\alpha)s_B(\beta)\\
t(\alpha,\beta)&=&t_A(\alpha)t_B(\beta)\\
s'(\alpha,\beta)&=&s_A'(\alpha)s_B'(\beta)\\
t'(\alpha,\beta)&=&t_A'(\alpha)t_B'(\beta)
\end{matrix}\ \qquad\textrm{ for all }\alpha,\beta\in \{01,10,11\}\ .
\label{eq:sjtjfcts}
\end{align}
That is, the circuit~$\cC$ gives rise to certain functions $s_A,s_B,t_A,t_B,s_A',s_B',t_A',t_B':\{01,10,11\}\rightarrow\{-1,+1\}$ such that~\eqref{eq:sjtjfcts} is satisfied. Moreover, since the event $E_{\cC}$ occurs we also have the functional (in)dependence
\begin{align}
x^1&=x^1(\alpha), \qquad x^2=x^2(\alpha), \qquad y^1=y^1(\beta), \qquad y^2=y^2(\beta) \ ,
\label{eq:xy}
\end{align}
that is, the circuit defines functions $x^1,x^2,x^3,y^1,y^2,y^3:\{01,10,11\}\rightarrow\{-1,+1\}$ where
\begin{align}
x^3(\alpha)&\equiv -x^1(\alpha)x^2(\alpha) \qquad\textrm{ and }\qquad y^3(\alpha)=y^1(\alpha)y^2(\alpha)\qquad\textrm{ for  }\alpha\in \{01,10,11\}\ .\label{eq:x3alphay3alpahdef}
\end{align}
Now note that under the restrictions expressed by~\eqref{eq:sjtjfcts} the 
necessary condition~\eqref{eq:winconditionnecessary}  takes on the form
\begin{align}
  x^{\iota(\beta)}y^{\iota(\alpha)} = f_{\alpha,\beta} ( s_{A}(\alpha), t_A(\alpha), s'_A(\alpha), t'_A(\alpha))  \cdot f_{\alpha,\beta} (s_B(\beta), t_B(\beta), s'_B(\beta), t'_B(\beta)) \,. \label{eq:necessaryconditionecevent}
\end{align}
where we used the factorization property~\eqref{eq:factorizationproperty} of the functions $f_{\alpha,\beta}$.

Suppose for the sake of contradiction that the outputs produced by the circuit on a random input $(\alpha,\beta)$ satisfy the condition~\eqref{eq:necessaryconditionecevent} with probability greater than~$8/9$. 
Using the functions introduced above, let us define the functions 
$\tilde{x}^i,\tilde{y}^j:\{01,10,11\}\rightarrow \{-1,+1\}$
for $i,j=1,2,3$ by 
\begin{align}
\begin{matrix}
\tilde{x}^i(\alpha)&=&x^i (\alpha) f_{\alpha,\iota^{-1}(i)} \big( s_{A}(\alpha), t_A(\alpha), s'_A(\alpha), t'_A(\alpha) \big)\qquad&\textrm{ for all }&\alpha\in \{01,10,11\}\ \ \\
\tilde{y}^j(\beta) &=&y^j (\beta) f_{\iota^{-1}(j),\beta} \big( s_{B}(\beta), t_B(\beta), s'_B(\beta), t'_B(\beta) \big)\qquad&\textrm{ for all }&\beta\in \{01,10,11\}\ .
\end{matrix}
\end{align}
We note that since $(\tilde{x}^1(\alpha),\tilde{x}^2(\alpha),\tilde{x}^3(\alpha))$ and 
$(\tilde{y}^1(\beta),\tilde{y}^2(\beta),\tilde{y}^3(\beta))$ can be computed from~$\alpha$ respectively~$\beta$ only, these functions constitute a classical strategy for the magic square game. 

 We claim that this strategy satisfies  the winning condition for the magic square game
 with probability exceeding~$8/9$. To verify this, note that  because of property
  Eq.~\eqref{eq:stprod} of the functions~$f_{\alpha,\beta}$ and Definition~\eqref{eq:x3alphay3alpahdef}  we have 
\begin{align}
 \prod_{i=1}^3 \tilde{x}^i(\alpha) = \prod_{i=1}^3 x^i(\alpha) \prod_{i=1}^3 f_{\alpha,\iota^{-1}(i)} \big( s_{A}(\alpha), t_A(\alpha), s'_A(\alpha), t'_A(\alpha) \big)=  \prod_{i=1}^3 x^i(a) = -1\ ,
\end{align}
and similarly $\prod_{j=1}^3 \tilde{y}^j(b) = 1$. Thus the parity constraints in the magic square game are satisfied with probability one.  On the other hand, the equality constraint $\tilde{x}^{\iota(\beta)}=\tilde{y}^{\iota(\alpha)}$ of the game, that is, $\tilde{x}^{\iota(\beta)}(\alpha)=\tilde{y}^{\iota(\alpha)}(\beta)$, is equivalent to the condition~\eqref{eq:necessaryconditionecevent} by definition of the strategy, hence satisfied with probability greater than~$8/9$ by assumption.

Since we know that the magic square game cannot be won using a classical strategy with probability exceeding $8/9$, this contradicts our assumption and concludes the proof.
\end{proof}

We can now combine the above three lemmas to prove the theorem.

\begin{proof}[Proof of Theorem \ref{thm:lowerbound}]
Because Lemma~\ref{lem:modifiedgamemain} holds for all pairs $(j,k)$ constituting the event $E_{\cC}$ (cf.~\eqref{eq:eventdefinitionproperty}), and because of
the necessity of satisfying the generalized Magic Square Relation when solving the 1D Magic Square Problem (see Lemma~\ref{lem:stst}), we conclude that the success probability of such a circuit~$\cC$ conditioned on the event~$E_{\cC}$ is bounded by
\begin{align}
\Pr\left[\cC\textrm{ succeeds}\ |\ E_{\cC}\right]\leq 8/9\ .
\end{align}
 Using this fact and Lemma~\ref{lem:lightcones2} we get
\begin{align}
\mathrm{Pr}\left[\cC \text{ succeeds}\right]& \leq \mathrm{Pr}\left[\cC \text{ succeeds}\big| E_{\cC}\right]+\left(1-\mathrm{Pr}[E_{\cC}]\right) \leq \frac{8}{9}+\frac{80K^{2D}}{n}.
\end{align}
Now suppose that the circuit succeeds with probability at least $9/10$ as stated in the theorem. Bounding the right hand side in this way and rearranging gives
\[
K^{2D}\geq \left(9/10-8/9\right)\frac{n}{80}\geq 0.00001n \,,
\]
and taking logarithms we arrive at the bound Eq.~\eqref{eq:depthbound}.
\end{proof}

\section{Noisy quantum circuits versus noiseless classical circuits\label{sec:noisyadvantageconstruction}}
So far, we have  considered the case where our quantum circuit is noise-free. We note that the quantum circuit for the 1D Magic Square Problem presented above is not fault-tolerant. In particular, in the limit of large problem sizes, it does not permit to observe a quantum advantage under any physically reasonable noise model:  for a constant error-rate per qubit, the probability  of producing an output satisfying the relation quickly falls  below the classical threshold value of~$8/9$  in this limit. This can be seen for example from the necessary condition Eq.~\eqref{eq:tripledefxyk} in Lemma~\ref{lem:stst}: for typical problem instances, this involves the parity of a number of output bits which scales linearly in~$n$. We note that the quantum circuit for the Hidden Linear Function Problem in~\cite{bragokoe18} suffers from the same issue in the presence of noise. 

In this section, we address this problem and establish a separation between noisy constant-depth quantum circuits and ideal constant-depth classical circuits. To this end, we construct new relation problems: these incorporate fault-tolerance mechanisms allowing for a solution by noisy constant-depth quantum circuits. We emphasize that these constructions do not proceed by ``amplification'' of the classical threshold (or the ``soundness'') towards~$1$ (as considered in~\cite{coudronetal18,legall18}). Indeed, as our relational problems involve an extensive number of output bits such amplification techniques do not appear to be suitable for this purpose. Rather,  we use quantum error-correcting codes.

Specifically, we show the following: given a certain relation problem providing a separation between constant-depth quantum circuits and classical circuits of sublogarithmic depth, we present a new relation problem which (a) can be solved with high probability with constant-depth noisy quantum circuit, and (b) which is still hard for classical circuits of a certain depth. The main underlying idea is that a quantum error-correction procedure can be folded into the relation.

This section is structured as follows: In Section~\ref{sec:noise}, we present the  details of the noise model we consider. In Section~\ref{subs:codes}, we list the required properties of quantum error-correcting codes used in our construction. In Section~\ref{sec:noisetolerantrelation}, we give the construction of a noise-tolerant relation, and give an (ideal) quantum circuit   solving the relation. In Section~\ref{sec:noisetoleranceqcircuit}, we prove that  
this circuit is noise-tolerant: it still produces a valid solution with high probability under local stochastic noise. In Section~\ref{sec:classicalhardnessnoisetolerant}, we argue that the noise-tolerant relation retains its hardness for classical circuits. Finally, in Section~\ref{sec:ftquantumadvantage}, we  instantiate this construction using the 1D Magic Square Problem, obtaining a quantum advantage using noisy constant-depth quantum circuits.

\subsection{The local stochastic quantum noise model}
\label{sec:noise}
Noise in a quantum computation can affect initial states, the execution of gates (which may include identities or ``wait locations'') and  measurement operations. Here we adopt a standard model to describe noise occuring during the execution of a quantum circuit. We refer to it simply as {\em local stochastic noise}, following the recent work~\cite{fawzi2018constant}.  The model has also been referred to as the {\em simplified model}, and is related to a more general {\em basic model} of fault-tolerance  in~\cite[Section 7]{Gottesman2013}.  In the simplified model, errors  occur in each time step on the physical qubits, and additionally,  the results of measurements can be erroneous.  In other words, both  physical qubits and classical measurement outcomes are affected by noise.

Below we consider random $n$-qubit Pauli errors $E\in \{I,X,Y,Z\}^{\otimes n}$.
Let $\mathrm{Supp}(E)\subseteq [n]$  be the support of $E$, that is, the
subset of qubits acted upon by either $X$, $Y$, or $Z$.
\begin{definition}
Let $p\in [0,1]$. A random $n$-qubit Pauli error $E$  is called {\em  $p$-local stochastic noise} if
\begin{equation}
\mathrm{Pr}\left[F\subseteq \mathrm{Supp}(E)\right]\leq p^{|F|}
\qquad \mbox{for all $F\subseteq [n]$}.
\label{eq:prbound}
\end{equation}
We write $E\sim \mathcal{N}(p)$ to denote random variables which are $p$-local stochastic noise.
\end{definition}
We will assume that each layer of gates in the ideal circuit is followed by
a random Pauli error $E\sim \calN(p)$ for some noise rate $p$.
Errors occurring after each layer may or may not be independent.
Namely, if $E_j$ is the error occurring after the $j$-th level of gates,
we only require that the marginal distribution of $E_j$ belongs to $\calN(p)$,
that is, $E_j\sim \calN(p)$.  No further assumptions are made about the
joint distribution of the errors $E_j$.
For simplicity we shall assume that the noise rate $p$ is identical
for each layer of gates.

A noisy preparation of the initial state $|0^n\ra$ will be modeled by 
the ideal state preparation followed by a random Pauli error $E_{in}\sim \calN(p_{in})$
for some noise rate $p_{in}$.
It produces a random basis vector $|x\ra$, where $x_i=1$ if $E_{in}$
acts on the $i$-th qubit by $X$ or $Y$, and $x_i=0$ otherwise.

Likewise, a noisy measurement of $n$ qubits in the computational basis
will be modeled by the ideal measurement preceded by a random Pauli error
$E_{out} \sim \calN(p_{out})$ for some noise rate $p_{out}$.
A noisy measurement of a state $\psi$
produces an outcome $z\in \{0,1\}^n$ with probability
$|\langle z|E_{out}|\psi\rangle|^2 = |\langle z\oplus x|\psi\rangle|^2$,
were $x_i=1$ if $E_{out}$ acts on the $i$-th qubit by 
$X$ or $Y$, and $x_i=0$ otherwise. Thus the random bit string $x$
determines positions of faulty measurement outcome bits. 

Let us now formally define what we mean by a noisy implementation of a quantum circuit. 
\begin{definition}[\textbf{Noisy implementation}]
Consider a circuit $U=U_D\cdots U_1$ of depth $D$, where
$U_j$ is a depth-$1$ circuit applied in the $j$-th time step/layer,  with the
initial state~$\ket{0^n}$ and a computational basis measurement  at the end. 
A noisy implementation of the circuit $U$ with noise rates $p_{in},p,p_{out}$ 
produces an output $z_{out}\in \{0,1\}^n$ according to the conditional distribution 
\begin{align}
\mathrm{Pr}(z_{out}|E_{in},E_1,\ldots,E_D,E_{out})&=|\bra{z_{out}}E_{out}E_DU_D\cdots 
E_1 U_1E_{in}\ket{0^n}|^2.
\label{eq:noisemodelgeneralconditionaloutput}
\end{align}
Here $E_{in},E_1,\ldots,E_D,E_{out}$ are random $n$-qubit Pauli errors 
drawn from some joint distribution such that
$E_{in}\sim \calN(p_{in})$, $E_j\sim \calN(p)$ for $1\le j\le D$, and $E_{out}\sim \calN(p_{out})$.
\end{definition}
To simplify the notations, below we assume that all noise rates
are identical, i.e. $p=p_{in}=p_{out}$.
This noise model is motivated by the concept of locally decaying and ``adversarial stochastic'' noise, where {\em every fault path} 
of $k$ locations in the circuit occurs with probability bounded by $p^k$, see
e.g.,~\cite{AliferisGottesmanPreskill2006,aliferisgottesmanpreskill2008}. In particular,  it does not assume independence of noise processes acting on different qubits or regions of the circuit. 
 Likewise, it does not assume independence of errors that occur at different time steps.
Local stochastic noise has the following basic features. 
\begin{lemma}[\textbf{Basic properties of local stochastic noise}]
\leavevmode
\begin{enumerate}[(i)]
\item{Suppose $E\sim \mathcal{N}(p)$, and $E'$ is a random Pauli such that $\mathrm{Supp}(E')\subseteq \mathrm{Supp}(E)$ with probability 1. Then $E'\sim \mathcal{N}(p)$.\label{it:claimone}}
\item{Suppose $E\sim \mathcal{N}(p)$ and $E'\sim \mathcal{N}(q)$ are independent random Paulis. Then $E\cdot E' \sim \mathcal{N}(p+q)$.\label{it:claimtwo}}
\item{Suppose $E\sim \mathcal{N}(p)$ and $E'\sim \mathcal{N}(q)$ are random Paulis which may be dependent. Then $E\cdot E'\sim \mathcal{N}(q')$ where $q'=2\max\{\sqrt{p},\sqrt{q}\}$.\label{it:claimthree}}
\item{Suppose $E\sim \mathcal{N}(p)$ is a random Pauli and $C$ is a depth-1 Clifford circuit composed of one- and two-qubit gates. Then $C E C^{\dagger}\sim \mathcal{N}(\sqrt{2p})$.\label{it:claimfour}}
\end{enumerate}
\label{claim:basic}
\end{lemma}

\begin{proof}
For convenience in the following we identify $n$-qubit Pauli operators with the $n$-bit string describing its support (that is, we omit the $\mathrm{Supp}$ notation).

Part~\eqref{it:claimone} follows directly from the definition of local stochastic noise and the fact that
\[
\mathrm{Pr}\left[F\subseteq E' \right]\leq \mathrm{Pr}\left[F\subseteq E\right]
\]
since $E' \subseteq E$. 

For part~\eqref{it:claimtwo}, note that
\[
\mathrm{Pr}\left[F\subseteq E\cdot E'\right]\leq \sum_{F=F_1F_2} \mathrm{Pr}\left[F_1 \subseteq E\right] \mathrm{Pr}\left[F_2 \subseteq E'\right] 
\]
where the right hand side is the sum over partitions of $F$ into two disjoint bit strings $F_1, F_2$. We arrive at (i) by plugging in Eq.\eqref{eq:prbound} and performing the sum.

For part~\eqref{it:claimthree} we can again sum over partitions of $F$ into two disjoint bit strings $F_1, F_2$:
\begin{equation}
\mathrm{Pr}\left[F\subseteq E\cdot E'\right]\leq \sum_{F=F_1F_2} \mathrm{Pr}\left[F_1 \subseteq E \text{ and } F_2 \subseteq E'\right]
\label{eq:ee1}
\end{equation}
Now we use the fact that
\begin{equation}
\mathrm{Pr}\left[F_1 \subseteq E \text{ and } F_2 \subseteq E'\right] \leq \min\{\mathrm{Pr}\left[F_1 \subseteq E\right], \mathrm{Pr}\left[F_2 \subseteq E'\right]\}\leq \min\{p^{|F_1|}, q^{|F_2|}\}\leq \max\{p,q\}^{|F|/2}.
\label{eq:ee2}
\end{equation}
where in the last line we used the fact that $|F_1|+|F_2|=|F|$. Plugging Eq.~\eqref{eq:ee2} into Eq.~\eqref{eq:ee1} we get
\[
\mathrm{Pr}\left[F\subseteq E\cdot E'\right]\leq 2^{|F|}\max\{p,q\}^{|F|/2},
\]
which establishes part (iii).

For part~\eqref{it:claimfour} we need to show that
\be
\mathrm{Pr}[F\subseteq CEC^\dag]\le (2p)^{|F|/2}.
\ee
Recall that $C$ is a depth-$1$ circuit
composed of one- and two-qubit Clifford gates. 
Let  $F'\subseteq F$ be the subset of qubits that
do not participate in any  two-qubit gate.
Note that $F'\subseteq CEC^\dag$ iff $F'\subseteq E$.
Suppose $C$ contains $m$ two-qubit gates that act non-trivially on $F$. Let these gates be $G_1,\ldots,G_m$ and $Q_i=\{q_i(0),q_i(1)\}\subseteq [n]$
be the qubits acted upon by $G_i$. By definition,
$F\cap Q_i\ne \emptyset$ for all $i$.
For each bit string $x\in \{0,1\}^m$ define a subset $F_x\subseteq [n]$ as 
\[
F_x = F' \cup \{q_1(x_1)\} \cup \{q_2(x_2)\} \cup \ldots \cup \{q_m(x_m)\}.
\]
Here all unions are disjoint.
We claim that 
\be
\label{claim_proof_eq1}
\mathrm{Pr}[F\subseteq CEC^\dag]\le \sum_{x\in \{0,1\}^m}\; \mathrm{Pr}[F_x\subseteq E].
\ee
Indeed, suppose $F\subseteq CEC^\dag$.
Then $F'\subseteq CEC^\dag$ and thus $F'\subseteq E$.
From $F\cap Q_i\ne \emptyset$ and
$F\subseteq CEC^\dag$ one infers that 
$CEC^\dag$ acts non-trivially on $Q_i$.
By assumption, $C$ is a depth-$1$ circuit. Thus
$CEC^\dag$ and $G_i E G_i^\dag$ have the same action on $Q_i$.
Since $G_i G_i^\dag =I$, we conclude that 
$E$ must act non-trivially on $Q_i$.
Thus $q_i(0)\in E$ and/or $q_i(1)\in E$ for each $i=1,\ldots,m$. 
The above shows that $F_x\subseteq E$ for at least one $x$.
The union bound now gives Eq.~\eqref{claim_proof_eq1}.

By assumption, $E\sim \calN(p)$ and thus $\mathrm{Pr}[F_x\subseteq E]
\le p^{|F_x|}$.
Note that $|F_x|=|F|-m_2$, where $m_2$
is the number of gates $G_i$ such that $Q_i\subseteq F$.
Substituting this into  Eq.~\eqref{claim_proof_eq1}  gives
\be
\label{claim_proof_eq2}
\mathrm{Pr}[F\subseteq CEC^\dag]\le \sum_{x\in \{0,1\}^m}\;  p^{|F_x|}
=2^m p^{|F|-m_2}\le (2p)^{|F|-m_2}.
\ee
Here the last inequality uses the bound $m=|F_x|-|F'|\le |F_x|=|F|-m_2$.
It remains to notice that $m_2\le |F|/2$ and thus
$(2p)^{|F|-m_2}\le (2p)^{|F|/2}$ assuming that $2p\le 1$.
\end{proof}

Lemma~\ref{claim:basic}  allows us to rewrite the error model defined by the conditional distribution~\eqref{eq:noisemodelgeneralconditionaloutput} in the case of 
constant depth Clifford circuits. That is, we have the following:
\begin{lemma}\label{lem:noisemodelcliffordcircuit}
Suppose $U=C_D\cdots C_2 C_1$, where $C_j$ are  depth-$1$ Clifford circuits composed of one- and two-qubit gates.  
Then a noisy implementation of $U$ with the noise rate
$p$  produces an output $z_{out}\in \{0,1\}^n$ according to the conditional distribution 
 \begin{align}
 \mathrm{Pr}(z_{out}|E)= |\bra{z_{out}}EU\ket{0^n}|^2,\label{eq:outputdistrpzoute}
 \end{align}
where  $E\sim \calN(4p^{4^{-D-1}})$ is a random $n$-qubit Pauli error.
\end{lemma}
In particular, if we consider  a noisy implementation of a  Clifford circuit of constant depth $D=O(1)$, we may without loss of generality assume that the output distribution is of the form~\eqref{eq:outputdistrpzoute} with $E\sim \mathcal{N}(q)$  for some constant $q\in (0,1]$. 
\begin{proof}
Suppose $C$ is a depth-$1$ Clifford circuit composed of one- and two-qubit gates.
Let $E\sim \calN(p)$ and $E'\sim \calN(q)$ be random Pauli errors
which may be dependent.
We claim that 
\be
\label{E'E''}
ECE' =E''C \quad \mbox{where} \quad E''\sim \calN(r), \quad r=2\max{\{ (2q)^{1/4},p^{1/2}\}}.
\ee
Indeed, part~(iv) of Lemma~\ref{claim:basic} implies
that $CE'=\tilde{E}C$, where $\tilde{E}\sim \calN(\sqrt{2q})$.
Merging $E$ and $\tilde{E}$ using
part~\eqref{it:claimthree} of Lemma~\ref{claim:basic} 
one gets Eq.~\eqref{E'E''}.
Let $E_{in},E_1,\ldots,E_{D},E_{out}$ be the random Pauli errors 
that occur in the noisy implementation of $U$, see Eq.~(\ref{eq:noisemodelgeneralconditionaloutput}).
By assumption, $E_{in}\sim \calN(p)$, $E_{out}\sim \calN(p)$, and $E_j \sim \calN(p)$
for all $1\le j\le D$. 
For a given realization of noise, the condition 
distribution of $z_{\mathrm{out}}$  is  $|\bra{z_{out}}U_{\textrm{noisy}}\ket{0^n}|^2$,
where
\[
U_{\textrm{noisy}}=E_{out} E_D C_D\cdots  E_2 C_2 E_1 C_1 E_{in}.
\]
Let us insert a dummy depth-$1$ circuit $C_{D+1}=I$ between
$E_{out}$ and $E_D$.
Repeatedly using Eq.~\eqref{E'E''} with $C\in \{C_1,\ldots,C_{D+1}\}$
one can commute all errors that appear in $U_{\textrm{noisy}}$  to the left
and merge them into a single Pauli error.
One arrives at $U_{\textrm{noisy}}= EU$,
where $E\sim \calN(q_{D+2})$ and 
$q_{D+2}$ is given by a recursive equation
\[
q_{j+1} = 2\max{\{ (2q_{j})^{1/4},p^{1/2}\}}, \quad j=1,\ldots,D+1
\]
with $q_1=p$. A simple algebra shows that 
$(2q_{j})^{1/4} \ge p^{1/2}$ for all $j$.
Thus $q_{j+1} = 2(2q_{j})^{1/4}$ for $j\ge 1$, that is,
\[
q_{j+1} =p^{4^{-j}} 2^{5\sum_{i=1}^j 4^{-i}}
\le 2^{5/3}p^{4^{-j}}  \le 4p^{4^{-j}}.
\]
We conclude that  $U_{\textrm{noisy}}= EU$
with $E\sim \calN(q_{D+2})$, $q_{D+2}\le 4p^{4^{-D-1}}$.
\end{proof}

\subsection{Quantum code properties}
\label{subs:codes}
In the following we shall make use of a 
CSS-type~\cite{calderbank1996good,steane1996multiple}
quantum error correcting code $\mathcal{Q}_m$ 
encoding one logical qubit into $m$ physical qubits. Such a code has logical basis states 
\begin{equation}
\label{logical01}
|\overline{0}\ra = \gamma \sum_{x\in \calB} |x\rangle
\quad \mbox{and} \quad
 |\overline{1}\ra = \gamma \sum_{x\in \calB} |x\oplus \beta\rangle,
\end{equation}
where $\calB\subseteq \{0,1\}^m$ is a linear subspace,
$\beta \notin \calB$ is some vector, and $\gamma = |\calB|^{-1/2}$
is a normalizing coefficient. 
Given a bit string $v\in \{0,1\}^m$ let $X(v)$ and $Z(v)$ be the 
products of Pauli $X$ and $Z$ respectively over all qubits $j\in [m]$ with $v_j=1$.
Logical Pauli operators  of $\calQ_m$ can be chosen as 
\begin{equation}
\label{eq:logicalZX}
\overline{Z} = Z(\alpha)  \quad \mbox{and} \quad \overline{X} = X(\beta),
\end{equation}
where $\beta$ is given in Eq.~\eqref{logical01} and $\alpha\in \{0,1\}^m$ must have odd overlap with $\beta$
(to ensure that $\overline{Z}$ and $\overline{X}$ anti-commute)
and $\alpha \in \calB^\perp$ (to ensure that
the logical states $|\overline{0}\ra$, $|\overline{1}\ra$ are eigenvectors
of  $\overline{Z}$).

We shall need an infinite family of codes $\mathcal{Q}_m$ as above 
for some diverging sequence of $m$'s that
obey the following conditions.
\begin{cond}
The logical Hadamard and the phase gate $S=\mathrm{diag}(1,i)$
can each be implemented  by a depth-1 Clifford circuit composed of one- and two-qubit Clifford gates.
\label{cond:logical}
\end{cond}
Recall that the logical CNOT  can be implemented transversally
for any CSS code~\cite{calderbank1996good,steane1996multiple}.
Thus any logical depth-$d$ Clifford circuit composed of $H,S,\CNOT$ gates
can be implemented by a physical depth-$d$ Clifford circuit composed of two-qubit Clifford gates.

Next, we  require that $\mathcal{Q}_m$ satisfies  a certain single-shot state preparation property~\cite{bombin2015single}. Namely, we assume that the logical basis state $|\overline{0}\rangle$ 
can be prepared by the procedure shown in Fig.~\ref{fig:singleshotstateprep}. In addition to $m$ qubits which hold the final logical state, it uses~$m_{\textrm{anc}}$ ancilla qubits.  The procedure involves initializing all  qubits in the state $|0\rangle$, applying a constant depth
Clifford circuit $W$, and  measuring all ancillas in the computational basis.
Let $s\in\{0,1\}^{m_{\textrm{anc}}}$ be the measurement outcome.
Finally, a suitable
Pauli recovery operator $\rec(s)$ is applied to the remaining $m$ qubits.

In order for this procedure to prepare the logical state~$\ket{\overline{0}}$ in the absence of noise, we require that $m_{\textrm{anc}}, W$ and $\rec(s)$ obey 
\begin{align}
(\rec(s)\otimes \ket{s}\bra{s})
W (\ket{0^m}\otimes\ket{0^{m_{\textrm{anc}}}})
&=
\gamma_s\ket{\overline{0}}\otimes \ket{s} \label{eq:noisefreestatepreparation}
\end{align}
for all $s\in \{0,1\}^{m_{\textrm{anc}}}$. Here 
$\gamma_s\in \mathbb{C}$ is a normalization coefficient and
the tensor product separates the $m$-qubit register used by the code $\calQ_m$
and a register of $m_{\mathrm{anc}}$ ancillary qubits.

\begin{figure}[h]
\subcaptionbox{State preparation algorithm \label{fig:stateprepalgo}}[0.55\textwidth]{
\begin{minipage}[t]{0.55\textwidth}
\centering
\begin{mdframed}[style=mystyle]
\begin{enumerate}
\item
Prepare $m+m_{\textrm{anc}}$ qubits  in the state~$\ket{0^m}\otimes\ket{0^{m_{\textrm{anc}}}}$.  
\item
Apply a constant-depth Clifford circuit~$W$.
\item Measure each ancilla qubit in the $Z$-basis, giving an outcome $s\in \{0,1\}^{m_{\textrm{anc}}}$.
\item Depending on the outcome~$s$, apply a suitable Pauli recovery~$\rec(s)$ to the state of the~$m$ unmeasured qubits.
\end{enumerate}
\end{mdframed}
\end{minipage}}
\hspace{0.25cm}
\subcaptionbox{Single shot state preparation circuit.\label{fig:stateprepcircuit}}[0.4\textwidth]{
\centering
\begin{minipage}{0.3\textwidth}
\begin{align}
  \Qcircuit @C=1.0em @R=0.7em {
      \lstick{\ket{0^m}}                  & \multigate{1}{W}  & \qw   & \gate{\rec(s)} &\qw \\ 
       \lstick{\ket{0^{m_{\textrm{anc}}}}} & \ghost{W} & \meter  & \cctrl{-1}\cw
 } 
\end{align}
 \vspace{0.25cm}
\end{minipage}}
\caption{The single-shot state preparation procedure. 
Below we will  incorporate the Pauli correction~$\rec(s)$ into
our computational problem, eliminating the need to evaluate~$\rec(s)$ by a quantum circuit.
\label{fig:singleshotstateprep}} 
\end{figure}
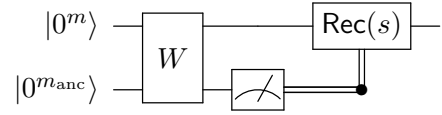

We require a stronger property in order to guarantee that even a noisy implementation of the circuit in Fig.~\ref{fig:stateprepcircuit} prepares the logical state~$\ket{\overline{0}}$ up certain ``manageable'' errors. 
A noisy implementation of the circuit 
outputs
the measured string $s$ and a state
\[
( \rec{(s)}\otimes \ket{s}\bra{s}) EW(\ket{0^m}\otimes\ket{0^{m_{\textrm{anc}}}}),
\]
where $E$  is a random Pauli error acting on all  $m+m_{\mathrm{anc}}$ qubits. 
 Indeed,  as in Lemma~\ref{lem:noisemodelcliffordcircuit}, all qubit initialization errors and
gate errors in the execution of~$W$ can be absorbed into $E$
by commuting them forward in time towards
the last gate of $W$. A Pauli error in the execution of $\mathsf{Rec}(s)$
either commutes or anti-commutes with $\mathsf{Rec}(s)$. Since the global phase
of a state is irrelevant, such error can be commuted backwards towards $W$
and incorporated into $E$. Likewise, a measurement 
error on some ancillary qubit~$u$ is equivalent to 
a Pauli error $X_u$ immediately preceding the measurement. 
Such errors can be absorbed into $E$.

The difference between the final state of the noisy circuit 
and the desired logical state~$\ket{\overline{0}}$  can be quantified using a {\em repair}
operator $\repair{(E)}$ which is an $m$-qubit Pauli operator 
satisfying
\be
\label{RecRepx}
(\rec{(s)}\otimes\ket{s}\bra{s}) EW(\ket{0^m}\otimes\ket{0^{m_{\textrm{anc}}}}) = \gamma_s 
(\repair{(E)}\ket{\overline{0}})\otimes\ket{s}
\ee
for all $s$ and $E$. Thus~$\repair{(E)}$ can be thought of as the residual error in the prepared state when using a noisy implementation of the state preparation circuit. Recall that a random Pauli operator $E$ is said to be $p$-local stochastic noise, denoted $E\sim \mathcal{N}(p)$, if Eq.~\eqref{eq:prbound} is satisfied. For the error~$\repair{(E)}$ to be ``manageable'', we require that it is local stochastic whenever~$E$ is. This leads to the following requirement for the error-correcting code~$\mathcal{Q}_m$.  Here and below we write $\pauli(m)$ for the $m$-qubit Pauli group.
\begin{cond}[\textbf{Single-shot basis state preparation}]
Let  $c,c',c'',d$ be some universal constants. 
For each code $\calQ_m$ in the family
there must exist an integer $m_{\mathrm{anc}}\leq m^{c}$, a 
depth-$d$ Clifford circuit $W$ acting on $m+m_{\mathrm{anc}}$ qubits, recovery and repair functions
\begin{align}
\begin{matrix}
\rec: &\{0,1\}^{m_{\mathrm{anc}}}&\rightarrow &\mathrm{Pauli}(m)\\
\repair: & \mathrm{Pauli}(m+m_{\mathrm{anc}})&\rightarrow &\mathrm{Pauli}(m)
\end{matrix}
\end{align}
such that 
\begin{equation}
\left(\rec(s)\otimes |s\rangle\langle s|\right) EW(|0^{m}\rangle\otimes |0^{m_{\mathrm{anc}}}\rangle)= \gamma_s (\repair(E) |\overline{0}\rangle)\otimes |s\rangle
\label{eq:F}
\end{equation}
for all $s\in \{0,1\}^{m_{\mathrm{anc}}}$ and $E\in \mathrm{Pauli}(m+m_{\mathrm{anc}})$. 
Here $\gamma_s \in \mathbb{C}$ is a normalization factor.
Furthermore, for all noise rates $p\in [0,1]$, we must have that
 $E\sim \mathcal{N}(p)$ implies $\repair(E)\sim \mathcal{N}(c'p^{c''})$. 
\label{cond:shot}
\end{cond}
\noindent 
In the noise-free case $p=0$ one has $E=I$ with certainty and 
$\repair(E)\sim \mathcal{N}(0)$, that is, $\repair(I)=I$.
Thus condition~\eqref{eq:F} specializes to its noise-free version~\eqref{eq:noisefreestatepreparation}. We emphasize that in this definition, we make no assumptions about how efficiently the recovery
function~$\rec$ can be computed, or whether or not it can be computed by a constant-depth circuit. In fact, the quantum circuits we construct will not  apply the recovery~$\rec(s)$ to physical qubits. Rather, this recovery is incorporated into the  computational problem  such  that only the efficiency of {\em verifying} the validity of a solution depends on~$\rec$ (see Section~\ref{sec:noisetolerantrelation}). Likewise, the repair function $\repair(E)$ does not have to be efficiently computable.

Our final requirement  is that
the logical qubit encoded by 
$\mathcal{Q}_m$ can be measured in the $Z$-basis  in a manner which is robust to local stochastic noise.  
Recall that the logical-$Z$ operator of the code $\mathcal{Q}_m$ is chosen
as $\overline{Z}=Z(\alpha)$
for some $\alpha \in \{0,1\}^m$.
Define a function
\[
\parity(x) = \sum_{j=1}^m x_j \alpha_j {\pmod 2},
\]
where $x\in \{0,1\}^m$. The eigenvalue of $\overline{Z}$ can be measured using the following protocol:
\begin{enumerate}
\item
Measure each of the $m$~qubits in the $Z$-basis, obtaining an outcome $x\in \{0,1\}^m$.
\item
Compute the value $\dec(x)\in \{0,1\}$ of a certain decoding function~$\dec:\{0,1\}^m\rightarrow\{0,1\}$. 
\item Output $(-1)^{\dec(x)}$.
\end{enumerate}
Let us first consider the noiseless case. Suppose this procedure is applied to 
a logical basis state $|\overline{b}\ra$, where $b\in \{0,1\}$.
From Eq.~\eqref{logical01} one infers that the outcome $x$ 
obtained at Step~1 always
belongs to a linear subspace 
\begin{equation}
\label{Lsubspace}
\calL = \mathrm{span}(\beta, \calB)\subseteq \{0,1\}^m.
\end{equation}
This subspace includes all $m$-qubit basis states that appear
in the logical states $|\overline{0}\rangle$, $|\overline{1}\rangle$, see Eq.~\eqref{logical01}
(equivalently, $\calL$ includes all basis vectors that obey
$Z$-type stabilizers of the code $\calQ_m$).
To ensure that $(-1)^{\dec(x)}=(-1)^b$ for each possible outcome $x$,
the decoding function must obey 
\begin{align}
\dec(x)=
\parity(x) \quad \mbox{for all $x\in \calL$}\ 
\end{align}
Indeed, this guarantees that $(-1)^{\dec(x)}=\la x|\overline{Z}|x\ra$ for all 
possible outcomes $x$. By linearity, the noiseless measurement
also works for any superposition of the logical basis states.

To deal with noise, we require a stronger property for the decoding function. It must produce -- with high probability -- a correct output even in the case of a noisy implementation of the above procedure. 
Here the noise may include errors in the input logical state
as well as faulty measurements of physical qubits.
As in Lemma~\ref{lem:noisemodelcliffordcircuit},
such errors can be merged into a single Pauli error $E$
preceding the ideal $m$-qubit measurement.
Furthermore, only the $X$-part of the error $E$ matters
since every qubit is measured in the $Z$-basis.
Thus we can assume without loss of generality that $E=X(v)$ for some random bit string $v\in \{0,1\}^m$.
This leads to the following requirement:
\begin{cond}[\textbf{Error threshold}]
Let $c,c',q_{\mathrm{th}}>0$ be some universal constants.
For each code $\calQ_m$ in the family there must exist 
a function  $\dec:\{0,1\}^m\rightarrow \{0,1\}$ such that 
the following holds. First,
\begin{equation}
\dec(x)=\parity(x)\quad \mbox{for all $x\in \calL$}.
\label{eq:paritydecode}
\end{equation}
Secondly, suppose $q<q_{\mathrm{th}}$ and
$v\in \{0,1\}^m$ is a random bit string such that $X(v) \sim \calN(q)$. Then 
\begin{equation}
\mathrm{Pr} \left[\dec(x\oplus v)=\parity(x)\right]\geq  1- \exp(-c' m^c) 
\label{eq:threshold}
\end{equation}
for all $x\in \calL$.
\label{cond:meas}
\end{cond}
This condition ensures that the logical~$\overline{Z}$ measurement 
can be realized by the above algorithm even if the physical measurements
as well as the input logical state are noisy,
provided that the noise rate is below a certain constant threshold value $q_{\mathrm{th}}$.
The threshold value~$q_{th}$  is a key figure of merit in our scheme: it determines how much noise can be tolerated while still guaranteeing that a noisy implementation produces correct outputs. Akin to fault-tolerance threshold theorems, we provide rigorous but rather  pessimistic analytical bounds on this quantity.

In Section~\ref{sec:codes} we show that the standard 2D surface code~\cite{bravyi1998quantum}
equipped with a suitable single-shot state preparation scheme
satisfies Conditions~1, 2 and~3.

\subsection{A noise tolerant relation from any controlled Clifford circuit}
\label{sec:noisetolerantrelation}
Recall that a relation~$R$   is defined by a subset of valid input-output pairs,
 $R:\{0,1\}^v\times \{0,1\}^n\rightarrow \{0,1\}$. An input-output pair $(b,z)\in \{0,1\}^v\times \{0,1\}^n$ is said to satisfy the relation if and only if $R(b,z)=1$. We say that a (classical or quantum) circuit {\em solves the relation problem defined by $R$ on input~$b\in \{0,1\}^v$} if it outputs $z\in \{0,1\}^n$  such that~$R(b,z)=1$.

 We will consider relations (and associated  problems)  defined by certain (ideal) quantum circuits: the 1D Magic Square Problem and the HLF problem considered in~\cite{bragokoe18} are examples. Suppose that $U$ is a depth-$D$ quantum circuit which acts on two registers, a data register of $n$ qubits and an input register of $v$ qubits. We specialize to the case where $U$ is a \textit{controlled Clifford circuit}. That is, every gate in the circuit is  a classically controlled Clifford gate which acts as $|\phi\rangle|b\rangle\rightarrow (C_b|\phi\rangle)|b\rangle$ for some Clifford unitary $C_b$. We also assume each gate acts nontrivially on at most $k=O(1)$ qubits. Thus, for any input $b\in \{0,1\}^v$,  a depth-$D$ Clifford unitary $C_b$ is applied to the data register:
\[
U|\phi \rangle |b\rangle=\left(C_{b}|\phi\rangle\right)|b\rangle \qquad \qquad b\in \{0,1\}^v.
\]
Now consider a quantum computation in which the circuit $U$ is applied to the initial state $|0^{n}\rangle|b\rangle$ and then all data qubits are measured in the computational basis, resulting in a bit string $z\in \{0,1\}^n$ sampled from the distribution
\begin{align}
p_b(z)=|\langle z|C_b|0^n\rangle|^2\ .\label{eq:pbzv}
\end{align}
A corresponding circuit diagram in shown in Fig.~\ref{fig:ideal}. We define the following relation:
\begin{definition}[\bf{Bare relation}]
Let $R_U:\{0,1\}^v\times \{0,1\}^n\rightarrow \{0,1\}$ be the relation 
\[
R_U(b,z)={\begin{cases} 1, & p_b(z)>0\\ 0, & \text{otherwise}.\end{cases}}.
\]
We will call this the {\em bare relation} associated with $U$.
\label{def:ideal}
\end{definition}That is, pairs $(b,z)$ satisfying~$R_U$ have the property that $z$~occurs with non-zero probability in the distribution~\eqref{eq:pbzv} over outputs in the above computation.
In particular, this definition trivially implies the following, for any classically controlled Clifford circuit~$U$. 
\begin{lemma}
For every input $b\in \{0,1\}^v$, the circuit $U$ solves the relation problem defined by $R_U$ with probability~$1$.
\end{lemma}

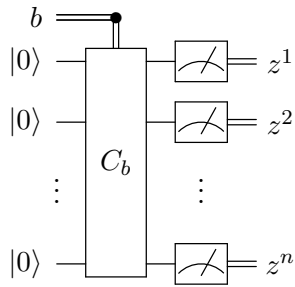
\begin{figure}
  \[
  \Qcircuit @C=1.0em @R=0.7em {
      \lstick{b} & \cctrl{1}\cw \\
      \lstick{\ket{0}}   & \multigate{6}{{C_b}} & \meter & \cw & \!\!z^1\\ 
      \lstick{\ket{0}}  & \ghost{{C_b}} & \meter & \cw & \!\!z^2\\ 
                                     &               & \\
   \vdots &  & \vdots  \\                     
                                  &                           & \\
                                 & &  \\
      \lstick{\ket{0}}  & \ghost{{C_b}} & \meter & \cw & \!\!z^n\\ 
   } \]
      \caption{(Ideal) circuit  $U$ defining the bare relation~$R_U$: It takes as input $b\in \{0,1\}^v$ and applies a classically controlled Clifford gate~$C_b$.\label{fig:ideal}}
\end{figure}

While the input/output of an \textit{ideal implementation} of $U$ satisfies the relation $R_U$, this may no longer hold for a noisy implementation. Moreover, if we use standard quantum fault-tolerance techniques
to protect the computation from noise, we would incur an undesirable super-constant overhead in circuit depth
when a  typical constant-depth circuit is recompiled into a fault-tolerant one,
see e.g. Ref.~\cite{raussendorf2007fault}.

 In the following we show that this overhead can be avoided by modifying the relation, that is, the computational problem, rather than the circuit computing solutions. In particular, for any constant-depth controlled Clifford circuit $U$ we describe a \emph{noise-tolerant relation} relation $\mathcal{R}_U$ with the following properties (informally): (a) The input/output of a constant-depth quantum circuit satisfies $\mathcal{R}_U$ with high probability, even in the presence of noise,  and (b) If a classical probabilistic circuit satisfies $\mathcal{R}_U$ then there is another classical probabilistic circuit with comparable depth that satisfies the bare relation $R_U$.

 The definition of~$\mathcal{R}_U$ relies on a quantum error-correcting code~$\mathcal{Q}_m$ with the properties outlined in Section~\ref{subs:codes}. Recall that this involves  recovery- and decoding functions
 \begin{align}
 \begin{matrix}
 \rec: &\{0,1\}^{m_{\mathrm{anc}}}&\rightarrow &\mathrm{Pauli}(m)\\
 \dec:&\{0,1\}^m & \rightarrow & \{0,1\}\ .
 \end{matrix}
 \end{align}
Below we consider $n$ copies of the code $\calQ_m$, where each copy encodes
one of the qubits acted upon by the circuit $C_b$. Accordingly, we shall use $n$-tuples
of syndromes $s=(s^1,\ldots,s^n)\in \{0,1\}^{nm_{\mathrm{anc}}}$. Let $\overline{C}_b$ be the encoded version of the Clifford circuit $C_b$,
where each qubit of $C_b$ is encoded into $m$ qubits using the code $\calQ_m$.
Note that  $\overline{C}_b$  is a Clifford circuit acting on $nm$ qubits.
For the definition of the relation~$\mathcal{R}_U$, we need to know how the 
tensor product of  Pauli
recovery operators $\rec(s^j)$  propagates through the Clifford circuit $\overline{C}_b$.
We have
\begin{equation}
\overline{C}_b \left( \rec(s^1)\otimes \cdots \otimes \rec(s^n)\right) \overline{C}_b^{\dagger}\sim X(f) Z(h)
\label{eq:Q}
\end{equation}
for some $f,h\in \{0,1\}^{mn}$ depending on $s$ and $b$. We write 
$f=f^1f^2\ldots f^n$, where $f^i$ is the restriction of $f$ onto the $i$-th codeblock.
Note that $f^i=f^i(s,b)$ since Eq.~\eqref{eq:Q} uniquely defines
$f^i$ for each $s$ and $b$. This can be described by functions
$f^i:\{0,1\}^{nm_{\textrm{anc}}}\times \{0,1\}^{v}\rightarrow \{0,1\}^m$ for $1\leq i\leq n$.
\begin{definition}[\bf{Noise-tolerant relation}]
   The  noise-tolerant relation
   \begin{align}
   \mathcal{R}_U: \{0,1\}^v\times \big(\{0,1\}^{nm_{\textrm{anc}}}\times \{0,1\}^{nm} \big)\rightarrow \{0,1\}
   \end{align}
    defined by $U$   is given by
 \[
\mathcal{R}_{U} (b,(s,y))={\begin{cases} 1, & \text{if } R_U(b,z)=1, \text{where } z_i=\dec(y^i\oplus f^i(s,b)))\textrm{ for }1\leq i\leq n\\ 0, & \text{otherwise}.
\end{cases}}
\label{definition:enc}
\]
\end{definition}
Note that $\mathcal{R}_U$  depends on the choice of recovery and decoding functions $\rec,\dec$ associated with the code $\mathcal{Q}_m$.  To motivate this definition, we present a quantum algorithm which solves the relation problem defined by $\mathcal{R}_U$ with certainty for any input.  We consider a system of $mn$ physical qubits partitioned into $n$ codeblocks $[mn]=B^1B^2\cdots B^n$, where each codeblock encodes a single logical qubit using a code $\mathcal{Q}_m$ of the type described in the previous section. Each codeblock will also be associated with an additional $m_{anc}$ ancilla qubits which are used for state preparation. Here and below we use superscripts to index codeblocks and subscripts to index individual bits.  We use an overbar to denote logical operators and states.  Suppose we are given an input bit string $b\in \{0,1\}^v$.  In the following we imagine that $b$ is held in a perfect classical memory. Consider the procedure described in Algorithm 1 (see Fig.~\ref{fig:algo} and the circuit realization shown in  Fig.~\ref{fig:algcircuit}). The output of the algorithm is the pair  $(s,y)\in \{0,1\}^{nm_{\textrm{anc}}}\times \{0,1\}^{nm}$.  We show the following:

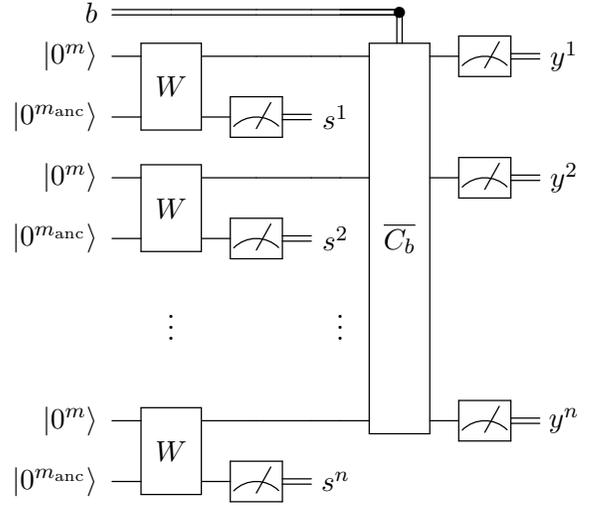
\begin{figure}
\subcaptionbox{Description of the algorithm.\label{fig:algo}}[0.5\textwidth]{
\begin{minipage}[t]{0.5\textwidth}
\centering
\begin{mdframed}[style=mystyle]
\begin{enumerate}
\item {\textbf{Single-shot state preparation:} For each codeblock $j=1,\ldots, n$, prepare $|0^m\rangle\otimes |0^{m_{\mathrm{anc}}}\rangle$, apply the constant-depth Clifford unitary $W$ from Condition \ref{cond:shot} and measure the $m_{\mathrm{anc}}$ ancilla qubits to get an outcome $s^j \in \{0,1\}^{m_{\mathrm{anc}}}$. Write $s=s^1s^2\cdots s^n$ for all measurement outcomes obtained in this step.}
\item{\textbf{Logical circuit: }Apply the logical Clifford circuit $\overline{C}_b$ as a sequence of depth-1 logical gates. }
\item{\textbf{Measurement: }For each $j=1, \ldots, n$, measure all physical qubits in $B^j$ in the computational basis. Write the measured bits as $y=y^1y^2\cdots y^n$ where $y^j\in \{0,1\}^m$.
}
\end{enumerate}
\end{mdframed}
\end{minipage}}
\hspace{1.5cm}
\subcaptionbox{Circuit realization.\label{fig:algcircuit}}[0.4\textwidth]{
\begin{minipage}[t]{0.4\textwidth}
\centering
\[
  \Qcircuit @C=1.0em @R=0.7em {
    \lstick{b} & \cw & \cw & \cw & \cw & \cctrl{1}\cw \\
      \lstick{\ket{0^m}}                  & \multigate{1}{W}  & \qw  & \qw & \qw & \multigate{10}{\overline{C_b}} & \meter & \cw & \!\!y^1\\ 
       \lstick{\ket{0^{m_{\textrm{anc}}}}} & \ghost{W} & \meter & \cw & \!\!s^1  \\
      \lstick{\ket{0^m}}                  & \multigate{1}{W}  & \qw  & \qw & \qw & \ghost{\overline{C_b}} & \meter & \cw & \!\!y^2\\ 
       \lstick{\ket{0^{m_{\textrm{anc}}}}} & \ghost{W} & \meter & \cw & \!\!s^2 & \\
                                     &                &              & \\
                               &                &              &\\
   & \vdots &  &  & \vdots \\                     
                                  &                &              & \\
                               &                &              & \\
                                 & & & \\
      \lstick{\ket{0^m}}                  & \multigate{1}{W}  & \qw  & \qw & \qw & \ghost{\overline{C_b}} & \meter & \cw & \!\!y^n\\ 
       \lstick{\ket{0^{m_{\textrm{anc}}}}} & \ghost{W} & \meter & \cw & \!\!s^n & 
   } \]
   \vspace{0.25cm}
\end{minipage}}
\caption{Algorithm~1 and its realization as a circuit.
\label{fig:algonecircuit}} 
\end{figure}

\begin{lemma}\label{lem:idealcircuitnoisetolerant}
The circuit in Fig.~\ref{fig:algonecircuit} solves the relation~$\mathcal{R}_U$ with certainty for any input~$b\in \{0,1\}^v$. 
\end{lemma}
Lemma~\ref{lem:idealcircuitnoisetolerant}  shows how the circuit in Fig.~\ref{fig:algonecircuit} performs in  the absence of noise. In Section~\ref{sec:noisetoleranceqcircuit}, we will show that a noisy implementation of this circuit still satisfies the relation~$\mathcal{R}_U$ with high probability. 
\begin{proof}
 For the state preparation step, this corresponds to noise rate $p=0$ in Condition \ref{cond:shot},  so that $\repair(E)=I$ in Eq. ~\eqref{eq:F}. Therefore the state after step 1. of the algorithm is
\[
\rec(s)|\overline{0}^n\rangle \quad \text{ where }\quad \quad \rec(s)=\bigotimes_{j=1}^{n} \rec(s^j).
\]
After applying the Clifford circuit in step 2. the state is 
\begin{equation}
|\phi_{b,s}\rangle=X(f) Z(h) \overline{C}_b|\overline{0}^n\rangle,
\label{eq:phiout}
\end{equation}
where $f,h\in \{0,1\}^{mn}$ are the functions of $(s,b)$ defined by Eq.~\eqref{eq:Q}.

Note from Eq.~\eqref{eq:phiout} that $X(f)|\phi_{b,s}\rangle$ has the same $Z$-basis measurement statistics as the encoded output state $\overline{C}_b|\overline{0}^n\rangle$ of the ideal controlled Clifford circuit. Therefore the $n$-bit string~$z$ with bits
\[
z_i=\dec(y^i\oplus f^i(s,b)))
\]
satisfies $R_U(b,z)=1$. That is, 
\begin{equation}
|\langle y|\phi_{b,s}\rangle|>0 \quad \text{ implies } \quad \mathcal{R}_U(b,(s,y))=1.
\label{eq:noisefree}
\end{equation}
We have shown that the input/output $(b,(s,y))$ pair of Algorithm 1 satisfies this relation with probability~$1$ in the absence of noise.  This is the claim.
\end{proof}

\subsubsection{Noise tolerance of quantum circuit}
\label{sec:noisetoleranceqcircuit}
We now show that a \textit{noisy implementation} of
the circuit  in Fig.~\ref{fig:algonecircuit}  satisfies $\mathcal{R}_U$ with high probability.  We consider a noise model in which a random Pauli $E\sim \mathcal{N}(p)$ is applied immediately before the measurements, see 
Lemma~\ref{lem:noisemodelcliffordcircuit} for a justification.  In particular, the output $(s,y)$ of the noisy algorithm with input $b$ is sampled from the distribution
\begin{equation}
P_{b}(s,y)=|\langle y\otimes s|E(\overline{C_b}\otimes I)W^{\otimes{n}}\left(|0^{mn}\rangle\otimes |0^{m_{\mathrm{anc}}n}\rangle\right) |^2\qquad \qquad E\sim \mathcal{N}(p).
\label{eq:noisyp}
\end{equation}
Here the tensor product separates the $n$ codeblocks from the $m_{\mathrm{anc}}n$ ancilla qubits used for state preparation. The random error $E$ may act nontrivially on both registers.  

\begin{theorem}[\textbf{Noise tolerance}]
Let $b\in \{0,1\}^v$ be an arbitrary input. Let $(b,(s,y))$ be the input/output of a noisy implementation the 
circuit  in Fig.~\ref{fig:algonecircuit},  i.e.,  $(s,y)$ are sampled from the distribution Eq.~\eqref{eq:noisyp}. We may choose $m=O(\mathrm{poly}(\log(n)))$ and $p_{th}=2^{-2^{O(D)}}$ such that for all $p<p_{th}$ we have 
\begin{equation}
\mathrm{Pr}\left[\mathcal{R}_U(b,(s,y))=1 \right]>0.99
\label{eq:pr}
\end{equation}
\label{thm:tol}
\end{theorem}
Theorem \ref{thm:tol} shows that a noisy constant-depth quantum circuit satisfies the relation $\mathcal{R}_U$ with high probability. 
\begin{proof}
Let us consider a noisy run of Algorithm 1 with input $b\in \{0,1\}^v$, in which a Pauli error $E\sim \mathcal{N}(p)$ is applied before the measurements. The amplitude for obtaining output $(s,y)$ is:
\begin{equation}
\label{A(s,y)}
A(s,y)=\langle y\otimes s|E(\overline{C_b}\otimes I)W^{\otimes n}|0^{mn}\otimes 0^{m_{\mathrm{anc}}n}\rangle.
\end{equation}
We may write $E=E'\otimes E''$ where $E'$ is an $mn$-qubit Pauli and $E''$ is an $nm_{\mathrm{anc}}$-qubit Pauli. Using part (i) of Claim \ref{claim:basic} we have $E', E''\sim \mathcal{N}(p)$. Applying
Condition \ref{cond:shot} to each of the $n$ codeblocks gives
\begin{equation}
\label{eq:single_shot1}
(\rec(s) \otimes |s\rangle \langle s|)(I\otimes E'')W^{\otimes n} |0^{mn}\otimes 0^{m_{\mathrm{anc}}n}\rangle
=\gamma_s (F |\overline{0}^n\rangle) \otimes |s\rangle,
\end{equation}
where $\gamma_s \in \mathbb{C}$ is a normalization coefficient and
$F=F(E'')$ is a tensor product of $n$ repair operators associated with each
codeblocks. More precisely, 
\[
F=\repair(I\otimes E^1)\otimes \cdots \otimes \repair(I\otimes E^n),
\]
where $E^j$ is the restriction of $E$ onto the ancillary register associated with the $j$-th codeblock.
Multiplying both sides of Eq.~\eqref{eq:single_shot1} by $\rec(s)$
and substituting it into Eq.~\eqref{A(s,y)} gives
\begin{align}
A(s,y)=\langle y|E'\overline{C_b}\langle s|E''W^{\otimes n}|0^{mn}\otimes 0^{m_{\mathrm{anc}}n}\rangle& = 
\pm \gamma_s \langle y|E'\overline{C_b}F\rec(s) |\overline{0}^n\rangle.
\end{align} 
Here we noted that $F$ and $\rec(s)$ are Pauli operators, so that 
commuting them through each other can only change the overall sign.
Since $E^j \sim \calN(p)$, 
Condition \ref{cond:shot} ensures that $\repair(I\otimes E^j) \sim \mathcal{N}(c'p^{c''})$ for some universal 
constants $c',c''$. Therefore $F\sim \mathcal{N}(c'p^{c''})$.

Now define
\[
F'=\overline{C_b} F\overline{C_b}^{\dagger}.
\]
Each layer of gates in $\overline{C_b}$ is a depth-$1$ Clifford circuit with gates acting on $O(1)$ qubits. Each gate acting on $O(1)$ qubits can be decomposed into $O(1)$ Clifford gates acting on $1$- and $2$ qubits. We may then apply part (iv) of Claim \ref{claim:basic}  to get 
\begin{equation}
F'\sim \mathcal{N}(2^{\sum_{j=1}^{O(D)} 2^{-j}}(c'p^{c''})^{2^{-O(D)}})\sim\mathcal{N}(2(c'p^{c''})^{2^{-O(D)}})
\label{eq:Frate}
\end{equation}
and
\begin{equation}
A(s,y)= \pm \gamma_s \langle y|E' F' \overline{C_b}\rec(s)|\overline{0}^n\rangle.
\label{eq:amp}
\end{equation}
Now let $E'F'=G H$ where $G$ is an $X$-type Pauli and $H$ is a $Z$-type Pauli. Write
\begin{equation}
G=G^1G^2\ldots G^n \qquad G^j\sim \mathcal{N}(q) \qquad \qquad q=O((c'p^{c''})^{2^{-O(D)}})
\label{eq:G}
\end{equation}
where we used part (iii) of Claim \ref{claim:basic}, Eq.~\eqref{eq:Frate} and the fact that $E'\sim \mathcal{N}(p)$.  Now we enforce $q<q_{\mathrm{th}}$ where $q_{\mathrm{th}}$ is the constant noise threshold from Condition \ref{cond:meas}. This is achieved by choosing
\[
p<p_{\mathrm{th}} \qquad \mbox{where} \quad p_{\mathrm{th}}=\Omega(q_{\mathrm{th}}^\beta)
\]
with $\beta = (c'')^{-1} 2^{O(D)}=O(1)$.
Using Eq.~\eqref{eq:amp} and the definition of $G$ we arrive at the probability distribution over outputs $(s,y)$:
\begin{equation}
|A(s,y)|^2=| \gamma_s |^2|\langle y\oplus \mathrm{Supp}(G)|\overline{C_b}\rec(s)|\overline{0}^n\rangle|^2= | \gamma_s |^2|\langle y\oplus \mathrm{Supp}(G)|\phi_{b,s}\rangle|^2
\label{eq:noisyprob}
\end{equation}
where $\phi_{b,s}$ is given by Eq.~\eqref{eq:phiout}. A pair $(s,y)$ which is obtained with positive probability satisfies
\begin{equation}
|\langle y\oplus \mathrm{Supp}(G)|\phi_{b,s}\rangle|^2 >0
\label{eq:phibscond}
\end{equation}
and applying Eq.~\eqref{eq:noisefree} gives
\begin{equation}
\mathcal{R}_U(b,(s,y\oplus \mathrm{Supp}(G)))=1.
\label{eq:Renc}
\end{equation}
Moreover for each $j=1,\ldots, n$ we have
\begin{equation}
y^j \oplus \mathrm{Supp}(G^j)\oplus f^j(b,s) \in \mathcal{L}
\label{eq:inV}
\end{equation}
where $f$ is defined in Eq.~\eqref{eq:Q} and $\mathcal{L}\subseteq \{0,1\}^m$ is the set of 
$m$-qubit basis states that appear in the logical states $|\overline{0}\rangle$, $|\overline{1}\rangle$,
see Eq.~\eqref{Lsubspace}.  Indeed, Eq.~\eqref{eq:inV} follows from Eq.~\eqref{eq:phibscond} and Eq.~\eqref{eq:phiout} which implies that $X(f)Z(h)|\phi_{b,s}\rangle$ is a logical state of the code $\mathcal{Q}_m$.

Let $z^j=\dec(y^j \oplus f^j(b,s))$ and $Z^j=\dec(y^j \oplus \mathrm{Supp}(G^j)\oplus f^j(b,s))$ for each $j=1,2,\ldots n$. Then Eq.~\eqref{eq:Renc} and Definition \ref{definition:enc} imply
\begin{equation}
R_U(b,Z)=1.
\label{eq:RZ}
\end{equation}
Now using Eq.~\eqref{eq:inV}, Condition \ref{cond:meas}, and the fact that $G^j\sim \mathcal{N}(q)$ with $q<q_{th}$, we get that for each $j=1,2,\ldots, n$:
\begin{equation}
\mathrm{Pr}\left[Z^j=z^j\right] \geq 1-\exp(-\Omega(m^c))
\label{eq:union}
\end{equation}
for some constant $c>0$.
By a union bound, we may choose $m=O(\mathrm{poly}(\log(n)))$ such that $Z^j=z^j$ for all $j=1,\ldots, n$ with probability at least $0.99$. Combining with Eq.~\eqref{eq:RZ} gives $R_U(b,z)=1$ with probability at least $0.99$ and plugging this into Definition \ref{definition:enc} we arrive at Eq.~\eqref{eq:pr}.
\end{proof}

\subsubsection{Circuit depth lower bound for classical circuits}
\label{sec:classicalhardnessnoisetolerant}
Consider a classical probabilistic circuit which satisfies the relation $\mathcal{R}_U$ with high probability. The following Theorem establishes a lower bound on the depth $\mathcal{D}$ that such a circuit must have, as a function of the depth required to satisfy the bare relation $R_{U}$.

\begin{theorem}[\textbf{Classical depth lower bound}]
Let $U$ be a controlled Clifford circuit of depth $D$ composed of $k$-qubit gates. Suppose there is a classical probabilistic circuit of depth $\mathcal{D}$ and gates of fan-in at most $K$, such that the input/output pairs~$(b,(s,y))$ of the circuit satisfy $\mathcal{R}_U$ with probability at least $1-p_{F}$ for a random input~$b\in S$ uniformly chosen from some subset~$S\subset \{0,1\}^v$. Then there is another classical probabilistic circuit with gates of fan-in at most $K+2k+\max\{m,m_{\mathrm{anc}}\}$ and depth at most
\[
\mathrm{Depth}\leq \mathcal{D}+D+3,
\] 
such that the input/output pairs~$(b,z)$ satisfy $R_U$ with probability at least $1-p_{F}$, for $b\in S$ chosen uniformly at random. 
\label{thm:depth}
\end{theorem}

\begin{proof}
It suffices to show that from any pair $(b,(s,y))$ satisfying $\mathcal{R}_U(b,(s,y))=1$, we may compute $z\in \{0,1\}^n$ such that $R_U(b,z)=1$ using a classical circuit of depth at most $D+3$ with gates of fan-in at most $2k+\max\{m,m_{\mathrm{anc}}\}$.

So suppose $(b,(s,y))$ satisfies $\mathcal{R}_U(b,(s,y))=1$. From the definition of $\mathcal{R}_U$ we have that $z\in \{0,1\}^n$ defined by
\[
z_i=\dec(y^i\oplus f^i(s,b))
\]
satisfies $R_U(b,z)=1$. Note that the $\dec$ function maps $m$ bits to $1$ bit, so applying this function in parallel to all codeblocks requires only one layer of gates with fan-in $m$. The $\oplus$ gates likewise only require depth $1$. 

It remains to show that $f^{i}$ (defined by Eq.~\eqref{eq:Q}) can be computed from $(b,s)$ in depth $D+1$ using gates of fan-in at most $2k+m_{\mathrm{anc}}$. First recall that $\rec(s)=\otimes_{j=1}^{n} \rec(s^j)$, where $\rec(s^j)$ is a function that depends only on $m_{\mathrm{anc}}$ bits of $s$. Therefore we can classically compute $\rec(s)$ from $s$ using a depth-1 circuit composed of gates with fan-in $m_{\mathrm{anc}}$. Here and below an $n$-qubit Pauli is represented (up to an overall global phase) by a $2n$-bit string in the usual way.

Next we compute $f^i$ from $\rec(s)$ using a circuit of depth $D$ using gates of fan-in at most $2k$. 
Suppose $P$ is a Pauli operator and $C$ is a depth-1 quantum circuit composed of $k$-qubit Clifford gates.
Using the standard stabilizer formalism one can construct a depth-1 classical
circuit with fan-in $2k$ computing  the function
$P \to CP C^{\dagger}$.
Likewise, if $C$ has depth $D$, the classical circuit computing 
the function $P\to CPC^{\dagger}$ has depth $D$ and fan-in $2k$.
By Condition~\ref{cond:logical}, the  circuit $\overline{C}_b$
has depth at most $D$ and consists of  two-qubit Clifford
gates, where each gate is classically controlled by at most $k$ bits of $b$. Thus a function $P,b \to \overline{C}_b P\overline{C}_b^{\dagger}$
can be computed by a classical  depth-$D$ circuit with fan-in $2k$.
Recall that $f^i(s,b)$ is defined by
\[
X(f)Z(h) = \overline{C}_b \rec(s) \overline{C}_b^{\dagger}.
\]
The above shows that
$f^i(s,b)$ can be computed by a classical circuit
with depth $D+1$ and  fan-in at most $2k+m_{\mathrm{anc}}$.
\end{proof}

\subsection{Fault-tolerant quantum advantage\label{sec:ftquantumadvantage}}
Theorems \ref{thm:tol} and \ref{thm:depth} can be applied to obtain a fault-tolerant quantum advantage with shallow circuits from a constant-depth controlled-Clifford circuit that achieves a quantum advantage in the absence of noise. In particular, we may construct the noise-tolerant relation $\mathcal{R}_U$  from the controlled-Clifford circuit. Theorem \ref{thm:tol} states that the noise-tolerant relation can still be solved by a constant depth quantum circuit, whereas Theorem \ref{thm:depth} can be used to lower bound the circuit depth required by classical probabilistic algorithms which satisfy the relation.  Crucially, the code length $m$ and the number of ancillas $m_{\mathrm{anc}}$ 
needed to achieve a  fault-tolerant quantum advantage in Theorem~\ref{thm:tol} 
scales only poly-logarithmically with the size of input/output strings in the relation problem.
Thus, the lower bound of Theorem~\ref{thm:depth}
applies to classical circuits with a poly-logarithmic fan-in. 
We shall see that a quantum advantage established in Section~\ref{sec:1Dmagicsquare} persists
even if the classical circuit may have poly-logarithmic fan-in.

To make this strategy work we need a controlled Clifford circuit that achieves a quantum advantage in the absence of noise.  Below we will use the relation $R^{MS}_U$ for the 1D Magic Square Problem defined in Section~\ref{sec:1Dmagicsquare}.  Note that we could have instead used other relation problems known to achieve a quantum advantage, such as the one described in Ref.~\cite{bragokoe18}.

Recall from Theorem \ref{thm:1dq} that the quantum circuit $U$ which solves the 1D Magic Square Problem is a depth $D=O(1)$ controlled Clifford circuit composed of $k$-qubit gates with $k=O(1)$. We can therefore define its noise-tolerant version $\mathcal{R}^{MS}_U$  which can be satisfied with probability close to one by the input-output statistics of a noisy constant-depth quantum circuit (by Theorem \ref{thm:tol}). On the other hand, we may use Theorem \ref{thm:depth} to establish the following lower bound on the depth of any classical circuit satisfying this relation.
\begin{theorem}\label{thm:lowerboundcircuitdepthDMS}
Suppose $\mathcal{C}$ is a classical probabilistic circuit composed of gates with fan-in at most $K=O(\mathrm{poly}(\log(n)))$  which satisfies the relation $\mathcal{R}^{MS}_U$ with probability greater than $9/10$ for inputs chosen uniformly at random  from the subset~$S\subset \{0,1\}^{4n}$ defined after Theorem~\ref{thm:1dq}.  Then its depth satisfies
\[
\mathcal{D}\geq \Omega\left(\frac{\log(n)}{\log(\log(n))}\right).
\]
\end{theorem}
\begin{proof}
Combining Theorems \ref{thm:lowerbound} and \ref{thm:depth} we get
\[
\mathcal{D}\geq \frac{\log(0.0001 n)}{2\log{K'}}-O(1)
\]
where $K'=K+2k+\max\{m,m_{\mathrm{anc}}\}=O(\mathrm{poly}(\log(n)))$.  Here we used the fact that the number of physical qubits $m$ per logical qubit and the number of ancilla qubits $m_{\mathrm{anc}}$ needed for fault-tolerant state preparation are both polylogarithmic in $n$ (cf. Theorem \ref{thm:tol} and Condition \ref{cond:shot}).
\end{proof}

This shows that the noise-tolerant 1D Magic Square relation $\mathcal{R}^{MS}_U$ separates noisy constant-depth quantum circuits from noise-free constant-depth classical circuits with at most polylogarithmic fan-in. At this point we cannot say very much about the geometric locality of the resulting quantum algorithm, however.  Recall that the (ideal) circuit~$U$ for the 1D Magic Square Problem introduced in Section~\ref{sec:1Dmagicsquare} only has gates acting on qubits located within a neighborhood of diameter~$O(1)$ on a 1D line graph; furthermore, in the given setup, every Clifford appearing in the circuit is controlled by input bits located near the qubits it acts on. However, the geometric locality of the ideal controlled-Clifford circuit $U$ solving the bare relation $R^{MS}_U$ is not necessarily inherited by the quantum circuit which solves the noise-tolerant relation $\mathcal{R}^{MS}_{U}$. Indeed, the circuit given in Fig.~\ref{fig:algcircuit}  may  require geometric non-locality in three different ways:
\begin{enumerate}[(a)]
\item
The constant depth Clifford circuit $W$ used for state preparation may be geometrically nonlocal. 
\item\label{it:firstissuegeometric}
The logical Clifford gates in the circuit $\overline{C}_b$ may be geometrically non-local.
\item\label{it:secondissuegeometric}
If a subset of input bits controls a certain gate in the ideal circuit~$U$, this set of input bits ends up potentially controlling unitaries acting on all qubits within one (or several) codeblock(s). Thus the classical control may be geometrically non-local. 
\end{enumerate}
In Section \ref{sec:3D} we will address these potential sources of geometric non-locality.

For this purpose it will be convenient to consider ideal and noise-tolerant relations where the initial basis state $|0^n\ra$ is replaced by an entangled state $|\Phi^{\otimes n/2}\ra$, where $\Phi$ is the Bell state.
Note that initial entangled states are more natural in the context of non-local games.
Accordingly, the bare relation $R_U(b,z)$ is satisfied iff
\[
p_b(z)=|\langle z|C_b|\Phi^{\otimes n/2}\rangle|^2>0.
\]
Here the input $b$, the output $z$, and the controlled Clifford circuit $C_b$
are the same as in Definition~\ref{def:ideal}.
The corresponding  fault-tolerant relation ${\cal R}_U(b,s,y)$ is based
on Algorithm~1 where the first step prepares
$n/2$ copies of the logical Bell state $\overline{\Phi}$ 
instead of the logical basis state. The $i$-th copy of the Bell state  is encoded into codeblocks $B^{2i-1}$ and $B^{2i}$.
The codeblocks $B^{2i-1}B^{2i}$ share the same subset of $m_{anc}$ qubits
and are initialized in the state $|0^{2m}\rangle\otimes |0^{m_{\mathrm{anc}}}\rangle$.
Each pair $B^{2i-1}B^{2i}$ then applies a constant-depth Clifford circuit $W$ satisfying 
a suitable single-shot state preparation property (stated below)
and measures the $m_{\mathrm{anc}}$ ancilla qubits to get an outcome $s^i \in \{0,1\}^{m_{\mathrm{anc}}}$. 
A noise-tolerant relation is then defined according to Definition~\ref{definition:enc},
with $s=s^1s^1s^2s^2\cdots s^{n/2}s^{n/2}$.
The modified version of Condition~2 tailored to preparation of the logical Bell state is as follows.

\setcounter{cond}{1}
\begin{cond}[\textbf{Single-shot Bell state preparation}]
Let  $c,c',c'',d$ be some universal constants. 
For each code $\calQ_m$ in the family
there must exist an integer $m_{\mathrm{anc}}\leq m^{c}$, a 
depth-$d$ Clifford circuit $W$ acting on $2m+m_{\mathrm{anc}}$ qubits, recovery and repair functions
\begin{align}
\begin{matrix}
\rec: &\{0,1\}^{m_{\mathrm{anc}}}&\rightarrow &\mathrm{Pauli}(2m)\\
\repair: & \mathrm{Pauli}(2m+m_{\mathrm{anc}})&\rightarrow &\mathrm{Pauli}(2m)
\end{matrix}
\end{align}
such that 
\begin{equation}
\left(\rec(s)\otimes |s\rangle\langle s|\right) EW(|0^{2m}\rangle\otimes |0^{m_{\mathrm{anc}}}\rangle)= \gamma_s (\repair(E) |\overline{\Phi}\rangle)\otimes |s\rangle
\end{equation}
for all $s\in \{0,1\}^{m_{\mathrm{anc}}}$ and $E\in \mathrm{Pauli}(2m+m_{\mathrm{anc}})$. 
Here $\gamma_s \in \mathbb{C}$ is the normalization.
Furthermore, the following must hold for all noise rates $p\in [0,1]$.
Suppose $E\sim \mathcal{N}(p)$. 
Then $\repair(E)\sim \mathcal{N}(c'p^{c''})$. 
\label{cond:shot}
\end{cond}
\noindent
The above modifications do not alter the proof of Theorems~\ref{thm:tol},\ref{thm:depth}
in any substantial way. Applied to the 1D Magic Square Problem,
 using this single-shot Bell state preparation eliminates the need for the initial layer of $\mathrm{CNOT}$ gates in the ideal circuit~$U$ of Fig.~\ref{fig:1d}. In Section \ref{sec:3D} we will argue that the remaining entangling gates,  stemming from the classically controlled gates  $U(\alpha)$ and $V(\beta)$ determining the measurement bases, as well as the gates  $W(\beta,\alpha)$ responsible for the entanglement SWAPS, can all be implemented in a geometrically local way. Finally, we show how the classical control can also be made geometrically local, if desired, by a suitable modification of the relation.

\section{Quantum code constructions\label{sec:codes}}

The quantum advantage established in  Theorems~\ref{thm:lowerbound}, \ref{thm:tol}, \ref{thm:depth} and~\ref{thm:lowerboundcircuitdepthDMS}
hinges on the existence of quantum error correcting codes 
with suitable properties.
Recall that we need a family of CSS-type codes that enable a depth-$1$ implementation
of the logical gates $S,H$, single-shot logical state preparation,
and a decoding algorithm that can tolerate local stochastic noise
with a small enough rate, see Conditions~1,2 and~3 in Section~\ref{subs:codes} 
for formal statements. In this section we show how to satisfy these conditions
using the standard 2D surface codes~\cite{bravyi1998quantum}. The latter are among the most promising codes for 
quantum fault-tolerance applications 
due to their high error threshold, efficient decoding algorithms,
and simple syndrome extraction circuits~\cite{dennis2002topological,raussendorf2007fault,Fowler2009}.

A distance-$d$ surface code encodes one logical qubit into
$m=d^2+(d-1)^2$ physical qubits located at edges
of a square lattice of size $d\times d$. 
The lattice has smooth top/bottom boundaries and rough
left/right boundaries, as shown at Fig.~\ref{fig:surface_code}.
Let us agree that a qubit is placed at the center of each edge.
We shall use notations $\calV$, $\calE$, and $\calF$
for the sets of vertices, edges, and faces of the surface code
lattice. The code is defined by a set of stabilizer generators
$A_v$ and $B_f$ associated with vertices $v\in \calV$ and 
faces $f\in \calF$. Specifically,
$A_v=\prod_{e\in \delta(v)} X_e$ and $B_f = \prod_{e\in \partial f} Z_e$.
Here $\delta(v)$  is the subset of edges incident to a vertex $v$
and $\partial f$ is the boundary of a face $f$.
Logical Pauli operators  can be chosen as
\begin{equation}
\overline{Z} = \prod_{e\in \calE^{\mathrm{diag}}} Z_e
\quad \mbox{and} \quad
\overline{X} = \prod_{e\in \calE^{\mathrm{diag}}} X_e,
\label{eq:logicalXZ}
\end{equation}
where $\calE^{\mathrm{diag}}\subset \calE$ is the subset of qubits lying
on the main diagonal of the lattice, see Fig.~\ref{fig:diagonal}.
It can be easily verified that $\overline{Z}\, \overline{X}=-\overline{X}\, \overline{Z}$.
Furthermore, $\overline{Z}$ and $\overline{X}$ commute with all stabilizer generators.

\begin{figure}[h]
\centerline{\includegraphics[height=4.5cm]{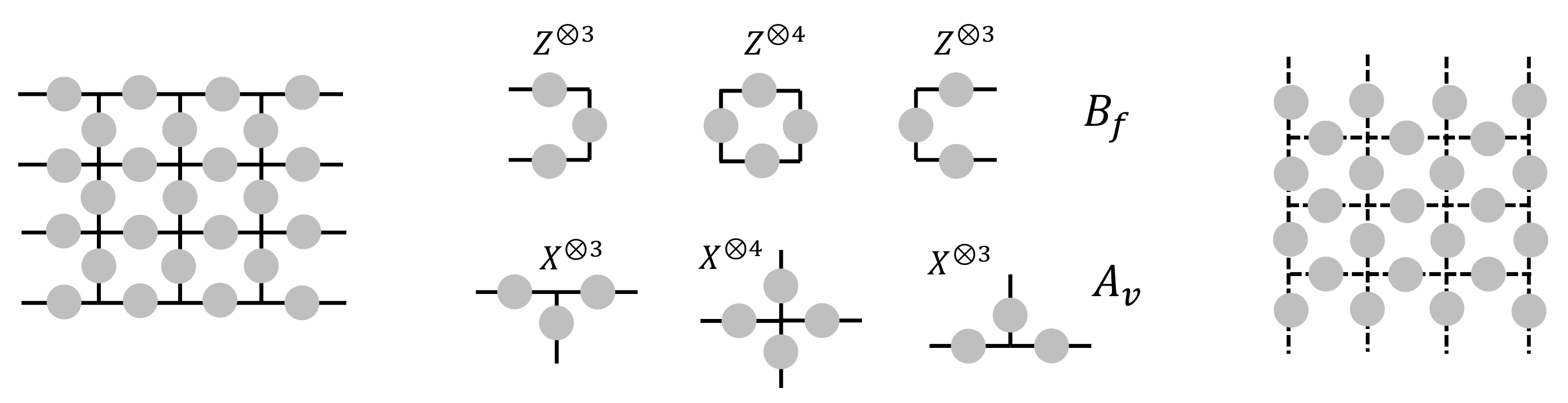}}
\caption{Example of the $d=4$ surface code lattice (left)
and the corresponding dual lattice (right). 
Stabilizer generators are shown at the center.\label{fig:surface_code}}
\end{figure}

\subsection{Geometrically local circuits for logical Clifford gates}
The logical gates $H$, $S$ can be implemented by 
depth-$1$ Clifford circuits using the lattice folding trick of 
Ref.~\cite{moussa2016transversal}.
For completeness, we restate the result of Ref.~\cite{moussa2016transversal} below.
We begin by introducing some extra notations.
Let $\sigma$ be a  reflection of the ambient space  ${\mathbb R}^2$
against the main diagonal of the surface
lattice, see Fig.~\ref{fig:diagonal}. Note that $\sigma$ maps
the surface code lattice to its dual and vice verse.
In other words, $\sigma$ defines bijective maps $\calE \to \calE$, $\calV\to \calF$, and $\calF\to \calV$.
More formally, suppose $v\in \calV$, $f\in \calF$,
and $e,e'\in \calE$.
Set $\sigma(v)=f$ if $\sigma(v)$ is the center of $f$.
Set $\sigma(f)=v$ if $\sigma$ maps the center of $f$ to $v$.
Set $\sigma(e)=e'$ if $\sigma$ maps the center of $e$ to the center of $e'$.

Consider the following operators:
\begin{equation}
\overline{H} = H^{\otimes m} \prod_{e\in \calE^{\mathrm{top}} }  \mathrm{SWAP}_{e,\sigma(e)}
\quad
\mbox{and} 
\quad
\overline{S} = \prod_{e\in \calE^{\mathrm{diag}} } S_e^{\phi(e)} \prod_{e\in \calE^{\mathrm{top}} } 
 \mathrm{CZ}_{e,\sigma(e)}.
\label{eq:logicalHS}
\end{equation}
Here $\calE^{\mathrm{top}} \subset \calE$ denotes  the subset of qubits lying above the main diagonal,
see Fig.~(\ref{fig:diagonal}),
$\phi(e)=+1$ for horizontal edges, and $\phi(e)=-1$ for vertical edges.

\begin{figure}[h]
\centerline{\includegraphics[height=4.5cm]{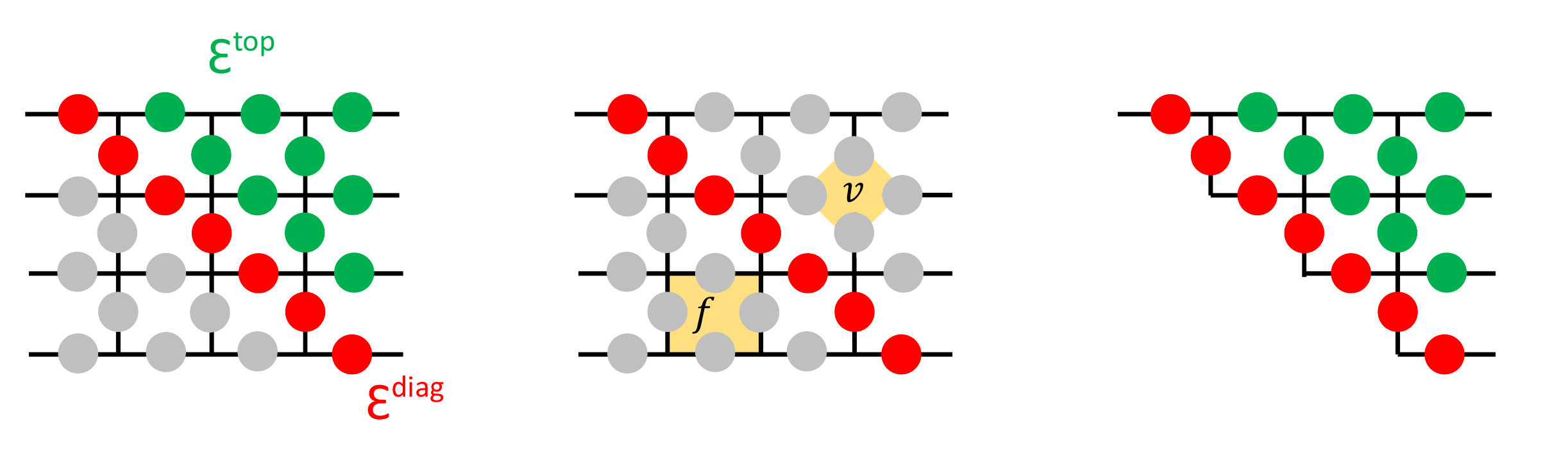}}
\caption{{\em Left:} Subsets of qubits $\calE^{\mathrm{diag}}$
and $\calE^{\mathrm{top}}$.
{\em Center:} Example of a face $f$ and a vertex $v$ mapped
to each other by the reflection against the main diagonal, that is,
$v=\sigma(f)$ and $f=\sigma(v)$.
{\em Right:} Folded surface code. Red and green circles
indicate single qubits and pairs of qubits   respectively.
\label{fig:diagonal}}
\end{figure}

\begin{lemma}[\bf \cite{moussa2016transversal}]
\label{lemma:HS}
The circuits $\overline{H}$ and $\overline{S}$ defined in Eq.~\eqref{eq:logicalHS} implement the logical Clifford
gates $H$ and $S$ respectively.
\end{lemma}
\begin{proof}
It suffices to check that the conjugation by $\overline{H}$ and $\overline{S}$
maps stabilizers $A_v, B_f$ to products of stabilizers
and implements the desired transformation of the logical Pauli operators, that is,
\begin{equation}
\label{eq:logicalZX1}
\overline{H} \, \overline{Z}\, \overline{H} = \overline{X},
\quad
\overline{H} \, \overline{X}\, \overline{H} = \overline{Z},
\end{equation}
\begin{equation}
\label{eq:logicalZX2}
\overline{S} \, \overline{Z}\, \overline{S}^{-1} = \overline{Z},
\quad
\overline{S} \, \overline{X}\, \overline{S}^{-1} =i \overline{X}\, \overline{Z}.
\end{equation}
Consider first the circuit $\overline{H}$.
Since $\overline{Z}$ and $\overline{X}$ have support only on the main diagonal,
the product of SWAP gates in Eq.~\eqref{eq:logicalHS} has trivial
action on $\overline{Z}$ and $\overline{X}$.
The bitwise Hadamard in Eq.~\eqref{eq:logicalHS} exchanges
all $X$ and $Z$. This proves
Eq.~\eqref{eq:logicalZX1}.
Next consider some vertex stabilizer $A_v$. 
The product of SWAP gates in Eq.~\eqref{eq:logicalHS} maps
the support of $A_v$ to the face $f=\sigma(v)$,
see Fig.~\ref{fig:diagonal} for an example.
The bitwise Hadamard in Eq.~\eqref{eq:logicalHS}
changes each Pauli  $X$ to $Z$.
Thus $\overline{H}\, A_v \overline{H} = B_{\sigma(v)}$.
A similar argument shows that $\overline{H} B_f \overline{H} = A_{\sigma(f)}$.
We conclude that $\overline{H}$ implements a logical $H$ gate.

Consider now the circuit $\overline{S}$. Clearly, $\overline{S}$ commutes
with $\overline{Z}$ and face stabilizers $B_f$.
Using identities $SXS^{-1} = Y$ and $S^{-1} XS = -Y$ one easily gets
\[
\overline{S} \, \overline{X}\, \overline{S}^{-1}
=(-1)^{d-1} \prod_{e\in \calE^{\mathrm{diag}} } Y_e = i \overline{X}\, \overline{Z}
\]
proving Eq.~\eqref{eq:logicalZX2}.
It remains to check that $\overline{S}$ maps vertex stabilizers $A_v$
to products of stabilizers.
We claim that 
\begin{equation}
\label{eq:SAS}
\overline{S}\, A_v \overline{S}^{-1} = A_v B_{\sigma(v)}
\end{equation}
for all $v\in \calV$. Indeed, suppose first that
$A_v$ has no overlap with the main diagonal. 
Then one can ignore the $S$-gates in Eq.~\eqref{eq:logicalHS}
and $\overline{S} X_e \overline{S}^{-1} = \mathrm{CZ}_{e,\sigma(e)}X_e\mathrm{CZ}_{e,\sigma(e)}^{-1}=X_e Z_{\sigma(e)}$ 
for any $e$ in the support of $A_v$. This proves Eq.~\eqref{eq:SAS}.
In the remaining case, 
$A_v$ overlaps with the main diagonal on some consecutive pair of edges
$e,e'$.  Since the product of $S$-gates in Eq.~\eqref{eq:logicalHS}
alternates between $S$ and $S^{-1}$, one has
\[
\overline{S} (X_e X_{e'}) \overline{S}^{-1} = -Y_e Y_{e'} = (X_e X_{e'}) (Z_e Z_{e'})
=(X_e X_{e'}) (Z_{\sigma(e)} Z_{\sigma(e')}).
\]
Here we noted that $\sigma(e)=e$ for all edges on the main diagonal.
Assume for concreteness that $A_v$ is a weight-4 stabilizer. Then
$A_v = X_e X_{e'} X_g X_h$, where $g,h\in \calE\setminus \calE^{\mathrm{diag}}$.
The product of CZ-gates in Eq.~\eqref{eq:logicalHS} 
maps $X_g $ to $X_g Z_{\sigma(g)}$
and maps $X_h$ to $X_h Z_{\sigma(h)}$.
Thus
\[
\overline{S} A_v \overline{S}^{-1} = (\overline{S} X_e X_{e'} \overline{S}^{-1})
\cdot (\overline{S} X_g X_h \overline{S}^{-1}) = 
(X_e X_{e'}) (Z_{\sigma(e)} Z_{\sigma(e')}) \cdot
(X_g X_h) (Z_{\sigma(g)} Z_{\sigma(h)}) = A_v B_{\sigma(v)}.
\]
The case of weight-3 stabilizers $A_v$ is completely analogous.
We conclude that $\overline{S}$ implements a logical $S$-gate.
\end{proof}

As proposed in Ref.~\cite{moussa2016transversal}, 
quantum circuits implementing the logical gates $\overline{H}$ and $\overline{S}$ 
can be made
geometrically local by folding the surface code lattice against the main diagonal,
as shown in Fig.~\ref{fig:diagonal}.
The folded lattice has a pair of qubits $(e,\sigma(e))$ with
$e\in \calE^{\mathrm{top}}$  located at the 
same edge.
Now  each SWAP and CZ gate in Eq.~\eqref{eq:logicalHS}
acts on qubits located at the same edge,
i.e. both logical gates $\overline{H}$ and $\overline{S}$ can be implemented
by a depth-$1$ circuit with geometrically local gates.
We shall make use of the folded surface code in Section~\ref{sec:3D}
to construct a 3D embedding of the encoded version of the quantum circuit solving the
1D Magic Square Problem.

\subsection{Error threshold\label{sec:singleshotlogicalmeas}}
Next let us construct a decoding function $\dec:\{0,1\}^m\rightarrow \{0,1\}$
that satisfies Condition~\ref{cond:meas} of Section~\ref{subs:codes}.
We shall use the minimum weight  decoder
and show that it has a non-zero error threshold for local stochastic noise.
The proof follows ideas of Refs.~\cite{dennis2002topological,fowler2012proof}. 
 Fix an error correction function $\cor\, : \, \{0,1\}^m\to \{0,1\}^m$
such that $y=\cor(x)$ is a minimum
weight bit string that has the same $Z$-syndrome as $x$,
that is, $\langle y|B_f|y\rangle = \langle x|B_f|x\rangle$ for any face $f$.
Note that $\cor(x)$ depends only on the $Z$-syndrome of $x$.
Define
\begin{equation}
\label{DecSC}
\parity(x) = \sum_{e\in \calE^{\mathrm{diag}}} x_e
\qquad \mbox{and} \qquad
\dec(x) =\parity(\cor(x)\oplus x).
\end{equation}
Here the sum is evaluated modulo two.
If $x$ has the trivial $Z$-syndrome
(i.e. $\langle x|B_f|x\rangle=1$ for all $f$)
 then  $\cor(x)=0^m$ 
and thus $\parity(x) = \dec(x)$ proving
Eq.~\eqref{eq:paritydecode}.
To prove
Eq.~\eqref{eq:threshold} we need 
\begin{lemma}
\label{lemma:threshold}
Consider a random $X$-type error $E\sim \calN(q)$ with $q\le 0.01$.
Let $y\equiv \mathrm{Supp}(E)$. Then
\begin{equation}
\mathrm{Pr}_{E} \left[\parity(\cor(y) \oplus y)=1\right]\leq
3d(6q^{1/2})^d.
\label{eq:threshold1}
\end{equation}
\end{lemma}
As  discussed below (see Eq.~\eqref{eq:measuredecoding}), the lemma implies that with high probability,
the change in the parity (and thus the estimated logical $\overline{Z}$-eigenvalue) 
caused by an error is properly accounted for by the correction function~$\cor$.
\begin{proof}
Let  $r\equiv \cor(y)$.
By definition of the function $\cor$, the string
$r\oplus y$ has trivial $Z$-syndome.
Note that a bit string has trivial $Z$-syndrome iff the corresponding subset
of edges is a cycle in the dual surface code lattice, see Fig.~\ref{fig:surface_code}.
Such cycle can be represented
(non-uniquely) as an edge-disjoint union of closed loops
and  paths terminating at the dangling edges.
The dual surface code lattice has dangling edges only at the 
top/bottom boundaries, see Fig.~\ref{fig:surface_code}.
Consider the event
\[
\mathrm{FAIL} =\{ E\, : \,\parity(r\oplus y)=1 \}\ .
\]
Direct inspection shows that any closed loop has even overlap with $\calE^{\mathrm{diag}}$.
Likewise, any path having both endpoints at the same boundary
has even overlap with $\calE^{\mathrm{diag}}$.
Since $r\oplus y$ is a cycle, $\mathrm{FAIL}$ implies that any
decomposition of $r\oplus y$ into edge-disjoint loops and paths
contains at least one path $\delta$ connecting top/bottom boundaries.
Note that such a path must contain at least $d$ edges.
Let $\Delta(k)$ be the set of all  paths of length $k$ connecting
top/bottom boundaries.
By the union bound,
\[
\mathrm{Pr}_E[\mathrm{FAIL}] \le \sum_{k=d}^\infty \;\sum_{\delta \in \Delta(k)} \mathrm{Pr}_E[\delta \subseteq r\oplus y].
\]
Suppose $\delta\in \Delta(k)$ satisfies $\delta \subseteq r\oplus y$.
We claim that at least half of the edges of $\delta$ are not contained in $r$.
Indeed, otherwise one could replace $r$ by $r\oplus \delta$
reducing the weight of $r$ without changing its $Z$-syndrome
(since any path $\delta\in\Delta(k)$ has trivial $Z$-syndrome). This would contradict  the minimality of $r$.
Thus  at least half of the edges of $\delta$ are contained in $y$,
that is, $\delta' \equiv \delta\cap y$ has size at least $k/2$. 
From $E\sim \calN(q)$ one gets $\mathrm{Pr}_E[\delta' \subseteq y]\le q^{|\delta'|}\le q^{k/2}$
for any fixed $\delta'$.
Noting that  $\delta$ has at most $2^k$ subsets $\delta'$, 
and $|\Delta(k)|\le d3^k$ (since the surface code lattice has vertex degree$\le 4$
and since there are $d$ choices for the starting edge of $\delta$),
one gets
\[
\mathrm{Pr}_E[\mathrm{FAIL}]
 \le d \sum_{k=d}^\infty (6q^{1/2})^k
\le 3d (6q^{1/2})^d
\]
provided that $q\le 0.01$. 
\end{proof}
We can now  verify  Eq.~\eqref{eq:threshold}. 
Suppose $v\in \{0,1\}^m$ is a random string such that $X(v)\sim \calN(q)$
is a local stochastic noise.
By definition of the subspace $\calL$, 
any vector $x\in \calL$ has trivial $Z$-syndrome.
Thus $v$ and  $x\oplus v$ have the same  $Z$-syndromes.
Accordingly,   $\cor(x\oplus v)=\cor(v)$ for any $v\in \{0,1\}^m$
and any $x\in \calL$. Therefore
\[
\parity(x)\oplus \dec(x\oplus v)
=\parity(x) \oplus\parity( \cor(x\oplus v)\oplus x\oplus v)
=\parity( \cor(v)\oplus v)
\]
because the parity is linear and $\parity(x)\oplus\parity(x)=0$. Thus 
\begin{align}
\parity(x)= \dec(x\oplus v)\quad \textrm{ if and only if }\quad 
 \parity( \cor(v)\oplus v)=0\ .\label{eq:measuredecoding}
 \end{align}
By Lemma~\ref{lemma:threshold}, the probability of this event is at most
$3d(6q^{1/2})^d \le 3\exp{(-0.2d)}$ for all $q\le 0.01$
and all $d\ge 7$. 
To summarize, we have  shown the following:
\begin{theorem}[{\bf Error threshold for surface codes}]\label{thm:singleshotlogicalmeasurementsurface}
The decoding and parity functions~$(\dec,\parity)$ defined by~\eqref{DecSC} for the distance-$d$ surface code satisfy the correctness condition~\eqref{eq:paritydecode} in the noise-free case. Furthermore, 
suppose $d\ge 7$  and $v\in \{0,1\}^m$ is a random string such that 
$X(v) \sim \calN(q)$ for some $q\le 0.01$. Then
\begin{equation}
\mathrm{Pr} \left[\dec(x\oplus v)=\parity(x)\right]\geq  1- 3 \exp{(-0.2d)}
\label{eq:thresholdtoric}
\end{equation}
for all $x\in \calL$.
\end{theorem}
We note that the local stochastic noise model can account
for correlated errors, such as those introduced by
the logical circuits 
$\overline{H}$ and $\overline{S}$ defined in Eq.~\eqref{eq:logicalHS}.
Theorem~\ref{eq:thresholdtoric} shows that the surface code has error
threshold of at least $1\%$ for such correlated noise.

\subsection{Single-shot logical state preparation\label{sec:singleshotlogicalstateprep}}
It remains to show that the surface code enables single-shot logical state preparation,  as stated in Condition~\ref{cond:shot}.
Here  we consider the Bell state version of Condition~\ref{cond:shot}, see  Section~\ref{sec:ftquantumadvantage}.
\begin{figure}[h]
\centerline{\includegraphics[height=4cm]{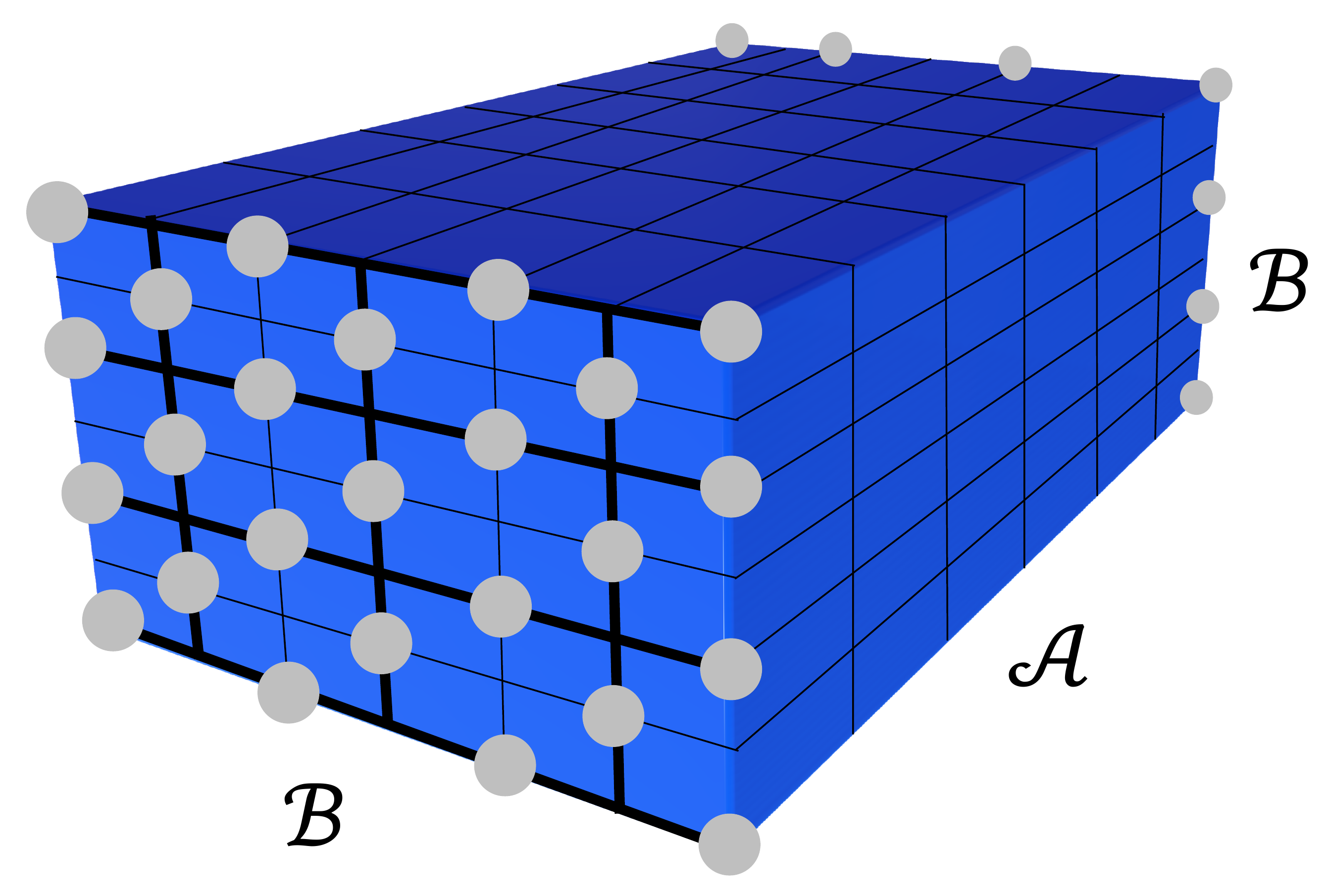}}
\caption{The cubic lattice $\calC=\calA\calB$. 
Region $\calB$ (gray circles) represents the Bell state encoded
by two surface codes located on the left and the right faces of $\calC$.
Region $\calA$ represents ancillary qubits.
\label{fig:surface_code3D}}
\end{figure}

Consider a 3D cubic lattice $\calC$
of size $r\times r\times r$, where $r=2d-1$.
We place qubits at sites of~$\calC$.
The Bell state $\Phi$ will be encoded into
a pair of distance-$d$ surface codes located on two opposite faces
of $\calC$, as shown on Fig.~\ref{fig:surface_code3D}.
Each surface code contains $m$ qubits.
Let $\calB\subset \calC$ be the subset of $2m$ qubits encoding the Bell state.
The rest of the lattice $\calA=\calC\setminus \calB$ 
represents ancillary qubits that are used
by the single-shot preparation  procedure.
Accordingly, $|\calA|=m_{anc}$. For a formal definition of $\calA,\calB,\calC$ see Section~\ref{sec:3Dcluster}.

Logical Bell state preparation  proceeds in three stages.
First, one initializes each qubit of $\calC$ in the state $|0\ra$
and applies a suitable constant-depth Clifford circuit $W$ obtaining 
a state $W|0\ra_{\calC}$.
Next one measures each ancillary qubit in the $Z$-basis.
Let $s\in \{0,1\}^{|\calA|}$ be the measurement outcome.
Finally, one applies a suitable Pauli recovery operator
$\rec{(s)}$ to the region $\calB$ obtaining a state
\[
(|s\ra\la s|_\calA \otimes \rec{(s)}_\calB) W|0\ra_\calC = \gamma_s |s\ra_\calA \otimes |\overline{\Phi}\ra_\calB.
\]
Here $\gamma_s\in \mathbb{C}$ is a normalization factor. A noisy version of this protocol may include initialization, gate, and measurement
errors. By Lemma~\ref{lem:noisemodelcliffordcircuit},
it suffices to consider a single Pauli error $E$  that occurs
immediately before the measurement.
Thus the noisy implementation
outputs the measured string $s$ and a state
$(|s\ra\la s|_\calA \otimes \rec{(s)}_\calB) EW|0\ra_\calC$,
where $E$ is a random Pauli error acting on $\calC$.
The difference between the final state of the noisy protocol
and the desired Bell state can be quantified using a repair
function $\repair{(E)}$ which
returns a Pauli operator acting on $\calB$ such that 
\be
\label{RecRep}
(|s\ra\la s|_\calA \otimes \rec{(s)}_\calB) EW|0\ra_\calC = \gamma_s |s\ra_\calA \otimes
 \repair{(E)} |\overline{\Phi}\ra_\calB
\ee
for all $s$ and $E$.  
The desired single-shot preparation property
is established in the following theorem.
\begin{theorem}[{\bf Single-shot Bell state preparation for surface codes}]
\label{thm:3D}
Let $d\ge 4$ be the desired surface code distance and $\calC$ be the cubic lattice of linear size $2d-1$
with one qubit per site. Suppose $E\sim \calN(p)$  is a local stochastic Pauli error acting on $\calC$. There exists a depth-$6$ Clifford circuit $W$ that uses
only nearest-neighbor gates on the lattice $\calC$, as well as 
recovery and repair functions 
$\rec{(s)}$ and $\repair{(E)}$ 
satisfying the logical Bell state preparation condition~\eqref{RecRep} for the distance-$d$ surface code such that 
$\repair{(E)}\sim  \calN(11 p^{1/128})$.
\end{theorem}
We shall divide the proof of the theorem
into two parts. The first part, presented  in the rest of this section, gives a general 
recipe  for choosing the repair and recovery functions.
We show that the logical  Bell state can be prepared in a single-shot fashion
whenever the repair function  obeys a certain  ``lifting" property.
The surface code structure and the lattice geometry play no role in this part of the proof. 
Thus we anticipate that the same recipe may be applied to other stabilizer codes
(as well as other types of logical states).
The second part of the proof, presented in Section~\ref{sec:3Dcluster},
deals specifically with the surface code.
This part relies crucially on ideas introduced in Ref.~\cite{raussendorfetal05}.

Given a state $\psi$ and a Pauli operator $P$, we say that $P$ is a stabilizer
of $\psi$ if $P|\psi\ra= |\psi\ra$. Stabilizers of any state generate an abelian
subgroup of the Pauli group.
Let $\calS$ be the stabilizer group of the state $W|0\ra_\calC$.
It  has generators 
$WZ_u W^\dag$ with $u\in \calC$.
 We shall identify two subgroups
\be
\calS_0 \subseteq \calS_1 \subseteq \calS
\ee
with the following properties.
\begin{enumerate}[(i)]
\item Any element of $\calS_0$ has the form $Z(\alpha)_\calA \otimes I_\calB$\label{it:firstpropprep}
for some subset $\alpha \subseteq \calA$.
\item Any element of $\calS_1$ has the form $Z(\alpha)_\calA \otimes S_\calB$, 
where $S$ is a stabilizer of  $\overline{\Phi}$ and $\alpha\subseteq \calA$.
\item If $S$ is a stabilizer of $\overline{\Phi}$ then 
$\calS_1$ contains an element $Z(\alpha)_\calA \otimes S_\calB$ for some $\alpha\subseteq \calA$.\label{it:lastpropprep}
\end{enumerate}
In Section~\ref{sec:3Dcluster}, we will discuss an explicit instantiation of these stabilizer groups. 

\subsubsection{Construction of recovery and repair functions}
\label{subsub:construction}
In the following, we analyze the  effect of errors, starting from the special case when an error~$E$ acts non-trivially only on~$\calA$. In particular, our aim here is to describe the construction of the functions~$\rec{(s)}$ and~$\repair{(E)}$. The stabilizer groups $\mathcal{S}_0$ and $\mathcal{S}_1$ play distinct roles, as follows. 
 
Elements of $\calS_0$ will be used to diagnose errors. 
Namely, suppose $S_0^1,\ldots,S_0^k$ are generators
of $\calS_0$. 
Consider the noisy state $EW|0\ra_\calC$.
Define an $\calS_0$-syndrome of $E$ 
as a bit string $\syn_0(E)\in \{0,1\}^k$ such that 
the $i$-th bit of $\syn_0(E)$ is one if
$E$ anti-commutes with $S_0^i$ and zero
otherwise. Importantly, the $\calS_0$-syndrome can be inferred
from the measurement outcome $s$. Indeed,
$S_0^i$ are $Z$-type Paulis acting on $\calA$ 
that stabilize the ideal state $W|0\ra_\calC$.
Since all qubits of $\calA$ are measured in the $Z$-basis, the $i$-th bit
of $\syn_0(E)$ is 
\[
\syn_0(E)_i=\parity(s, \mathrm{Supp}(S_0^i)).
\]
Here and below we use the notation
\[
\parity(s,L)\equiv \sum_{u\in L} s_u \pmod{2}.
\]

Elements of $\calS_1$ can be identified with stabilizers of the desired
final state $\overline{\Phi}$ on $\mathcal{B}$. More precisely,
suppose $B^1,\ldots,B^{2m}$ is a complete set of stabilizers
of $\overline{\Phi}$  (recall that $|\calB|=2m$).
By condition (iii), one can choose Pauli operators
$S_1^1,\ldots,S_1^{2m} \in \calS_1$  such that the restriction of $S_1^i$
onto $\calB$ coincides with $B^i$. Furthermore, $S_1^i$  acts on $\calA$
only by Pauli $Z$.
Since $S_1^i$ stabilizes the state $W|0\ra_\calC$
and acts on $\calA$ only by Pauli $Z$,
we conclude that $S_1^i$ also stabilizes the state 
$(|s\ra\la s|_\calA \otimes I_\calB)W|0\ra_\calC$.
It follows that
\be
\label{outline_eq0}
(I_\calA \otimes B^i_\calB)(|s\ra\la s|_\calA \otimes I_\calB)W|0\ra_\calC = (-1)^{\sigma_i}
(|s\ra\la s|_\calA \otimes I_\calB)W|0\ra_\calC
\ee
where $\sigma\in \{0,1\}^{2m}$ is defined by 
\[
\sigma_i =\parity(s,  \mathrm{Supp}(S_1^i)\cap \calA).
\]
Given a Pauli error $F$ acting on $\calC$, define 
an $\calS_1$-syndrome of $F$ 
as a bit string $\syn_1(F)\in \{0,1\}^{2m}$ such that 
the $i$-th bit of $\syn_1(F)$ is one if
$F$ anti-commutes with $S_1^i$ and zero
otherwise. From Eq.~\eqref{outline_eq0} one gets 
\[
(I_\calA \otimes B^i_\calB)(|s\ra\la s|_\calA \otimes I_\calB)EW|0\ra_\calC = (-1)^{\sigma_i +  \syn_1(E)_i} (|s\ra\la s|_\calA \otimes I_\calB)EW|0\ra_\calC.
\]
The above equation shows that the reduced state of $\calB$
after the measurement is an eigenvector of $B^i$ with an eigenvalue
$(-1)^{\sigma_i +  \syn_1(E)_i}$.
Thus, to ensure that the final state of $\calB$ is proportional to~$\overline{\Phi}$ it would suffice to choose the recovery $\rec{(s)}$
as a Pauli operator acting on $\calB$ with $\calS_1$-syndrome $\sigma\oplus \syn_1(E)$.

Unfortunately, the syndrome $\syn_1(E)$ 
cannot be inferred from the measurement outcome $s$. Instead,
we shall compute an approximate version of $\syn_1(E)$ 
by replacing $E$ with a suitable Pauli operator $M=M(s)$ which depends only on 
$s$ and acts as a proxy for the actual error~$E$.
  The definition of the recovery~$\rec{(s)}$
 based on the proxy $M$ is summarized in 
Fig.~\ref{fig:state_prep}. Namely,  for each $s$ 
choose  $\rec{(s)}$ as some fixed Pauli operator acting on $\calB$ such that
\be
\label{Corr1}
\syn_1(\rec{(s)}) =  \sigma \oplus \syn_1(M).
\ee
We choose the proxy~$M$ such that it  is
consistent with the $\calS_0$-syndrome caused by the actual error $E$
and has the smallest possible weight subject to this condition.
Specifically, choose $M$ as a minimum weight Pauli operator  acting on $\calC$
such that $\syn_0(M)=\syn_0(E)$.  Note that that $M$ acts non-trivially only on $\calA$.
Indeed, since any element of $\calS_0$ acts trivially on~$\calB$, the
$\calS_0$-syndrome of $M$ depends only on the restriction of~$M$ onto $\calA$.
Thus the weight of $M$ is minimal only if $M$ acts trivially on $\calB$.
Furthermore, 
since all elements of $\calS_0$ are $Z$-type Paulis, we can assume wlog
that $M$ is an $X$-type Pauli. Note that $M$ can be viewed either as a function of $s$ 
(since $\syn_0(E)$ can be inferred from $s$) or as a function of $E$.
This completes the construction of the proxy~$M$ for the actual error~$E$, and with~\eqref{Corr1}, the definition of the 
recovery~$\rec{(s)}$.

\begin{figure}[h]
\begin{tabular}{c|c} 
& syndrome of the stabilizers $B^1,\ldots,B^{2m}$  on the final state \\
\hline
ideal preparation & $\sigma(s)$ \\
ideal preparation + recovery & $0$ \\
\hline
noisy preparation  & $\sigma(s) \oplus  \syn_1(E)$  \hspace{1cm} (actual syndrome)\\
  &  $\sigma(s) \oplus \syn_1(M(s))$ \hspace{1cm} (guessed syndrome) \\
noisy preparation + recovery & $\syn_1(M(s))\oplus \syn_1(E) = \syn_1(\repair(E))$ \\
noisy preparation + recovery + repair & $0$ \\
\end{tabular}
\caption{Here we consider a complete set of stabilizers $B^1,\ldots,B^{2m}$ for the
desired logical state $\overline{\Phi}$ and their syndrome on the
output of the state preparation circuit before and after applying the recovery/repair
operators. The state preparation succeeds whenever the final syndrome is zero.
In the ideal case ($E=I$), the syndrome of $B^i$ can be inferred from the measument
outcome $s$.  
In the noisy case, the syndrome of $B^i$ can be guessed by
replacing the unknown error $E$ with a proxy $M=M(s)$. 
We choose the recovery operator $\rec(s)$ based on the guessed syndromes of $B^i$
such that $\syn_1(\rec{(s)}) =  \sigma(s) \oplus \syn_1(M(s))$.
The repair operator $\repair(E)$ compensates for the difference between
the guessed and the true syndromes of $B^i$
producing the desired logical state $\overline{\Phi}$.
\label{fig:state_prep}
}
\end{figure}

Finally, the repair operator compensates for the difference between $\syn_1(E)$ and $\syn_1(M)$. We choose $\repair{(E)}$ as a minimum weight Pauli operator acting on $\calB$
such that 
\be
\syn_1(\repair{(E)}) = \syn_1(E) \oplus \syn_1(M).\label{eq:repairconditionsyne}
\ee
(More precisely, we shall impose the  minimum weight condition independently
on the $X$- and $Z$-parts of  $\repair{(E)}$, see Section~\ref{sec:3Dcluster}.) Recall that $M$ can be viewed as a function of $E$, so that
$\repair{(E)}$ is well defined.
The above arguments show that 
\[
(|s\ra\la s|_\calA \otimes \repair{(E)}_\calB\cdot \rec{(s)}_\calB) EW|0\ra_\calC
\sim |s\ra_\calA \otimes |\overline{\Phi}\ra_\calB
\]
for all $s,E$. Multiplying both sides by $I_\calA\otimes \repair{(E)}_\calB$ gives Eq.~\eqref{RecRep}.

This completes the construction of the recovery and repair functions in the case where the error~$E$ acts trivially on~$\mathcal{B}$. In fact,  in the following, we will only consider this case, and show that for local stochastic
noise $E=E_{\mathcal{A}}\otimes I_{\mathcal{B}}\sim \mathcal{N}(p)$ we have
\begin{align}
\repair{(E_{\mathcal{A}}\otimes I_{\mathcal{B}})}\sim \mathcal{N}{(q)}\qquad\textrm{ for  }\qquad q=p^{\Omega(1)}\label{eq:intermedea}
\end{align}
In the general case one can write 
$E=E_\calA \otimes E_\calB$.
The error $E_\calB$ commutes with the measurements of $\calA$.
Thus neither the proxy $M$ nor the recovery function $\rec{(s})$ depend on $E_\calB$.
We define the repair function for a general error $E$ as
\[
\repair{(E)} = (I_\calA \otimes E_\calB) \repair{(E_\calA \otimes I_\calB)}.
\]
The same arguments as above show that such repair function obeys Eq.~\eqref{RecRep}.
By definition of local stochastic noise, 
$E\sim \calN(p)$ implies $E_\calA \sim \calN(p)$ and $E_\calB \sim \calN(p)$, see Lemma~\ref{claim:basic}. 
With~\eqref{eq:intermedea} and using part~\eqref{it:claimthree} of Claim~\ref{claim:basic}
we conclude that
\[
\repair{(E)}  \sim \calN(\tilde{q}), \qquad \tilde{q}=2\max{(p^{1/2},q^{1/2})}.
\]
In Section~\ref{sec:3Dcluster}, we show how to instantiate this construction
using a fault-tolerant scheme for 
preparing long-range entanglement 
in noisy 3D cluster states~\cite{raussendorf2005long}.
This scheme gives $q=26 p^{1/64}$ resulting
in $\tilde{q}\le 11p^{1/128}$. This is the bound
quoted in Theorem~\ref{thm:3D}.

\subsubsection{Single-shot logical state preparation from the lifting property}
We have described a general construction of repair and recovery functions starting from stabilizer groups $\mathcal{S}_0\subseteq\mathcal{S}_1\subseteq\mathcal{S}$ satisfying the properties~\eqref{it:firstpropprep}--\eqref{it:lastpropprep}. It remains to prove that $\repair{(E)}$ is a local stochastic error with the
noise rate $p^{\Omega(1)}$, that is,
\be
\label{outline_eq1}
\mathrm{Pr}_{E}[ K\subseteq \mathrm{Supp}(\repair{(E)})]\le p^{\Omega(|K|)}\qquad\textrm{ for any subset }K\subseteq \calB\ .
\ee
As discussed above, we can assume without loss of generality that
$E=E_{\mathcal{A}}\otimes I_\mathcal{B}\sim \mathcal{N}(p)$ is a local stochastic error acting on the ancilla qubits~$\mathcal{A}$ only. Accordingly, we can assume
that the repair function $\repair{(E)})$ takes as input a Pauli error $E$ acting on $\calA$
and outputs a Pauli error acting on $\calB$.

Next, we identify a certain property of the repair function which is sufficient to imply~\eqref{outline_eq1}. For convenience, and as this proof strategy may be applicable  to other codes, we introduce the following definition.
\begin{definition}\label{def:lifting}
A repair function $\repair(E)$ has the {\em lifting property}
if there exists a function $\lift(K)$ that takes as input
a subset $K\subseteq \mathcal{B}$
and outputs a set of subsets of $\mathcal{C}=\mathcal{A}\mathcal{B}$ 
such that 
\begin{enumerate}[(i)]
\item\label{prop:decay}
For all $\lambda>0$ and $K\subseteq \mathcal{B}$, we have 
\begin{align}
\label{outline_eq2}
\sum_{L\in \lift(K)} \lambda^{|L|} \le (c_1\lambda^{c_2})^{|K|}\ .
\end{align}
\item\label{prop:tricky}
For any Pauli 
$E$ acting on $\calA$ 
and any subset 
$K\subseteq \mathrm{Supp}(\repair{(E)})$
there exists a set $E'\in \lift(K)$  such that
\be
\label{outline_eq4}
|E'\cap \mathrm{Supp}(E)| \ge c_3|E'|
\ee
\end{enumerate}
Here $c_1,c_2,c_3>0$ are some universal constants.
\end{definition}
Here we allow the possibility $\lift(K)=\emptyset$ for some $K$'s. Let us agree 
that the sum in Eq.~\eqref{outline_eq2} is zero whenever $\lift(K)=\emptyset$. In Section~\ref{sec:3Dcluster} we provide an explicit construction of the function $\lift$ 
for the surface code case  
and compute the constants~$c_1,c_2,c_3$ (see Lemma~\ref{prop:Lift},\ref{lemma:tricky1}).
Here property~\eqref{prop:tricky} is particularly non-trivial to establish as it involves the repair function.

Given the lifting property, statement~\eqref{outline_eq1} can be shown as follows.
Consider some fixed subset $K\subseteq \calB$.
For each error $E$ with  $K\subseteq \mathrm{Supp}(\repair{(E)})$ fix a set $E'\in \lift(K)$ satisfying Eq.~\eqref{outline_eq4}, that is, $E'=E'(E)\subseteq \calC$
is a function of the error $E$.  The union bound gives  
\be
\label{outline_eq5}
\mathrm{Pr}_E[K \subseteq \mathrm{Supp}(\repair{(E)})] \le \sum_{L\in \lift(K)}\; \mathrm{Pr}_E[E'=L].
\ee
Let $F\equiv E'\cap \mathrm{Supp}(E)$.
From Eq.~\eqref{outline_eq4} one infers that $|F|\ge c_3 |E'|= c_3|L|$.
By the union bound,
\be
\label{outline_eq6}
\mathrm{Pr}_E[E'=L] \le \sum_{F\subseteq L\, : \, |F|\ge c_3|L|}\; \mathrm{Pr}_E[F\subseteq \mathrm{Supp}(E)]
\le  2^{|L|} p^{|F|} \le (2p^{c_3})^{|L|}.
\ee
Here we used the assumption $E\sim \calN(p)$
and noted that $L$ has at most $2^{|L|}$ subsets $F$.
Substituting Eq.~\eqref{outline_eq6} into Eq.~\eqref{outline_eq5}
and using property Eq.~\eqref{outline_eq2} of the $\lift$ function we arrive at
\be
\label{repair_error_rate}
\mathrm{Pr}_E[K \subseteq \mathrm{Supp}(\repair{(E)})] \le \sum_{L\in \lift(K)}
(2p^{c_3})^{|L|}
\le q^{|K|}, \qquad q\equiv c_12^{c_2} p^{c_2 c_3}.
\ee
This confirms Eq.~(\ref{outline_eq1}).

To summarize, we have shown that any repair function with the lifting property converts a local stochastic error~$E\sim \mathcal{N}(p)$ to a local stochastic error~$\repair{(E)}\sim \mathcal{N}(q)$ where $q=p^{\Omega(1)}$. In the case of the surface code considered in Section~\ref{sec:repairfunctiondef}, the repair function of interest 
will be  a product of four functions~$\repair_X,\repair_Z,\repair_{\bar{X}},\repair_{\bar{Z}}$, 
see Eq.~\eqref{repairXZ}. We will show that $\repair_X,\repair_Z$ satisfy the lifting property
and thus preserve the local stochasticity property,
that is, $\repair_X$ and $\repair_Z$ have noise rate $q=p^{\Omega(1)}$.
Furthermore,  $\repair_{\bar{X}},\repair_{\bar{Z}}$ will be random 
logical Pauli errors that apply the logical operators $\overline{X}$, $\overline{Z}$ 
to one of the surface codes with probability $p^{\Omega(d)}$.
We shall see that such logical errors automatically obey the local stochasticity condition
with the noise rate $p^{\Omega(1)}$.
Thus the full repair function $\repair{(E)}$ is a product of four 
local stochastic  errors with the noise rate $p^{\Omega(1)}$.
By Lemma~\ref{claim:basic}, we conclude that $\repair{(E)}$
is a local stochastic error with the noise rate $p^{\Omega(1)}$.

\section{Single-shot Bell state preparation from a 3D lattice of qubits}
\label{sec:3Dcluster}
Here we provide all missing steps in the proof 
of Theorem~\ref{thm:3D} outlined in Section~\ref{sec:codes}. 
In Section~\ref{sec:latticeconstruction}, we specify the geometric arrangement of qubits and define the relevant stabilizer groups.  Together with the recovery function already introduced in Section~\ref{sec:singleshotlogicalstateprep}, this determines a scheme for preparing encoded Bell states. We note that this scheme is essentially that introduced in~\cite{raussendorf2005long}. In the latter paper, the authors show that performing single-qubit measurement in the bulk of a 3D~cluster state leads to an encoded Bell state of two surface codes on two boundaries, up to an error determined by the measurement outcomes. Remarkably, this was shown to be robust with respect to errors on the bulk qubits, demonstrating that a 3D~cluster state has noise-resilient long-range localizable entanglement.

The remainder of this section is devoted to establishing the fault-tolerance property of this scheme.  This analysis goes beyond~\cite{raussendorf2005long} by not requiring noise-free operations on the boundaries. In Section~\ref{sec:repairfunctiondef}, we define the repair function used in our analysis to express the residual error. In Section~\ref{sec:liftingproperty}  we show that this repair function 
satisfies the lifting property. Finally, in Section~\ref{sec:parameterscomplete}, we combine this with the arguments from Section~\ref{sec:singleshotlogicalstateprep} to complete the proof of Theorem~\ref{thm:3D}.

\subsection{3D lattice code construction\label{sec:latticeconstruction}}
We begin by defining the 3D lattice $\calC$, a constant-depth Clifford circuit $W$
acting on $\calC$,
and stabilizer groups~$\calS_0$,~$\calS_1$ and $\calS$.
Let $d$ be the surface code distance, $r=2d-1$, and 
\[
\calC'=\{  (u_1,u_2,u_3) \in \mathbb{Z}^3 \, : \, 1\le u_1,u_3\le r, \quad 0\le u_2\le r-1\}.
\]
We shall refer to triples of integers $u=(u_1,u_2,u_3)\in \calC'$ as {\em sites}.
Let $e$ and $o$ denote  arbitrary even and odd integers.
Qubits are placed at sites of the sublattice
\be
\label{latticeC}
\calC = \calC'\setminus \{ (o,o,o),(e,e,e)\}.
\ee
In other words, $\calC$ contains all sites of $\calC'$ that have at least
one odd and at least one even coordinate.
The region $\calB$ encoding the Bell state  is defined as
\[
\calB = \{ (e,o,1),(o,e,1),(e,o,r),(o,e,r)\in \calC\}.
\]
In other words, $\calB$ contains all qubits $u$ located on the faces
of $\calC$ with $u_3\in \{1,r\}$ such that $u_1$, $u_2$ have different
parity. Ancillary qubits live at sites of the region $\calA\equiv \calC\setminus \calB$.
Let $n=|\calC|$ be the total number of qubits.
Given a site $u\in \calC'$ define a set of nearest neighbors of $u$ as
\[
\neigh(u)=\{ v\in \calC \, : \, |u_1-v_1|+|u_2-v_2|+|u_3-v_3|=1\}.
\]
Note that each site $u\in \calC$ has at most four nearest neighbors.
For example, a site $(2,2,1)$ has nearest neighbors
$(1,2,1)$, $(3,2,1)$, $(2,1,1)$, and $(2,3,1)$, see Eq.~\eqref{latticeC}.
Define
\be
\label{HCZH}
W=H^{\otimes n}  \prod_{(u,v)\in \calC}\; CZ_{u,v} \; H^{\otimes n},
\ee
where the product runs over all pairs of nearest neighbor sites in $\calC$.
One can easily verify that the product of CZ gates  in Eq.~\eqref{HCZH} can be implemented by a
depth-four circuit. Thus the full circuit $W$ has depth six, as promised in  Theorem~\ref{thm:3D}.
The state $W|0^n\ra$ is a stabilizer state with stabilizer generators
\be
\label{stabilizerSu}
G_u = Z_u \prod_{v\in \neigh(u)}\; X_v, \qquad u\in \calC.
\ee
 Let $\calS$ be the group generated by $\{ G_u\}_{u\in\calC}$.

 Let us briefly comment on how this relates to the construction of~\cite{raussendorf2005long}. The starting point of~\cite{raussendorf2005long} is a cluster (or graph) state  $\ket{\Psi_{G}}=\prod_{(u,v)\in E}\; CZ_{u,v} \; H^{\otimes n}\ket{0^n}$ associated with a particular 3D graph~$G=(V,E)$. A certain measurement pattern consisting of single-qubit $X$- and $Z$-measurements is then applied to the state~$\ket{\Psi_{G}}$, resulting in the desired target state on a subset of qubits on the boundaries. We note that measuring a qubit in the $Z$-basis in a graph state amounts to
 removing the corresponding vertex from the graph: in other words, such qubits may be removed from the beginning. This is what we do in our description here. Furthermore, we introduce another layer of Hadamards on each qubit for convenience (see~\eqref{HCZH}), meaning that the remaining qubits will be measured in the $Z$-basis rather than the $X$-basis as opposed to the description in~\cite{raussendorf2005long}. Below we mostly follow notations introduced in~\cite{raussendorf2005long}.

To  set the stage for what follows
it is convenient to describe $\calC'$ as a union of four graphs denoted
$T_e$ (even graph), $T_o$ (odd graph),  $T_{sc}$ (surface code graph),
and $T_{sc}^*$ (dual surface code graph). These graphs  have sets of vertices 
\begin{align}
\calV(T_e) =& \{ u=(e,e,e) \in \calC'\},\\
\calV(T_o) =& \{ u=(o,o,o) \in \calC' \, : \, u_3 \ne 1,r\},\\
\calV(T_{sc}) =&  \{ u=(e,e,o) \in \calC' \, : \, u_3 =1,r\},\\
\calV(T_{sc}^*) =&  \{ u=(o,o,o) \in \calC' \, : \, u_3 =1,r\}
\end{align}
and sets of edges
\begin{align}
\calE(T_e) =& \{u=(e,e,o), \; (e,o,e), \; (o,e,e) \in \calC'\},\\
\calE(T_o) = & \{u=(o,o,e), \; (o,e,o), \; (e,o,o) \in \calC' \, : \, u_3 \ne 1,r\},\\
\calE(T_{sc}) = \calE(T_{sc}^*) =& \{u=(o,e,o), \; (e,o,o) \in \calC' \, : \, u_3 =1,r\}.
\end{align}
Here a pair of vertices $u,v$ is connected by an edge
if one can obtain $v$ from $u$ by changing a single coordinate by $\pm 2$.
We show examples of the graph $T_e$ and $T_{sc}$ on Fig.~\ref{fig:Te}.
Note that each graph has dangling edges (i.e. edges with only one endpoint)
located at certain external faces of $\calC'$.
By definition,  
\[
\calA=\calE(T_e) \cup \calE(T_o) \quad \mbox{and} \quad
 \calB = \calE(T_{sc})= \calE(T_{sc}^*).
\]
Vertices of the graphs $T_e,T_o$ and $T_{sc},T_{sc}^*$ will be associated
with generators of the subgroups $\calS_0\subseteq \calS$
and $\calS_1\subseteq \calS$ respectively,  see Section~\ref{sec:codes}.

Let us first define the subgroup $\calS_0$.
Recall that elements of $\calS_0$ must act trivially on $\calB$ and may
act on $\calA$ only by Pauli $Z$.  The group $\calS_0$ has generators

\be
\label{generatorS0u}
S_0^u=
\prod_{v\in \neigh(u)} G_v
=\prod_{v\in \neigh(u)} Z_v,
  \qquad u \in \calV(T_e) \cup \calV(T_o).
\ee
Here we noted that all $X$-type Pauli in the product cancel each other.
Also note that $\neigh(u)\cap \calB=\emptyset$ for any site
$u \in \calV(T_e) \cup \calV(T_o)$. Thus $S_0^u$ acts trivially on $\calB$,
as desired.

\begin{figure}[h]
\centerline{\includegraphics[height=6cm]{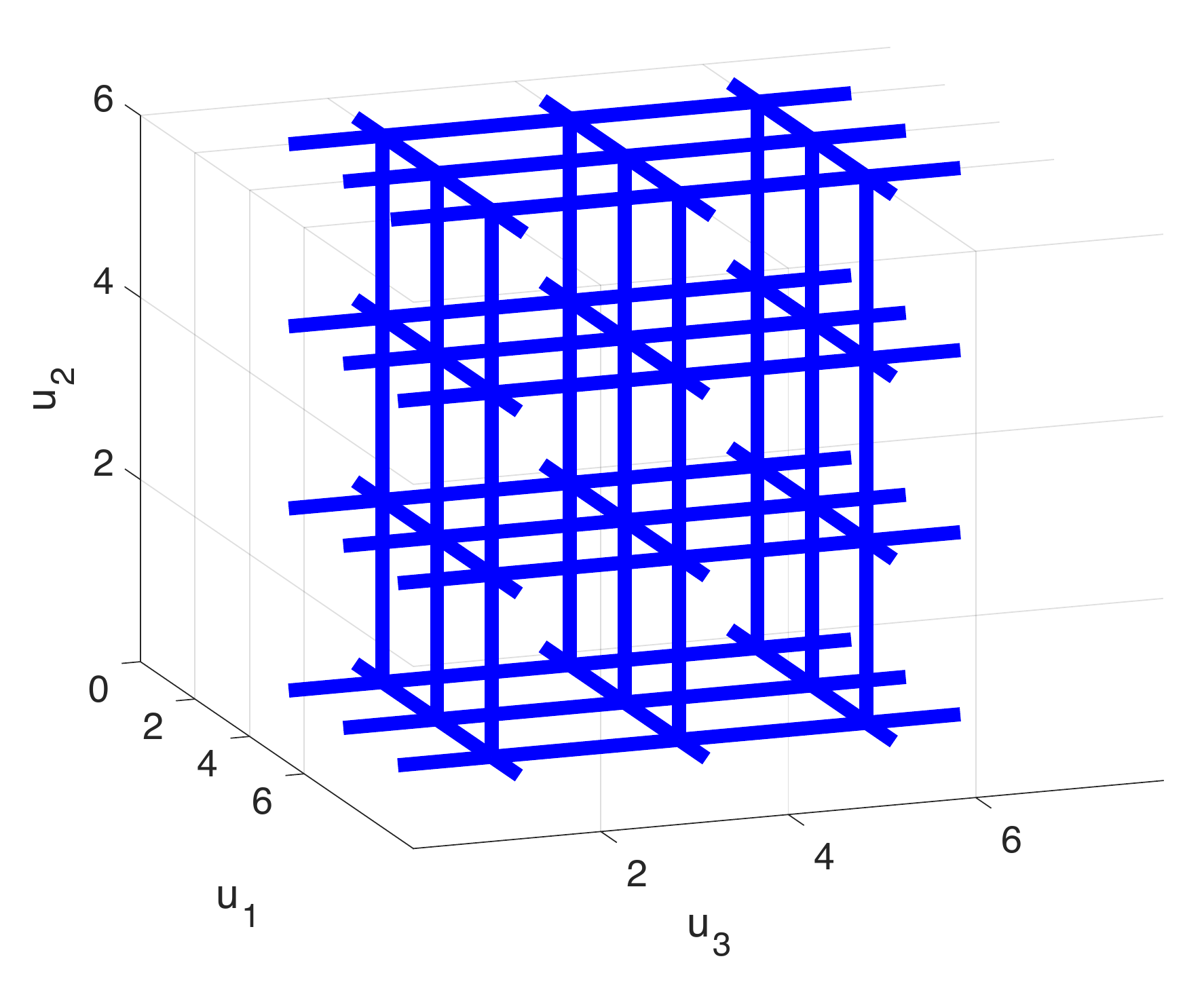}
\hspace*{5mm}\includegraphics[height=6cm]{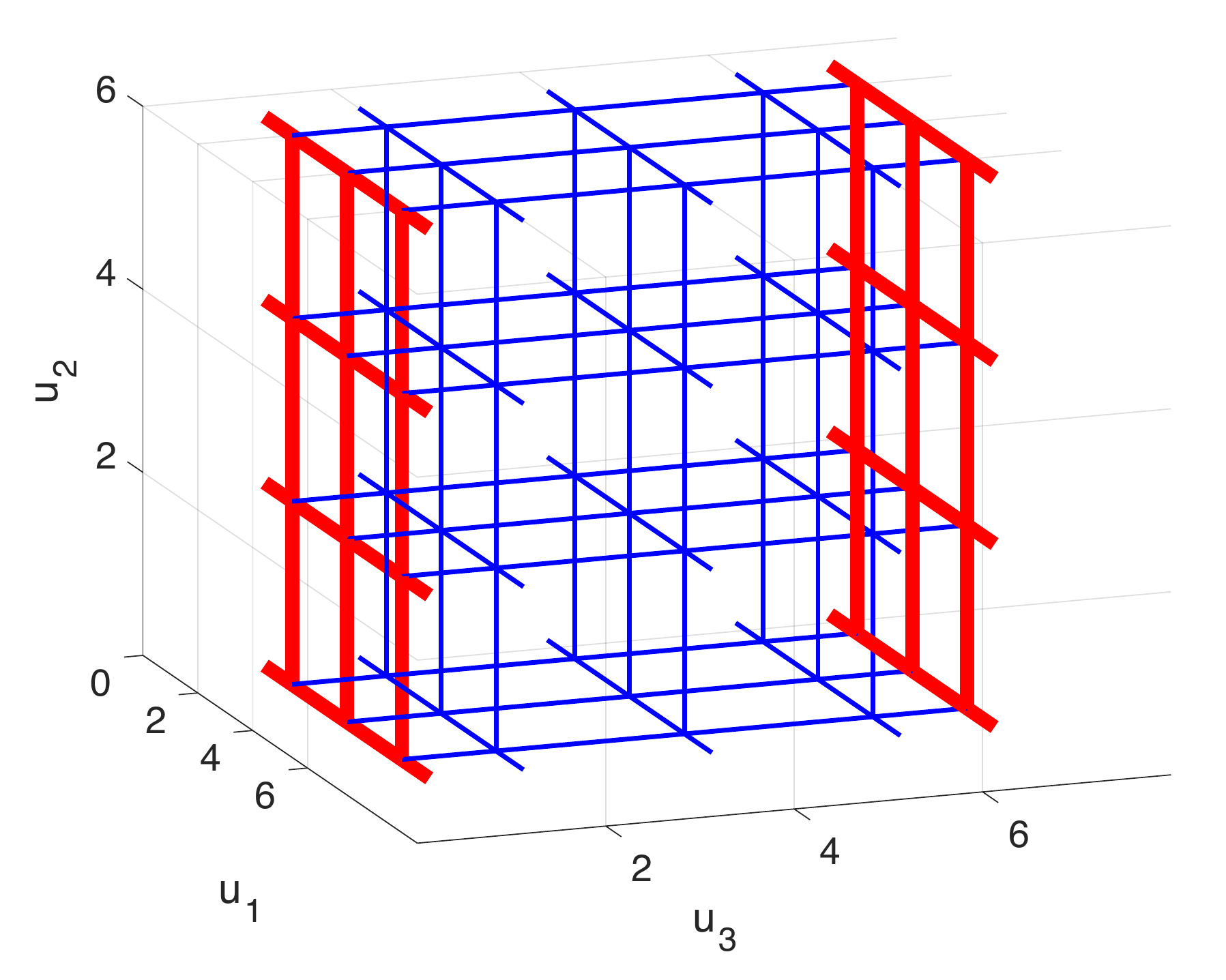}}
\caption{Examples of the graphs $T_e$ (blue) and $T_{sc}$ (red) for
the distance-$4$ surface code ($r=7$).
The glued graph $\gluedgraph$ is constructed
from  $T_e\cup T_{sc}$
 by attaching left and right
dangling edges of $T_e$ to the corresponding vertices of $T_{sc}$.
\label{fig:Te}
}
\end{figure}

We proceed to defining the subgroup $\calS_1$.
Recall that elements of $\calS_1$ may act on $\calA$ only by Pauli~$Z$.
In addition, the restriction of any element of $\calS_1$ onto $\calB$
must be a stabilizer of $\overline{\Phi}$.
Generators of $\calS_1$ come in several types. First, 
choose any site $u \in \calV(T_{sc})$ and
define
\be
\label{generatorS1u'}
S_1^u=G_u =
Z_u\prod_{v\in \neigh(u)} X_v, 
\qquad u \in \calV(T_{sc}).
\ee
We claim that $\neigh(u)\subseteq \calB$ for any site $u \in \calV(T_{sc})$.
Indeed, such site has the form $u=(e,e,1)$ or $(e,e,r)$.
Changing the third coordinate of $u$ by $\pm 1$ gives
a site $(e,e,e)$ which is not contained in $\calC$, see Eq.~\eqref{latticeC}.
Changing the first or the second coordinate of $u$ by $\pm 1$
gives a site in $\calB$. Thus $\neigh(u)\subseteq \calB$, that is, 
all Pauli $X$ in the generator $S_1^u$ act on $\calB$.
The restriction of $S_1^u$ onto $\calB$ becomes
a vertex stabilizer for one of the two surface codes,
see Section~\ref{sec:codes}.

Next, choose any site $u\in \calV(T_{sc}^*)$ and define
\be
\label{generatorS1u''}
S_1^u =G_{u\pm (0,0,1)}\prod_{v\in \neigh(u)\cap \calB} G_v
=Z_{u\pm (0,0,1)} \prod_{v\in \neigh(u)\cap \calB} Z_v,
\qquad u\in \calV(T_{sc}^*).
\ee
By definition, any site $u \in \calV(T_{sc}^*)$ has the form $u=(o,o,u_3)$
with $u_3\in \{1,r\}$. The sign in Eq.~\eqref{generatorS1u''}
is plus for $u_3=1$ and minus for $u_3=r$.
Such site $u$ has exactly one nearest neighbor not in $\calB$, namely
$u\pm (0,0,1)$. One can easily check that all Pauli $X$
in Eq.~\eqref{generatorS1u''} cancel each other.
The restriction of $S_1^u$ onto $\calB$ becomes
a face stabilizer of the surface code, see Section~\ref{sec:codes}.

Finally, the group $\calS_1$ has two special generators
$S_1^X$ and $S_1^Z$ 
corresponding to the logical Bell state stabilizers
$\overline{X}_1 \, \overline{X}_2$ and
$\overline{Z}_1 \, \overline{Z}_2$. 
The latter can be chosen as 
\[
\overline{X}_1 \, \overline{X}_2 = \prod_{u=(1,e,1)\in \calB} X_u  \prod_{u=(1,e,r)\in \calB} X_u
\]
(rough boundaries of the surface code lattice) and
\[
\overline{Z}_1 \, \overline{Z}_2 = \prod_{u=(o,0,1)\in \calB} Z_u  \prod_{u=(o,0,r)\in\calB} Z_u
\]
(smooth boundaries of the surface code lattice). 
We set
\be
\label{S1X}
S_1^X = \prod_{u=(1,e,e)\in \calE(T_e)}\; G_u 
=(\overline{X}_1\overline{X}_2)_\calB \prod_{u=(1,e,e)\in \calE(T_e)}\; Z_u.
\ee
Here the product runs  over all dangling edges of the graph
$T_e$ that cross the face $u_1=1$. 
One can easily check that all Pauli $X$ in this
product cancel each other, except for those 
that appear in the logical operators $\overline{X}_1$
and $\overline{X}_2$. Finally, set
\be
\label{S1Z}
S_1^Z =  \prod_{u=(o,0,o)\in\calC}\; G_u
=(\overline{Z}_1\overline{Z}_2)_\calB \prod_{u=(o,0,o)\in \calE(T_o)}\; Z_u.
\ee
Note that  $\{ (o,0,o)\in\calC\} \subseteq \calE(T_o)\cup \calE(T_{sc})$
where the union is disjoint. 
Furthermore, 
$\calE(T_o)\subset \calA$ and $\calE(T_{sc})=\calB$.
One can  check that all Pauli $X$ in the above product cancel each other
and the action of $S_1^Z$ onto $\calB$ gives $\overline{Z}_1\overline{Z}_2$.
This completes the construction of the subgroup $\calS_1$.

\subsection{Definition of the $\repair{}$-function\label{sec:repairfunctiondef}}
Suppose $E$ is a Pauli error acting
on $\calA$. By definition, the syndrome $\syn_0(E)$
can be viewed as a subset of $\calV(T_e)\cup \calV(T_o)$ such that
$u\in \syn_0(E)$ iff the generator $S_0^u$ anti-commutes with $E$.
As discussed in Section~\ref{sec:codes},
the construction of recovery and repair functions
is based on a minimum weight $X$-type Pauli operator $M$
supported on $\calA$ such that $\syn_0(M)=\syn_0(E)$.
Using Eq.~\eqref{generatorS0u} one can check that 
an error $X_u$ on any edge $u\in \calE(T_e)$
creates $\calS_0$-syndrome at the endpoints of $u$ in the graph
$T_e$.
Likewise, an error $X_u$ on any edge $u\in \calE(T_o)$
creates $\calS_0$-syndrome at the endpoints of $u$ in the graph
$T_o$. Recall that the graphs $T_e$ and $T_o$ have dangling edges.
Errors $X_u$ on such edges create only one  bit of $\calS_0$-syndrome
at the single endpoint of the edge.

To describe the construction of $M$ more formally we need some terminology. 
Suppose $\calG=(\calV,\calE)$ is a simple graph that may have dangling edges.
Given a subset of edges 
$F\subseteq \calE$, the {\em boundary} of $F$ denoted $\partial F$ is 
the subset of vertices  $u\in \calV$ 
such that $u$ has an odd number of incident edges from $F$.
We say that $F$ is a {\em cycle} if $\partial F=\emptyset$.
We say that $F$ is a {\em minimum matching} if 
for any $F'\subseteq \calE$ such that 
$\partial F'=\partial F$ one has $|F'|\ge |F|$.
(Note that our definition of a  minimum matching is slightly different from the one
commonly used in graph theory.)

Since $M$ includes only $X$-type errors, we shall 
identify $M$ and its support.
The above shows that $M\cap \calE(T_e)$ is a minimum matching
in the graph $T_e$ with the boundary $\syn_0(E)\cap \calV(T_e)$.
Likewise, $M\cap \calE(T_o)$ is a minimum matching
in the graph $T_o$ with the boundary $\syn_0(E)\cap \calV(T_o)$.
This completes the construction of $M$.

At this point we have well-defined $\calS_1$-syndromes $\syn_1(E)$, $\syn_1(M)$
and we are ready to construct the repair function  $\repair{(E)}$.
Recall that $\repair{(E)}$ is defined as a minimum weight Pauli operator acting on $\calB$
such that 
\be
\label{repair_full}
\syn_1{(\repair(E))} = \syn_1{(E)}\oplus \syn_1{(M)}.
\ee
Define partial $\calS_1$-syndromes associated
with $X$- and $Z$-type  surface code stabilizers.
Given a Pauli error $F$ acting on $\calC$,
let  $\syn_{1X}(F)$ be the combined syndrome of generators $S_1^u$ 
located at vertices of the surface code graph, $u\in \calV(T_{sc})$,
see Eq.~\eqref{generatorS1u'}. Likewise, let $\syn_{1Z}(F)$ be the combined syndrome of generators $S_1^u$ 
located at vertices of the dual surface code graph, $u\in \calV(T_{sc}^*)$,
see Eq.~\eqref{generatorS1u''}. The repair function is defined as
\be
\label{repairXZ}
\repair{(E)}=\repair_X{(E)} \cdot  \repair_{\overline{X}}{(E)} \cdot \repair_Z{(E)} 
\cdot  \repair_{\overline{Z}}{(E)},
\ee
where $\repair_X{(E)}$ is a minimum weight $X$-type Pauli operator
acting on $\calB$ such that 
\be
\label{repairX}
\syn_{1Z}(\repair_X{(E)}) = \syn_{1Z}(E) \oplus \syn_{1Z}(M),
\ee
$\repair_Z{(E)}$ is a minimum weight $Z$-type Pauli operator
acting on $\calB$ such that 
\be
\label{repairZ}
\syn_{1X}(\repair_Z{(E)}) = \syn_{1X}(E) \oplus \syn_{1X}(M).
\ee
Equivalently,  the support of $\repair_Z(E)$ is a minimum matching in the graph
$T_{sc}$ with the boundary $\syn_{1X}(E) \oplus \syn_{1X}(M)$.
The support of $\repair_X{(E)}$ is a minimum matching in the graph
$T_{sc}^*$ with the boundary $\syn_{1Z}(E) \oplus \syn_{1Z}(M)$.

The operators $\repair_{\overline{X}}{(E)}$
and $\repair_{\overline{Z}}{(E)}$ 
can be viewed as residual logical errors. 
They ensure that  $\repair(E)$ and $E\cdot M$ have the same syndromes for
the generators $S_1^X$ and $S_1^Z$ associated with the stabilizers
$\overline{X}\, \overline{X}$ and $\overline{Z}\, \overline{Z}$ of the logical Bell state. 
This is discussed in details in Section~\ref{sec:parameterscomplete}.

We shall see that $\repair_{\overline{Z}}{(E)}=I$
with probability exponentially close to one, and that  $\repair_{\overline{Z}}(E)$ is a local stochastic error with rate~$p^{\Omega(1)}$. We will also show that~$\repair_Z$ satisfies the lifting property such that $\repair_Z{(E)}\sim \mathcal{N}(p^{\Omega(1)})$. Exactly the same arguments apply to the 
repair functions $\repair_{\overline{X}}{(E)}$
and $\repair_X{(E)}$   if one replaces the graphs $T_e$ and $T_{sc}$
with $T_o$ and $T_{sc}^*$ respectively. 
Part~(iii) of Claim~\ref{claim:basic} then implies that
the full repair operator
$\repair{(E)}$ is a local stochastic error with rate $p^{\Omega(1)}$.

\subsection{Proof of the lifting property\label{sec:liftingproperty}}
Here we argue that $\repair_Z$ satisfies the lifting property (see Definition~\ref{def:lifting}).  To define  the function $\lift$ we shall need some basic facts from graph
theory.  
Suppose $\calG=(\calV,\calE)$ is a simple graph that may have dangling edges.
A subset of edges $F\subseteq \calE$ is called a
{\em forest} iff it contains no cycles. In other words, $F$ is an edge-disjoint union of trees.
Given a subset of vertices $S\subseteq \calV$, let $\calF(S;\calG)$ be 
the set of all forests $F$  in the graph $\calG$ such that $\partial F= S$.
We shall use the simpler notation $\calF(S)$ whenever the graph $\calG$
is clear from the context.
Note that $\calF(\emptyset,\calG)=\emptyset$.
\begin{lemma}
\label{prop:graph1}
Any forest $F\in \calF(S)$ can be partitioned (non-uniquely) 
into edge-disjoint  paths, $F=F_1\cdots F_k$, such that 
each path $F_i$ has endpoints in $S$ and each vertex in $S$ is 
an endpoint of exactly one path $F_i$.
Some paths $F_i$
may have only one endpoint (if the graph has dangling edges).
\end{lemma}
\begin{proof}
We shall use induction in $|S|$. If $S=\emptyset$ then $\calF(S)=\emptyset$
and there is nothing to prove. Suppose $|S|\ge 1$
and $F\in \calF(S)$. 
Choose any vertex $u\in S$ and let $F'$ be the connected component
of $F$ that contains $u$. Note that $F'\ne \emptyset$
since $u\in \partial F$ implies that $F$ contains at least one edge
incident to $u$. We can consider $F'$ as a tree rooted at $u$.
Let $e\in \calE$ be any leaf of $F'$ and $F_1$ be the unique path
in $F'$ connecting $u$ with $e$. If $e$ is a dangling edge then 
$\partial F_1=\{u\}$ and thus $\partial (F\setminus F_1)=S\setminus \{u\}$.
Otherwise, let $v$ be the endpoint of $e$ such that the path $F_1$ terminates at $v$.
 Note that $v\in S$ since $e$ is the only edge of $F$ incident to $v$.
Thus $\partial (F\setminus F_1) = S\setminus \{u,v\}$.
In both cases one can apply the induction hypothesis to the forest
$F\setminus F_1$.
\end{proof}
\begin{lemma}
\label{prop:graph2}
Any minimum matching is a forest.
\end{lemma}
\begin{proof}
Indeed, if $F$ is a minimum matching and $F$ contains a cycle $C$
then $|F\oplus C|=|F|-|C|<|F|$ and $\partial (F\oplus C)=\partial F \oplus \partial C = \partial F$.
This contradicts the minimality of $F$.
\end{proof}
\begin{lemma}
\label{prop:graph3}
Suppose $F$ is a minimum matching and $K\subseteq F$.
Then $K$ is a minimum matching.
\end{lemma}
\begin{proof}
Assume the contrary, that is, there exists
a subset of edges $K'$ such that
$\partial K'=\partial K$ and $|K'|<|K|$. Define
$F'=F\oplus K\oplus K'$.
We have $\partial F'=\partial F$ and 
\[
|F'|\le |F\oplus K| + |K'| = |F|-|K| + |K'| <|F|.
\]
This contradicts the minimality of $F$. 
\end{proof}
\begin{lemma}
\label{prop:graph4}
Suppose $\calG$ is a graph with vertex degree at most $D$.
Let $M$ be a minimum matching in $\calG$.
Choose any $q\ge 0$ such that $16 D q^{1/2}\le 1$. Then  
\be
\label{SAWsum}
\sum_{F \in \calF(\partial M;\calG)} q^{|F|} \le (16Dq^{1/2})^{|M|}
\ee
\end{lemma}
\begin{proof}
Given a vertex $u$ let $\calP(u)$ 
be  the set of all paths in the graph $\calG$ that have $u$ as an endpoint.
Suppose $\partial M = \{ u_1,\ldots,u_k\}$.
Choose any forest $F\in  \calF(\partial M)$ and let 
$F=F_1\cdots F_k$ be the decomposition of $F$ into edge-disjoint
paths established in Lemma~\ref{prop:graph1}.
Define a $k$-tuple of paths $H_j\in \calP(u_j)$, $j=1,\ldots,k$ as follows.
Each path $F_i$ with a single endpoint $u_j\in \partial M$ gives rise
to a path $H_j=F_i$. Each path $F_i$ with two endpoints $u_p,u_q\in \partial M$
gives rise to a pair of paths $H_p=H_q=F_i$.
By construction,  
\[
|M|\le  |F|\le  r\equiv \sum_{i=1}^k |H_i| \le 2|F|.
\]
The above shows that each
forest $F\in \calF(\partial M)$ can be mapped
to a $k$-tuple of paths $H_j \in \calP(u_j)$ such that
$r\ge |M|$ and $q^{|F|} \le q^{r/2}$. 
Thus
\be
\label{extendedSAW}
\sum_{F \in \calF(\partial M)} q^{|F|}
\le
\sum_{r=|M|}^{\infty}\; q^{r/2}
\sum_{H_1\in \calP(u_1)} \ldots \sum_{H_k \in \calP(u_k)}
\delta(r,|H_1|+\ldots +|H_k|)
\ee
Here $\delta(x,y)=1$ if $x=y$ and $\delta(x,y)=0$ otherwise. 
The number of  length-$s$ paths starting at 
a given vertex  is at most $(D-1)^s \le D^s$.
Thus
\be
\sum_{F \in \calF(\partial M)} q^{|F|}
\le
\sum_{r=|M|}^\infty \Gamma(r,k)\cdot (Dq^{1/2})^r 
\ee
where $\Gamma(r,k)$ is the number of ways to write $r$ as a 
sum of $k$ non-negative integers (order matters). 
Noting that  $|\partial M|\le 2|M|$ one gets 
$k=|\partial M|\le 2|M| \le 2r$ and thus
\be
\Gamma(r,k)={r+k-1 \choose k-1} \le 2^{r+k-1}\le 8^r.
\ee
We get
\be
\sum_{F \in \calF(\partial M)} q^{|F|}
\le
\sum_{r=|M|}^\infty (8Dq^{1/2})^r \le 2  (8Dq^{1/2})^{|M|}
\ee
since, by assumption, $8Dq^{1/2}\le 1/2$. We can assume wlog that $M\ne \emptyset$ (otherwise
both sides of Eq.~\eqref{SAWsum} equal to one).
Then $2\le 2^{|M|}$ which proves the lemma.
\end{proof}

Consider a subset $K\subseteq \calB$. We shall define
$\lift(K)$ as the set of forests with the boundary $\partial K$
in a  suitable graph. Namely, let $\gluedgraph$ be the graph obtained by gluing together 
the graphs
$T_e$ and  $T_{sc}$ such that 
the dangling edges $(e,e,1), (e,e,r) \in \calE(T_e)$ are attached to the respective
vertices of $\calV(T_{sc})$, see Fig.~\ref{fig:Te}.
The graph $\gluedgraph$ has the set of vertices $\calV(T_{sc}) \cup \calV(T_e)$
and the set of edges $\calE(T_{sc})\cup \calE(T_e)$.
Note that $\calB=\calE(T_{sc})$ becomes a subset of edges in the glued graph
$\gluedgraph$. Thus we define
\be
\label{LiftDef}
\lift(K) = \calF(\partial K; \gluedgraph)
\ee
as the set of forests in~$\gluedgraph$ with boundary~$\partial K$. 
Below we shall use the following property.
\begin{lemma}
\label{prop:graph5}
Suppose $M$ is a minimum matching
in the surface code graph $T_{sc}$. 
Then $M$ is also a minimum matching in the glued graph $\gluedgraph$.
\end{lemma}
\begin{proof}
Let $H\subseteq \calE(\gluedgraph)$ be a minimum matching in the graph $\gluedgraph$
such that $\partial H=\partial M$. Here the boundary is taken in the graph $\gluedgraph$.
 It suffices  to check that $|H|\ge |M|$.
Lemma~\ref{prop:graph2} implies that $H$ is a forest.
Let $H=H_1\cdots H_k$ be the partition of $H$ into
edge-disjoint paths established in Lemma~\ref{prop:graph1}.
We claim that for any path
$H_i\subseteq \calE(\gluedgraph)$ that has both endpoints in $\calV(T_{sc})$
there exists a subset $M_i\subseteq \calE(T_{sc})$ such that
$\partial H_i = \partial M_i$ and $|H_i|\ge |M_i|$.
Indeed, let $u,v\in \calV(T_{sc})$ be the endpoints of $H_i$.

Suppose first that $u$ and $v$ belong to the same connected component of $T_{sc}$
(i.e. both $u$ and $v$ belong to the same copy of the surface code).
Then the desired path $M_i$ can be chosen as a shortest path in the graph $T_{sc}$
connecting $u$ and $v$. We have $|H_i|\ge |M_i|$ since $M_i$ is also a shortest
path in the graph $\gluedgraph$ connecting $u$ and $v$.

Suppose now that $u$ and $v$ belong to different connected components of $T_{sc}$
(i.e. the path $H_i$ connects the two surface codes).
Then the length of $H_i$ must be at least $d$.
Let $M_i'$ and $M_i''$ be shortest paths connecting $u$ and $v$
to the nearest  rough boundary in the respective connected components of $T_{sc}$.
Since any vertex of $T_{sc}$ is within distance $d/2$ from some rough boundary,
we have $|M_i'|\le d/2$ and $|M_i''|\le d/2$. Choose $M_i=M_i'M_i''$.
Then $\partial H_i = \partial M_i$
and $|M_i|\le d\le |H_i|$, as claimed.
(We note that this is the 
only step in the proof of Theorem~\ref{thm:3D} that requires the separation between
the two surface codes to be sufficiently large.)

Likewise, for any path $H_i\subseteq \calE(\gluedgraph)$
that starts at some vertex of $\calV(T_{sc})$ and terminates
at a dangling edge of $\gluedgraph$
there exists a path $M_i\subseteq \calE(T_{sc})$ such that
$\partial H_i = \partial M_i$ and $|H_i|\ge |M_i|$.
Let $M'=M_1\oplus \cdots \oplus M_k$. By construction,
$\partial M'=\partial M =\partial H$ and $|H|\ge |M'|$.
The minimality of $M$ implies $|M'|\ge |M|$.
Thus $|H|\ge |M|$.
\end{proof}
Now we are ready to verify that the function~$\repair_Z$ 
satisfies -- together with the function~$\lift$ -- the property stated in  Eq.~(\ref{outline_eq2}), that is, Property~\eqref{prop:decay} of Definition~\ref{def:lifting}. 
\begin{lemma}
\label{prop:Lift}
Suppose $K\subseteq \repair_Z{(E)}$ for some Pauli error $E$ acting on $\calA$.
Then 
\be
\label{LiftTe}
\sum_{L\in \lift(K)}\; \lambda^{|L|} \le (96\lambda^{1/2})^{|K|}.
\ee
\end{lemma}
\begin{proof}
By definition of the repair function, $\repair_Z{(E)}$
is a minimum matching in the surface code graph $T_{sc}$.
Using Lemma~\ref{prop:graph5} one infers that~$\repair_Z{(E)}$
is a minimum matching in the glued graph~$\gluedgraph$.
By Lemma~\ref{prop:graph3}, $K$ is also a minimum matching
in $\gluedgraph$. The maximum vertex degree of $\gluedgraph$ is $D=6$
since $\gluedgraph$ is isomorphic to the 3D cubic lattice.
Now Eq.~\eqref{LiftTe} follows from Eq.~\eqref{LiftDef} and Lemma~\ref{prop:graph4}.
\end{proof}
To establish that $\repair_Z$  satisfies the lifting property, it remains to map an error $E$ to a subset $E'\in \lift(K)$,
as stated in property~\eqref{prop:tricky} of Definition~\ref{def:lifting}.  To this end we need the following.
\begin{lemma}
\label{prop:graph6}
Suppose $F\subseteq \calE(T_{sc})$ is a minimum matching
in the graph $T_{sc}$. Suppose $Y\subseteq \calE(T_e)$
is an arbitrary subset such that $\partial Y = \partial F$, 
where the boundary is taken in the graph $\gluedgraph$.
For any subset $K\subseteq F$ there exists a forest
$L \in \lift(K)$ such that 
\be
\label{LY1}
L \cap \calE(T_e) \subseteq Y
\ee
and
\be
\label{LY2}
|L| \le 2|L\cap \calE(T_e)|.
\ee
\end{lemma}
\begin{proof}
Use induction in the size of $Y$.
The base of induction is $Y=\emptyset$. We have
$\partial F=\partial Y = \emptyset$. Since $F$ is a minimum matching,
$F=\emptyset$. Thus $K=\emptyset$ and one
can choose $L=\emptyset$.

Suppose now that $Y$ is non-empty. 
Let $C=F\cup Y$. Since the graphs $T_{sc}$ and $T_e$ have no common edges,
one has $F\cap Y=\emptyset$ and thus 
 $C=F\oplus Y$. Note that $C$ is a cycle in the graph $\gluedgraph$.
Let $O\subseteq C$ be an arbitrary closed loop or a path starting and ending
at a dangling edge. Set
\[
F'=F\setminus O \quad \mbox{and} \quad Y'=Y\setminus O.
\]
Lemma~\ref{prop:graph3} implies that $F'$ is a minimum matching
in the graph $T_{sc}$.
We claim that $\partial F' = \partial Y'$. Indeed,
$F'\oplus Y' = F \oplus Y\oplus O=C\oplus O$ is a sum of two cycles.
Thus $F'\oplus Y'$ is a cycle, that is, $\partial F' = \partial Y'$. The loop $O$ must use at least one edge
of $Y$ since $F$ contains no cycles, see Lemma~\ref{prop:graph2}.
Thus $Y'$ contains at most $|Y|-1$ edges.
Consider two cases.

\noindent
{\em Case~1:} $O\cap K=\emptyset$.
 Then $K\subseteq F'$. 
The desired forest $L$ can be constructed by applying the induction hypothesis
to $K,F',Y'$.

\noindent
{\em Case~2:} $O\cap K\ne\emptyset$. 
Set 
\[
O'=O\setminus K \quad \mbox{and} \quad K'=K\setminus O
\]
Apply the induction hypothesis to $K',F',Y'$ to construct 
a forest $L'\in \lift(K')$ such that 
$L'\cap \calE(T_e)\subseteq Y'$ and
\be
\label{inductionL'}
|L'|\le 2|L'\cap \calE(T_e)|.
\ee
Define
\be
L=L'\oplus O'.
\ee
Then
\be
\partial L = \partial L' \oplus \partial O' = \partial K' + \partial O' =
\partial (K'\oplus O') = \partial (K\oplus O) = \partial K \oplus \partial O = \partial K.
\ee
Below we show that 
\be
\label{LL}
|O|\le 2|O\cap \calE(T_e)|.
\ee
Therefore
\[
|L|\le |L'| + |O'| \le |L'|+|O| \le 2|L'\cap \calE(T_e)| + 2|O\cap \calE(T_e)|  
=2|L'\cap \calE(T_e)| + 2|O'\cap \calE(T_e)|
=
2|L\cap \calE(T_e)|.
\]
Here the third inequality uses Eqs.~\eqref{inductionL'},\eqref{LL}.
The last equality uses the assumption $L'\cap \calE(T_e)\subseteq Y'$ 
and the identity $Y'\cap O'=\emptyset$.

It remains to prove Eq.~\eqref{LL}. Let $O_{sc}=O\cap \calE(T_{sc})$ 
and $O_e = O\cap \calE(T_e)$. We have $\partial O_{sc}=\partial O_e$,
where the boundary is taken in the graph $\gluedgraph$. Since $F$ is a minimum matching
in the graph $T_{sc}$
and $O_{sc}\subseteq F$, we infer (from Lemmas~\ref{prop:graph3},\ref{prop:graph5}) that
$O_{sc}$ is a minimum matching in the graph  $\gluedgraph$.
Finally, $\partial O_{sc}=\partial O_e$ implies that $|O_{sc}|\le |O_e|$ and thus
$|O|\le 2|O_e|$ which is equivalent to Eq.~\eqref{LL}.
\end{proof}

The following establishes property~\eqref{prop:tricky} of Definition~\ref{def:lifting} for the $Z$-part of the repair function, that is, the function~$\repair_Z$. 
\begin{lemma}
\label{lemma:tricky1}
Consider a Pauli error $E$ acting on $\calA$ and a subset
$K\subseteq \repair_Z{(E)}$. There exists a set $L\in \lift(K)$
such that 
\be
\label{L14}
|L\cap \mathrm{Supp}(E)|\ge \frac14 |L|.
\ee
\end{lemma}
\begin{proof}
Let $M\subseteq \calE(T_e)\cup \calE(T_o)$ be a minimum weight
Pauli-$X$ error such that $\syn_0(M)=\syn_0(E)$. 
Below we shall identify Pauli errors and their supports.  
Set 
\[
Y=(E\oplus M)\cap \calE(T_e).
\]
By construction,  $M\cap \calE(T_e)$ and  
$E \cap \calE(T_e)$ have the same boundary in the graph $T_e$.
Thus $Y$ is a cycle in the graph $T_e$.
Let $\partial Y \subseteq \calV(T_{sc})$ be the boundary of $Y$
in the graph $\gluedgraph$. 
By definition of the repair function, $\repair_Z{(E)}$ is a minimum
matching in the graph $T_{sc}$ with the boundary $\partial Y$. Apply Lemma~\ref{prop:graph6} with $F\equiv \repair_Z{(E)}$
and $Y$ defined above
to construct a forest $L\in \lift(K)$ 
satisfying Eqs.~\eqref{LY1},\eqref{LY2}.
Let us check that $L$ obeys Eq.~\eqref{L14}.
Indeed, set
\[
M' = M \oplus (L\cap \calE(T_e)).
\]
By definition of the lift function, $\partial L =\partial K$ where the boundary
is taken in the graph $\gluedgraph$. Thus $L\cap \calE(T_e)$ is a cycle in the graph
$T_e$ and $\partial M'=\partial M$ (in the graph $T_e$).
The minimality of $M$ gives $|M'|\ge |M|$.
Thus at least half of the edges of $L\cap \calE(T_e)$ are not contained in $M$. By Eq.~\eqref{LY1},
$L\cap \calE(T_e) \subseteq Y$, i.e. at least half of the edges of $L\cap \calE(T_e)$ are 
contained in $E$. We get
\[
|L\cap E| \ge |L\cap \calE(T_e) \cap  E|
\ge \frac12 |L\cap \calE(T_e)| \ge \frac14 |L|.
\]
Here the last inequality follows from Eq.~\eqref{LY2}.
\end{proof}

\subsection{Explicit constants: concluding the proof of Theorem~\ref{thm:3D}\label{sec:parameterscomplete}}
We have  established that the function $\repair_Z$ satisfies the lifting property, see Definition~\ref{def:lifting}. Thus it converts a local stochastic error with rate~$p$ to an error with rate~$p^{\Omega(1)}$. In more detail, we can use Eq.~\eqref{repair_error_rate}
of Section~\ref{sec:codes}
to upper bound the error rate of $\repair_Z{(E)}$.
Indeed, the universal constant $c_1,c_2,c_3$ from
Eq.~\eqref{repair_error_rate} can be extracted from
Lemma~\ref{prop:Lift} and Lemma~\ref{lemma:tricky1}.
We get $c_1=96$, $c_2=1/2$, and $c_3=1/4$. 
Substituting this into Eq.~\eqref{repair_error_rate} one gets
\be
\label{repairZrate}
\repair_Z{(E)} \sim \calN(q), \qquad q=96\sqrt{2} p^{1/8}.
\ee

So far we have ignored the generators $S_1^X,S_1^Z$ 
corresponding to the logical Bell state stabilizers $\overline{X}_1\, \overline{X}_2$ and $\overline{Z}_1\, \overline{Z}_2$, 
see Eqs.~\eqref{S1X},~\eqref{S1Z}. Consider $S_1^X$ first. 
We claim that  the repair operator  $\repair_Z{(E)}$ 
satisfies the syndrome condition 
\begin{align}
\syn_{S_1^X}(\repair_Z{(E)}) = \syn_{S_1^X}(E) \oplus \syn_{S_1^X}(M)\ .\label{eq:syndromesonex}
\end{align}
with probability exponentially close to one. Here $\syn_{S_1^X}(F)\in \{0,1\}$ denotes the syndrome bit
of the logical Bell state stabilizer~$S_1^X$ for a Pauli error $F$. Combined with Eq.~\eqref{repairZ}, this implies that -- except with exponentially small probability -- the repair operator~$\repair_Z{(E)}$ 
obeys the part of syndrome  condition~\eqref{eq:repairconditionsyne} associated with
 all stabilizer generators of~$\calS_1$ defined by $Z$-type stabilizers of the encoded Bell state. In particular, defining $\repair_{\overline{Z}}(E)$ appropriately (see Eq.~\eqref{eq:repairzbare} below) and arguing analogously about $X$-type stabilizers ensures that the product $\repair{(E)}$ (cf.~\eqref{repairXZ}) satisfies  the syndrome condition~\eqref{eq:repairconditionsyne} with certainty.

To prove that~\eqref{eq:syndromesonex} is satisfied with probability exponentially close to one, let $\mathrm{FAIL}$ be the set of errors $E$
such that $\repair_Z{(E)}$ and $E\cdot M$ have different syndromes
for the generator $S_1^X$. Using the explicit form of $S_1^X$,
see Eq.~\eqref{S1X}, one gets
\[
\mathrm{FAIL}=
\{E\, : \,\parity(\repair_Z{(E)}\oplus E_e\oplus M_e, \Omega)=1\},
\]
where $E_e\equiv \mathrm{Supp}(E)\cap \calE(T_e)$,
$M_e \equiv \mathrm{Supp}(M)\cap \calE(T_e)$, and 
\[
\Omega = \{ (1,e,1),(1,e,r)\in \calB\} \cup \{ (1,e,e) \in \calE(T_e)\}.
\]
Note that $\Omega$ includes all dangling edges of the graph $\gluedgraph$
located on the face $(1,e,e)$.
By construction, 
\[
C\equiv \repair_Z{(E)}\oplus E_e\oplus M_e
\]
is a cycle in the graph $\gluedgraph$. 
Thus the event $\mathrm{FAIL}$ happens
iff $C$ contains at least one ``homologically non-trivial" path that starts at the face
 $(1,e,e)$ and ends at the face $(r,e,e)$.
Let us fix such a path~$H$ for each error $E\in \mathrm{FAIL}$.
Let $\mathrm{Pr}[H]$ be the combined probability of all errors
$E\in \mathrm{FAIL}$ that give rise to a given path $H$.
Denote
\[
H_e= H\cap \calE(T_e)
 \qquad H_{sc}= H\cap \calE(T_{sc}).
\]
Note that $H_e$ 
is a cycle in the graph $T_e$. 
The minimality of $M$ implies that  $|M_e\oplus H_e|\ge |M_e|$. Since $H_e$ is contained
in $C\cap \calE(T_e)=M_e\oplus E_e$, we infer that 
 $H_e$  has at least half of the edges in the error $E_e$.
Thus 
\be
\label{PHe}
\mathrm{Pr}[H] \le \sum_{k=|H_e|/2}^{|H_e|} {|H_e| \choose k}p^k \le (2p^{1/2})^{|H_e|}.
\ee 
We have already shown that $\repair_Z{(E)} \sim \calN(q)$,
see Eq.~\eqref{repairZrate}.
Since $H_{sc}\subseteq \repair_Z{(E)}$, one has
\be
\label{PHsc}
\mathrm{Pr}[H] \le \sum_{k=|H_{sc}|/2}^{|H_{sc}|} {|H_{sc}| \choose k}  q^k \le (2q^{1/2})^{|H_{sc}|}.
\ee
For sufficiently small $p$ one has $p\ll q$
and thus 
\be
\label{Pfail1}
\mathrm{Pr}[\mathrm{FAIL}]\le \sum_H\mathrm{Pr}[H]
\le \sum_H (2q^{1/2})^{|H|/2}.
\ee
where the sum runs over all paths $H$ in the graph $\gluedgraph$ connecting
the face $(1,e,e)$ and the face $(r,e,e)$.
Note that such path must have length $l\equiv |H|\ge r$.
The number of length-$l$ paths in the graph $\gluedgraph$ that start at the face $(1,e,e)$  is at most $r^2 6^l$
since $\gluedgraph$ has vertex degree at most $6$. Thus
\[
\mathrm{Pr}[\mathrm{FAIL}]
 \le r^2 \sum_{l=r}^\infty 6^l (2q^{1/2})^{l/2}
\le 2 r^2 (12q^{1/4})^r \le  (24q^{1/4})^r
\]
provided that $24q^{1/4}\le 1$ and $r\ge 7$. This completes the proof of the claim that~\eqref{eq:syndromesonex} is satisfied with probability exponentially close to one.

Now define the random error 
\begin{align}
\repair_{\overline{Z}}{(E)} = \left\{ \begin{array}{rcl}
(\overline{Z}_1)_\calB &\mbox{if} & E\in \mathrm{FAIL} \\
I &\mbox{if} & \mbox{otherwise}.
\\
\end{array}
\right. \label{eq:repairzbare}
\end{align}
Here the logical operator $\overline{Z}_1$ acts on the first surface code. Clearly, this definition ensures that the
operator $\repair{(E)}$ defined in Eq.~\eqref{repairXZ} satisfies~\eqref{eq:syndromesonex}, and thus the part of the syndrome condition~\eqref{eq:repairconditionsyne} associated with all $X$-type stabilizers of the encoded Bell state.

To show that~$\repair{(E)}$ is a local stochastic error with rate as given in Theorem~\ref{thm:3D}, 
recall that we have shown in~\eqref{repairZrate} that the factor $\repair_{Z}{(E)}$ in its definition is a local stochastic error with rate~$q$. We claim that $\repair_{\overline{Z}}{(E)}$ satisfies
\be
\repair_{\overline{Z}}{(E)} \sim \calN(q_0),
\qquad
q_0=600q^{1/2}.
\ee
Indeed, consider some fixed subset $K\subseteq \calB$
and suppose that $K\subseteq \repair_{\overline{Z}}{(E)}$.
Then $|K|\le d$ since $\overline{Z}_1$ has weight $d$.
Recall that $r=2d-1$. Thus
\[
\mathrm{Pr}_E[K\subseteq  \repair_{\overline{Z}}{(E)} ] \le \mathrm{Pr}[\mathrm{FAIL}]
\le (24q^{1/4})^r \le (600 q^{1/2})^d = q_0^d \le q_0^{|K|}.
\]
Exactly the same arguments (with the graphs $T_e,T_{sc}$ replaced
by $T_o,T_{sc}^*$) show that
$\repair_X{(E)}\sim \calN(q)$ and 
$\repair_{\overline{X}}{(E)}\sim \calN(q_0)$.
Finally, using part~(iii) of Lemma~\ref{claim:basic} we conclude
that the full repair operator $\repair{(E)}$ defined in Eq.~\eqref{repairXZ}
obeys 
\[
\repair{(E)} \sim \calN(26 p^{1/64}).
\]
This completes the proof of Theorem~\ref{thm:3D}.

\section{Fault-tolerant quantum advantage on a 3D grid}
\label{sec:3D}

Here we consider Algorithm~1 specialized to the 1D Magic Square
Problem and encode each qubit using the surface code.
 We show how to implement this algorithm by a constant-depth 
quantum circuit that uses only nearest-neighbor gates on a 3D grid 
with $O(1)$ qubits per site. For simplicity, we allow classical control to be geometrically non-local. 
At the end of this section we will discuss how one can modify the relation problem to remove this assumption.

Recall that the ideal quantum circuit solving the 
1D Magic Square Problem operates on a register of $4n$ qubits
labeled as $p_1,\ldots,p_{2n}$ and $q_1,\ldots,q_{2n}$,
see Fig.~\ref{fig:1d}. The circuit consists of the following operations :
\begin{enumerate}[(i)]
\item Initializing a pair of qubits $(p_{2i-1},p_{2i})$ or $(q_{2i-1},q_{2i})$ in the Bell state $\Phi$.\label{it:stepone}
\item Applying CNOT, SWAP to a pair of qubits $(p_{2i},p_{2i+1})$  or $(q_{2i},q_{2i+1})$ or $(p_j,q_j)$.\label{it:steptwo}
\item Applying a single-qubit Clifford gate $H$, $Z$ or $S$.\label{it:stepthree}
\item Measuring a qubit in the $Z$-basis.\label{it:stepfour}
\end{enumerate}
Here the operations~(\ref{it:steptwo},\ref{it:stepthree}) are classically controlled by the input bits 
specifying an instance of the problem.

\begin{figure}[h]
\centerline{\includegraphics[height=4cm]{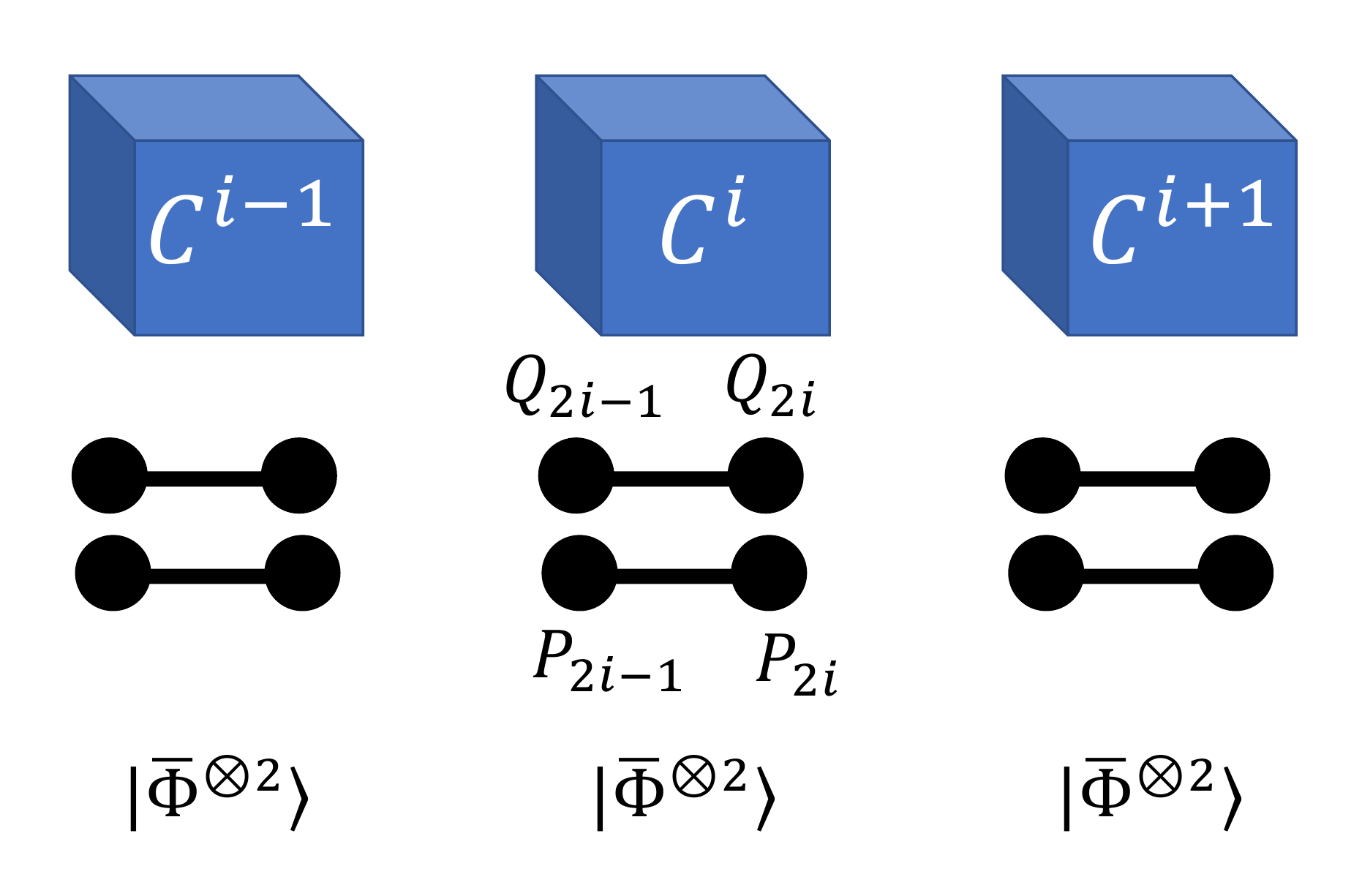}}
\caption{A chain of 3D cubes $\calC_i$.
Each cube is a copy of the 3D lattice $\calC$ shown on Fig.~\ref{fig:surface_code3D}.
A pair of surface codes is attached to the left and to the right face
of each cube $\calC_i$.
\label{fig:cubes3D}}
\end{figure}

We shall encode each qubit $p_i$ and $q_i$ by the distance-$d$ surface code
denoted $P_i$ and $Q_i$ respectively. 
Each surface code is attached to a face of a 3D cubic lattice 
$\calC$ of linear size $r=2d-1$ shown on Fig.~\ref{fig:surface_code3D}
(see Section~\ref{sec:3Dcluster} for a formal definition of $\calC$).
Let $\calC_1,\ldots,\calC_n$ be $n$ copies of the lattice $\calC$.
For brevity, we shall refer to each lattice $\calC_i$ as a {\em cube}.
Each site of $\calC_i$ holds $O(1)$ physical qubits (we shall need at most four qubits per site).
Surface codes $P_{2i-1}$ and $Q_{2i-1}$ are attached to the left face
of the cube $\calC_i$ such that the two codes share the same 
subset of sites on the face of $\calC_i$.
 Likewise, we attach
surface codes $P_{2i}$ and $Q_{2i}$ to the right face
of the cube $\calC_i$. We arrange the cubes $\calC_1,\ldots,\calC_n$ into a one-dimensional chain 
such that the right face of $\calC_i$ is next to the left face of $\calC_{i+1}$, see  Fig.~\ref{fig:cubes3D}.

 A pair of logical Bell states $\overline{\Phi}$
shared between the codes $P_{2i-1}$, $P_{2i}$
and  between the codes $Q_{2i-1}$, $Q_{2i}$
can now be created in a single-shot
fashion by a depth-$6$ Clifford circuit 
operating on the cube $\calC_i$
with nearest neighbor two-qubit gates,
see Theorem~\ref{thm:3D}. This
provides a robust (logical) realization of the initialization operation~\eqref{it:stepone}.
\begin{figure}[h]
\centerline{\includegraphics[height=4cm]{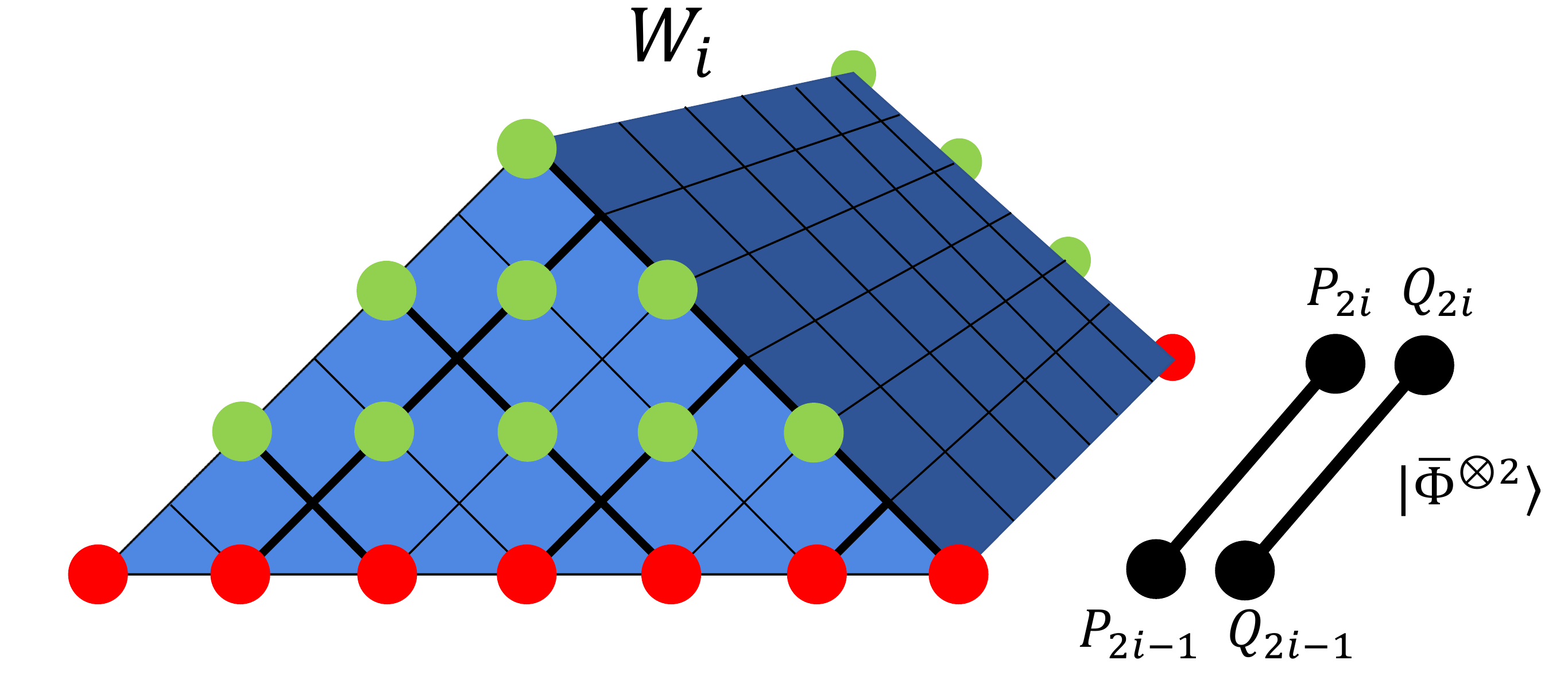}}
\caption{A wedge $\calW_i$ obtained from a cube $\calC_i$
by folding it against the diagonal plane. 
A pair of folded surface codes 
is attached to the left and to the right faces of $\calW_i$.
\label{fig:wedge}}
\end{figure}

Recall that the surface code enables  transversal logical CNOT and SWAP gates. 
In the one-dimensional chain shown in Fig.~\ref{fig:cubes3D},   a pair of the surface codes $P_{2i},Q_{2i}$ is located next to 
$P_{2i+1},Q_{2i+1}$. 
Furthermore, the codes $P_j$ and $Q_j$ share the same subset of sites.
Thus a logical CNOT (respectively SWAP) can be applied to 
pairs of logical qubits $(P_{2i},P_{2i+1})$, $(Q_{2i},Q_{2i+1})$,
and $(P_j,Q_j)$  by a depth-$1$ quantum circuit composed
of geometrically local physical CNOT (respectively SWAP) gates.  This provides the necessary encoded two-qubit operations~\eqref{it:steptwo}.

The surface codes also permits a transversal implementation of the logical~$\overline{Z}$ gate. 
To make the  logical gates $\overline{H}$, $\overline{S}$ geometrically local we shall
replace each surface code by its folded version defined in Section~\ref{sec:codes}.
Accordingly, each cube $\calC_i$ is replaced by a {\em wedge} $\calW_i$
in which pairs of sites  obtained from each other
by a reflection against the diagonal plane are identified. An example
of the wedge $\calW_i$ and the folded surface codes attached to its faces
are shown  on Fig.~\ref{fig:wedge}. 
Since any reflection is a distance-preserving operation, 
all geometrically local gates used at steps~(\ref{it:stepone},\ref{it:steptwo}) remain
geometrically local after mapping cubes $\calC_i$ to wedges $\calW_i$.
As shown in Section~\ref{sec:codes}, 
the folded surface code enables implementation of 
logical $\overline{H}$, $\overline{Z}$ and $\overline{S}$ gates by a depth-$1$ Clifford circuit
composed of geometrically local gates. This completes the description of the fault-tolerant realization of the operations~\eqref{it:stepthree}.

Finally, we recall that a logical $Z$-measurement  can be realized fault-tolerantly in the surface code  by measuring each qubit in the~$Z$-basis and decoding the result as discussed in Section~\ref{sec:singleshotlogicalmeas}, see Theorem~\ref{thm:singleshotlogicalmeasurementsurface}. We conclude that fault-tolerant analogues of all operations~(\ref{it:stepone}--\ref{it:stepfour}) can be implemented by geometrically local constant-depth Clifford circuits on a 3D grid of qubits. Thus we can implement Algorithm~1 and solve the noise-tolerant version of the 1D Magic Square relation using a constant-depth quantum circuit with geometrically local gates in 3~dimensions.

Let us now briefly sketch how one can also make the classical control geometrically local, if desired. Note that every input bit to the 1D Magic Square Problem only acts as a control in $O(1)$ Clifford gates in the  ideal quantum circuit~$U$ which solves it. We may then imagine prepending a classical copying circuit~$C_{copy}$ to the quantum circuit. The circuit~$C_{copy}$ simply creates a local copy of each input bit next to every gate location where it is used as a control in the fault-tolerant circuit of Fig.~\ref{fig:algcircuit}. Then we can write down an {\em extended} fault-tolerant quantum circuit~$U^{\textrm{ext}}$ which accesses these copies of input bits and which only involves locally controlled gates. 

Matching the definition of the extended quantum circuit~$U^{\textrm{ext}}$, we may define an {\em extended fault-tolerant relation}~$\mathcal{R}^{\textrm{ext}}_{U}$. The outputs of this relations are identical to that of~$\mathcal{R}_{U}$, but the input is  modified as there are now additional input bits. Suppose a subset~$S\subset \{0,1\}^v$ of problem instances (inputs) can be used to show a quantum advantage for the relation~$\mathcal{R}_U$. We claim that the subset $C_{copy}(S)$ of inputs for~$\mathcal{R}^{\textrm{ext}}_{U}$ can be used to show an advantage for the extended relation: Clearly, the input/output pairs of the extended quantum  circuit~$U^{\textrm{ext}}$, for any input belonging to~$C_{copy}(\{0,1\}^v)$, satisfy the relation~$\mathcal{R}^{\textrm{ext}}_U$ with high probability. To show that the extended relation ``remains hard'' for classical circuits, observe that  by assumption and because every code block has size~$m$,
the copying circuit~$C_{copy}$ can be realized by a depth-$1$ circuit using $O(m)$-local gates. Thus any classical circuit~$C^{\textrm{ext}}$ for the extended relation~$\mathcal{R}^{\textrm{ext}}_U$ can be modified to act as a classical circuit for~$\mathcal{R}_U$ by prepending~$C^{\textrm{ext}}$, increasing the circuit depth only by a constant without changing the fan-in beyond the restriction given in Theorem~\ref{thm:lowerboundcircuitdepthDMS}. This implies the claim.

\section{Acknowledgements}
SB acknowledges support from the IBM Research
Frontiers Institute and funding from the MIT-IBM Watson AI Lab 
under the project Machine Learning in Hilbert space. RK acknowledges support by the Technical University of Munich --  Institute of Advanced Study, funded by the German Excellence Initiative and the European Union Seventh Framework Programme under grant agreement no. 291763, by the DFG cluster of excellence  2111 (Munich Center for Quantum Science and Technology),  and by the German Federal Ministry of Education through the funding program Photonics Research Germany, contract no. 13N14776 (QCDA-QuantERA).


\bibliographystyle{unsrt}
\bibliography{library,library_MT}

\begin{thebibliography}{10}

\bibitem{bremner2016average}
Michael~J Bremner, Ashley Montanaro, and Dan~J Shepherd.
\newblock Average-case complexity versus approximate simulation of commuting
  quantum computations.
\newblock {\em Physical Review Letters}, 117(8):080501, 2016.

\bibitem{bremner2016achieving}
Michael~J. Bremner, Ashley Montanaro, and Dan~J. Shepherd.
\newblock Achieving quantum supremacy with sparse and noisy commuting quantum
  computations.
\newblock {\em {Quantum}}, 1:8, April 2017.

\bibitem{farhi2016quantum}
Edward Farhi and Aram~W Harrow.
\newblock Quantum supremacy through the quantum approximate optimization
  algorithm.
\newblock {\em arXiv preprint arXiv:1602.07674}, 2016.

\bibitem{bermejo2017architectures}
Juan Bermejo-Vega, Dominik Hangleiter, Martin Schwarz, Robert Raussendorf, and
  Jens Eisert.
\newblock Architectures for quantum simulation showing a quantum speedup.
\newblock {\em Phys. Rev. X}, 8:021010, Apr 2018.

\bibitem{terhal2002adaptive}
Barbara~M Terhal and David~P DiVincenzo.
\newblock Adaptive quantum computation, constant depth quantum circuits and
  {A}rthur-{M}erlin games.
\newblock {\em Quant. Inf. Comp.}, 4(2):134--145, 2004.

\bibitem{aaronson2005quantum}
Scott Aaronson.
\newblock Quantum computing, postselection, and probabilistic polynomial-time.
\newblock In {\em Proceedings of the Royal Society of London A: Mathematical,
  Physical and Engineering Sciences}, volume 461, 2063, pages 3473--3482. The
  Royal Society, 2005.

\bibitem{boixo2018characterizing}
Sergio Boixo, Sergei~V Isakov, Vadim~N Smelyanskiy, Ryan Babbush, Nan Ding,
  Zhang Jiang, Michael~J Bremner, John~M Martinis, and Hartmut Neven.
\newblock Characterizing quantum supremacy in near-term devices.
\newblock {\em Nature Physics}, 14(6):595, 2018.

\bibitem{pednault2017breaking}
Edwin Pednault, John~A Gunnels, Giacomo Nannicini, Lior Horesh, Thomas
  Magerlein, Edgar Solomonik, and Robert Wisnieff.
\newblock Breaking the 49-qubit barrier in the simulation of quantum circuits.
\newblock {\em arXiv preprint arXiv:1710.05867}, 2017.

\bibitem{boixo2017simulation}
Sergio Boixo, Sergei~V Isakov, Vadim~N Smelyanskiy, and Hartmut Neven.
\newblock Simulation of low-depth quantum circuits as complex undirected
  graphical models.
\newblock {\em arXiv preprint arXiv:1712.05384}, 2017.

\bibitem{li2018quantum}
Riling Li, Bujiao Wu, Mingsheng Ying, Xiaoming Sun, and Guangwen Yang.
\newblock Quantum supremacy circuit simulation on {S}unway {T}aihu{L}ight.
\newblock {\em arXiv preprint arXiv:1804.04797}, 2018.

\bibitem{chen2018classical}
Jianxin Chen, Fang Zhang, Mingcheng Chen, Cupjin Huang, Michael Newman, and
  Yaoyun Shi.
\newblock Classical simulation of intermediate-size quantum circuits.
\newblock {\em arXiv preprint arXiv:1805.01450}, 2018.

\bibitem{bravyi2018simulation}
Sergey Bravyi, Dan Browne, Padraic Calpin, Earl Campbell, David Gosset, and
  Mark Howard.
\newblock Simulation of quantum circuits by low-rank stabilizer decompositions.
\newblock {\em arXiv preprint arXiv:1808.00128}, 2018.

\bibitem{bragokoe18}
Sergey Bravyi, David Gosset, and Robert Koenig.
\newblock Quantum advantage with shallow circuits.
\newblock {\em Science}, 362(6412):308--311, Oct 2018.

\bibitem{coudronetal18}
Matthew Coudron, Jalex Stark, and Thomas Vidick.
\newblock Trading locality for time: certifiable randomness from low-depth
  circuits.
\newblock {\em arXiv preprint arXiv:1810.04233}, October 2018.

\bibitem{legall18}
François~Le Gall.
\newblock Average-case quantum advantage with shallow circuits.
\newblock {\em arXiv preprint arXiv:1810.12792}, October 2018.

\bibitem{Watts2019advantage}
Adam~Bene Watts, Robin Kothari, Luke Schaeffer, and Avishay Tal.
\newblock Exponential separation between shallow quantum circuits and unbounded
  fan-in shallow classical circuits.
\newblock Presented at the 22nd Annual Conference on Quantum Information
  Processing (QIP), January 2019.

\bibitem{peres90}
Asher Peres.
\newblock {Incompatible results of quantum measurements}.
\newblock {\em Physics Letters A}, 151(3-4):107--108, dec 1990.

\bibitem{mermin90}
N.~David Mermin.
\newblock {Simple unified form for the major no-hidden-variables theorems}.
\newblock {\em Physical Review Letters}, 65(27):3373--3376, dec 1990.

\bibitem{fawzi2018constant}
Omar Fawzi, Antoine Grospellier, and Anthony Leverrier.
\newblock Constant overhead quantum fault-tolerance with quantum expander
  codes.
\newblock In {\em 2018 IEEE 59th Annual Symposium on Foundations of Computer
  Science (FOCS)}, pages 743--754. IEEE, 2018.

\bibitem{bravyi2006lieb}
Sergey Bravyi, Matthew~B. Hastings, and Frank Verstraete.
\newblock Lieb-{R}obinson bounds and the generation of correlations and
  topological quantum order.
\newblock {\em Physical Review Letters}, 97(5):050401, 2006.

\bibitem{eldar2017local}
Lior Eldar and Aram~W Harrow.
\newblock Local {H}amiltonians whose ground states are hard to approximate.
\newblock In {\em 2017 IEEE 58th Annual Symposium on Foundations of Computer
  Science (FOCS)}, pages 427--438. IEEE, 2017.

\bibitem{aharonov2018quantum}
Dorit Aharonov and Yonathan Touati.
\newblock Quantum circuit depth lower bounds for homological codes.
\newblock {\em arXiv preprint arXiv:1810.03912}, 2018.

\bibitem{calderbank1996good}
A~Robert Calderbank and Peter~W Shor.
\newblock Good quantum error-correcting codes exist.
\newblock {\em Physical Review A}, 54(2):1098, 1996.

\bibitem{steane1996multiple}
Andrew Steane.
\newblock Multiple-particle interference and quantum error correction.
\newblock {\em Proceedings of the Royal Society of London. Series A:
  Mathematical, Physical and Engineering Sciences}, 452(1954):2551--2577, 1996.

\bibitem{bombin2015single}
H{\'e}ctor Bomb{\'\i}n.
\newblock Single-shot fault-tolerant quantum error correction.
\newblock {\em Physical Review X}, 5(3):031043, 2015.

\bibitem{moussa2016transversal}
Jonathan~E Moussa.
\newblock Transversal {C}lifford gates on folded surface codes.
\newblock {\em Physical Review A}, 94(4):042316, 2016.

\bibitem{raussendorfetal05}
Robert Raussendorf, Sergey Bravyi, and Jim Harrington.
\newblock Long-range quantum entanglement in noisy cluster states.
\newblock {\em Phys. Rev. A}, 71:062313, Jun 2005.

\bibitem{knill312190scalable}
E.~Knill.
\newblock Scalable quantum computing in the presence of large detected-error
  rates.
\newblock {\em Phys. Rev. A}, 71:042322, Apr 2005.

\bibitem{GottesmanChuangNature}
Daniel Gottesman and Isaac~L. Chuang.
\newblock Demonstrating the viability of universal quantum computation using
  teleportation and single-qubit operations.
\newblock {\em Nature}, 402:390--393, 1999.

\bibitem{Gottesman2013}
Daniel {Gottesman}.
\newblock {Fault-Tolerant Quantum Computation with Constant Overhead}.
\newblock {\em arXiv e-prints}, page arXiv:1310.2984, Oct 2013.

\bibitem{AliferisGottesmanPreskill2006}
Panos Aliferis, Daniel Gottesman, and John Preskill.
\newblock Quantum accuracy threshold for concatenated distance-3 codes.
\newblock {\em Quantum Info. Comput.}, 6(2):97--165, March 2006.

\bibitem{aliferisgottesmanpreskill2008}
Panos Aliferis, Daniel Gottesman, and John Preskill.
\newblock Accuracy threshold for postselected quantum computation.
\newblock {\em Quantum Info. Comput.}, 8(3):181--244, March 2008.

\bibitem{bravyi1998quantum}
Sergey Bravyi and Alexey Kitaev.
\newblock Quantum codes on a lattice with boundary.
\newblock {\em arXiv preprint quant-ph/9811052}, 1998.

\bibitem{raussendorf2007fault}
Robert Raussendorf and Jim Harrington.
\newblock Fault-tolerant quantum computation with high threshold in two
  dimensions.
\newblock {\em Physical Review Letters}, 98(19):190504, 2007.

\bibitem{dennis2002topological}
Eric Dennis, Alexei Kitaev, Andrew Landahl, and John Preskill.
\newblock Topological quantum memory.
\newblock {\em Journal of Mathematical Physics}, 43(9):4452--4505, 2002.

\bibitem{Fowler2009}
Austin~G. Fowler, Ashley~M. Stephens, and Peter Groszkowski.
\newblock High-threshold universal quantum computation on the surface code.
\newblock {\em Phys. Rev. A}, 80:052312, Nov 2009.

\bibitem{fowler2012proof}
Austin~G Fowler.
\newblock Proof of finite surface code threshold for matching.
\newblock {\em Physical Review Letters}, 109(18):180502, 2012.

\bibitem{raussendorf2005long}
Robert Raussendorf, Sergey Bravyi, and Jim Harrington.
\newblock Long-range quantum entanglement in noisy cluster states.
\newblock {\em Physical Review A}, 71(6):062313, 2005.

\end{thebibliography}

\appendix

\end{document}